\documentclass[11pt]{article}

\usepackage{array}
\usepackage[T1]{fontenc}
\usepackage{graphicx}
\usepackage{amsfonts}
\usepackage{amsmath}
\usepackage{mathabx}   
\usepackage{amsfonts}
\usepackage{amsthm}
\usepackage{amssymb}
\usepackage{amscd}
\usepackage{epsfig}
\usepackage{verbatim}
\usepackage{fancybox}
\usepackage{moreverb}
\usepackage{psfrag}
\usepackage{cellspace}
\usepackage{hyperref}
\usepackage{latexsym}
\usepackage{psfrag}
\usepackage{color}

\usepackage{tikz}
\usetikzlibrary{arrows,positioning} 
\tikzset{
    >=stealth',
    punkt/.style={
           rectangle,
           rounded corners,
           draw=black, very thick,
           text width=4.9cm,
           minimum height=1cm,
           text centered},
    pil/.style={
           ->,
           thick,
           shorten <=2pt,
           shorten >=2pt,}
}
\tikzset{
    >=stealth',
    punktt/.style={
           rectangle,
           rounded corners,
           draw=black, very thick,
           text width=1cm,
           minimum height=1cm,
           text centered},
    pill/.style={
           ->,
           thick,
           shorten <=2pt,
           shorten >=2pt,}
}

\numberwithin{equation}{section}
\definecolor{blanc}{rgb}{1.,1.,1.}
\definecolor{vert}{rgb}{0.,0.5,0.}
\definecolor{rouge}{rgb}{0.8,0.,0.}
\definecolor{violet}{rgb}{0.5,0.,0.4}
\definecolor{bleu}{rgb}{0.,0.,0.5}
\definecolor{orange}{rgb}{0.8,0.4,0.}
\definecolor{light-blue}{rgb}{0.5,0.5,0.7}
\definecolor{light-red}{rgb}{0.8,0.4,0.4}
\definecolor{noir}{rgb}{0.,0.,0.}
\definecolor{gris}{rgb}{0.8,0.8,0.7}

\newcommand{\eps}{\varepsilon}
\newcommand{\avec}{\mathbf a}

\newcommand{\pvec}{\mathbf p} 
\newcommand{\Pvec}{\mathbf P} 
 
\newcommand{\qvec}{\mathbf q} 
\newcommand{\Qvec}{\mathbf Q} 
\newcommand{\rvec}{\mathbf r} 
\newcommand{\Rvec}{\mathbf R}

\newcommand{\xvec}{\mathbf x} 
\newcommand{\Xvec}{\mathbf X}

\newcommand{\vvec}{\mathbf v} 
\newcommand{\Vvec}{\mathbf V} 
\newcommand{\yvec}{\mathbf y} 
 
\newcommand{\ytrf}{\text{\small$\boldsymbol\Upsilon$}} 
 
\newcommand{\zvec}{\mathbf z}

\newcommand{\Ccal}{{\cal C}}

\newcommand{\Pcal}{{\cal P}}

\newcommand{\Scal}{{\cal S}}

\newcommand{\PerOneD}{{\text{\rm per}}}
\newcommand{\PerND}{\#}

\newcommand{\Rest}{{\text{\Large$\boldsymbol{\rho}$}}}
\newcommand{\rest}{{\text{\Large$\rho$}}}
\newcommand{\ball}{{\text{\large$\mathfrak b$}}}

\newcommand{\UsCoo}{\grave}
\newcommand{\FctCoordPuiss}[2]{\text{\Large $\mathbf{r}$}^{\mathbf{#1}}_{#2}}

\newcommand{\rit}{\mathbb R} 
 
\newcommand{\Nabla}{\nabla\hspace{-2pt}}

\newcommand{\ds}{\displaystyle}
\newcommand{\fracp}[2]{\frac{\partial #1}{\partial #2}}
\newcommand{\poibrack}[2]{{\left\{{#1},{#2}\right\}}}

\newtheorem{theorem}{Theorem}[section]
\newtheorem{lemma}[theorem]{Lemma}

\newtheorem{property}[theorem]{Property}
\newtheorem{corollary}[theorem]{Corollary}
\newtheorem{definition}[theorem]{Definition}

\newtheorem{remark}[theorem]{Remark}

\newtheorem{algorithm}[theorem]{Algorithm}

\begin{document}

\title{On the Geometrical Gyro-Kinetic Theory}
\author{
Emmanuel Fr\'enod\thanks{Universite de Bretagne-Sud,  UMR 6205, LMBA, F-56000 Vannes, France \& Inria Nancy-Grand Est, Tonus Project.}
\and 
Mathieu Lutz\thanks{Universit\'{e} de Strasbourg, IRMA,  7 rue Ren\'e Descartes, F-67084 Strasbourg Cedex, France \& Projet INRIA Calvi.}
}
\date{}
\maketitle
{\small {\bf Abstract - }
Considering a Hamiltonian Dynamical System describing the motion of charged particle in a Tokamak or a Stellarator, we build a change of coordinates
to reduce its dimension. This change of coordinates is in fact an intricate succession of mappings that are built using Hyperbolic Partial Differential Equations,
Differential Geometry, Hamiltonian Dynamical System Theory and Symplectic Geometry, Lie Transforms and a new tool which is here introduced : Partial Lie Sums. 
}\\

{\small {\bf Keywords - }
Tokamak; Stellarator;  Gyro-Kinetic Approximation; Hyperbolic Partial Differential Equations; Differential Geometry; Hamiltonian Dynamical System Theory; 
Symplectic Geometry; Lie Transforms; Partial Lie Sums.
}

\section*{Notations}

\begin{enumerate}


\item 
\label{201402262154}  
For a $2\pi$-periodic set $I^{\#}$ included in $\mathbb{R},$
$\mathcal{C}^\infty_\PerOneD\!\left(I^{\#}\right)$ stands of the space of functions being in $\mathcal{C}^\infty(I^{\#})$ and $2\pi$-periodic. 

\item 
\label{201402252149}  
For a set $\boldsymbol{\mathfrak{M}}^{\#}$ included in $\mathbb{R}^{m}$ (where $m\in\mathbb{N}$ and $m
\geq2$) which is $2\pi$-periodic with respect to the $l$-th variable ($l\leq m$)
we denote by $\mathcal{C}_{\#,l}^{\infty}\left(\!\boldsymbol{\mathfrak{M}}^{\#}\!\right)$ the space of 
functions being in $\mathcal{C}^\infty(\boldsymbol{\mathfrak{M}}^{\#})$ and $2\pi$-periodic with respect to the l-th variable. 

\item  $\mathcal{C}_{\PerND}^{\infty}\left(\!\boldsymbol{\mathfrak{M}}^{\#}\!\right) =  \mathcal{C}_{\#,(m-1)}^{\infty}\left(\!\boldsymbol{\mathfrak{M}}^{\#}\!\right)$.
\label{201402161039} 

\item For $m\in\mathbb{N}^\star$, $\mathcal{C}_{b}^{\infty}\!\left(\mathbb{R}^{m}\right)$ stands of the space of 
functions being in $\mathcal{C}^\infty(\rit^m)$ and with their derivatives at any order which are bounded.

 \item 
 \label{201402272144} 
 $\mathcal{O}_{T,b}^{\infty}$ stands for the algebra of functions 
spanned by the functions of the form
\begin{align*}
      &\left(\mathbf{y},\theta\right)\mapsto f_{1}\left(\mathbf{y}\right)\cos\left(\theta\right)+f_{2}\left(\mathbf{y}\right)\sin\left(\theta\right),
\end{align*}
where $f_{1},f_{2}\in\mathcal{C}_{b}^{\infty}\!\left(\mathbb{R}^2\right)$.

\item  
\label{201402161042} 
$\mathcal{Q}_{T,b}^{\infty}$ stands for the space of functions

\begin{gather*}
\begin{aligned}
      \mathcal{Q}_{T,b}^{\infty}=\bigg\{&f\in\mathcal{C}^{\infty}\left(\mathbb{R}^{3}\times\left(0,+\infty\right)\right),\ f\left(\mathbf{y},\theta,k\right)=\underset{n\in\mathbb{I}_{f}}{\sum}c_{n}\left(\mathbf{y},\theta\right)\sqrt{k}^{n} 
      \\  
      &\text{ \hspace{3cm} where }\mathbb{I}_{f}\subset\mathbb{Z}\text{ is finite and }\forall n\in\mathbb{I}_{f},\ c_{n}\in\mathcal{O}_{T,b}^{\infty}\bigg\}.
\end{aligned}
\end{gather*}

\item 
\label{201402161044}   
For an open subset $U\subset\rit^p$, we denote by $\mathcal{A}\left(U\right)$ the space of real analytic functions on $U$. 

\item For a formal power series $S$, we denote by $\boldsymbol{\Sigma}_S$ its set of convergence.

\item $\ball^{n}(\boldsymbol{m}_0,R_0)$ stands for the open euclidian ball of radius $R_0$ and of center $\boldsymbol{m}_0$ in $\rit^n$.

\item $\ball^{\#}\!\big(\boldsymbol{m}_0,R_{\boldsymbol{m}_{0}}\big)$ stands for 
\begin{align*}
      &\left\{ \boldsymbol{m}\in\rit^{4}\ s.t.\ (m_{1},m_{2},m_{4})\in\ball^{3}\!\big(\boldsymbol{m}_0,R_{\boldsymbol{m}_{0}}\big)\right\} 
\end{align*}

\item $\boldsymbol{\mathfrak{C}}(a,b)$ stands for the open crown
\label{201402230841}   
\begin{align*}
     &\boldsymbol{\mathfrak{C}}(a,b)=\big\{ \mathbf{v}\in\rit^{2}\ s.t.\ \left|\mathbf{v}\right|\in\left(a,b\right)\big\} 
\end{align*}

 
\item $\mathcal{CO}(\mathbf{m}_0,R_0;a,b)$ stands for the subset of $\mathbb{R}^4$ defined by 
\label{201402160959} 
\begin{align*}
       &\mathcal{CO}(\mathbf{m}_0,R_0;a,b)=\ball^{2}(\mathbf{m}_0,R_0)\times\rit\times\left(a,b\right).
\end{align*}

 

\end{enumerate}


\tableofcontents

\section{Introduction}

  At the end of the 70',
{Littlejohn \cite{littlejohn:1979,littlejohn:1981,littlejohn:1982}}
shed new light on what is called \emph{the Guiding Center Approximation.}
His approach incorporated high level mathematical concepts from
Hamiltonian Mechanics, Differential Geometry and Symplectic Geometry into a physical affordable theory in order 
to clarify what has been done for years in the domain 
(see {Kruskal \cite{Kruskal1965}},
{Gardner \cite{Gardner1959}},
{Northrop  \cite{northrop:1961}},
{Northrop \& Rome \cite{10.1063/1.862226}}).
This theory is a nice success.
It has been beeing widely
used by physicists to deduce related models 
(\emph{Finite Larmor Radius Approximation, Drift-Kinetic Model, Quasi-Neutral Gyro-Kinetic Model, etc.}, see for instance 
{Brizard \cite{10.1063/1.871465}},
{Dubin \emph{et al.} \cite{dubin/etal:1983}},
{Frieman  \& Chen \cite{10.1063/1.863762}},
{Hahm \cite{10.1063/1.866544}},
{Hahm, Lee \& Brizard \cite{10.1063/1.866641}},
{Parra \& Catto \cite{0741-3335-50-6-065014,0741-3335-51-6-065002,0741-3335-52-4-045004}})
 making up the \emph{Gyro-Kinetic Approximation Theory}, which is the basis of
all kinetic codes used to simulate Plasma Turbulence emergence and evolution in Tokamaks and Stellarators 
(see for instance {Brizard \cite{10.1063/1.871465}}, {Quin \emph{et al} \cite{CTPP:CTPP200610034,qin:056110}},
{Kawamura \& Fukuyama \cite{kawamura:042304}}, {Hahm \cite{hahm:4658}}, {Hahm, Wang \& Madsen \cite{hahm:022305}},
{Grandgirard \emph{et al.} \cite{Grandgirard2006395,0741-3335-49-12B-S16}},
and the review of {Garbet \emph{et al.} \cite{0029-5515-50-4-043002}}).
\\
Yet, the resulting Geometrical Gyro-Kinetic Approximation Theory remains a physical theory which is formal from the mathematical point of view and not directly accessible for 
mathematicians.
The present paper is a first step towards providing a mathematical affordable theory, particularly for the analysis, the applied mathematics and computer sciences communities.
\\

  Notice that, beside this Geometrical Gyro-Kinetic Approximation Theory, an alternative approach, based on 
 Asymptotic Analysis and Homogenization Methods was developed by
 Fr\'enod \& Sonnendr\"ucker \cite{frenod/sonnendrucker:1997,frenod/sonnendrucker:1998, frenod/sonnendrucker:1999},
Fr\'enod, Raviart \& Sonnendr\"ucker \cite{FRS:1999},
Golse \& Saint-Raymond \cite{GSR1}
and {Ghendrih}, {Hauray} \& {Nouri} \cite{2010arXiv1004.5226G}.
 \\

The purpose of this paper is to provide a mathematical framework for the formal 
Guiding-Center reduction introduced in {Littlejohn \cite{littlejohn:1979}}.
 The domain of application of this theory is that of a charged particle under the action of a strong magnetic field.
Hence we will consider the following dynamical system :
\begin{align} 
     &\fracp{\Xvec}{t}=\Vvec, & & \Xvec(0) =\xvec_0, 
     \label{systeqs1} 
     \\
     &\fracp{\Vvec}{t}=\frac{1}{\eps}\; B\left(\mathbf{X}\right)^{\;\perp\!}\Vvec, & &\Vvec(0) =\vvec_0,
     \label{systeqs2} 
\end{align}
where $\mathbf{X}=(X_1,X_2)$ stands for the position, $\mathbf{V}=(V_1,V_2)$ stands for the velocity, $\mathbf{V}^\perp=(V_2,-V_1)$,
$\xvec_0$ and $\vvec_0$ stand for the initial position and velocity, 
and $\eps$ is a small parameter.
We notice that equations \eqref{systeqs1}-\eqref{systeqs2} can be obtained from the six dimensional system by taking a magnetic field in the $x_3$-direction
that only depends on $x_1$ and $x_2$.\\
When the magnetic field is constant, the trajectory associated with \eqref{systeqs1}-\eqref{systeqs2} is a circle of center $\mathbf{c}_0=\mathbf{x}_0+\eps\mathbf{v}_0$
and of radius $\eps\left|\mathbf{v}_{0}\right|$.
Otherwise, the dynamical system \eqref{systeqs1}-\eqref{systeqs2} can be viewed as a perturbation of the system obtained when the magnetic field is constant.
Hence, in the general case of a magnetic field depending on position, the evolution of a given particle's position is a combination of two disparate in time motions: a slow evolution of what is the center of the circle in the case when $B$ is constant, usually called the Guiding Center, and a fast rotation with a small radius about it.
The Guiding-Center reduction consists in replacing the trajectory of the particle by the trajectory of a quantity close to the guiding-center and free of fast oscillations.
%

This purpose can easily be translated within a geometric formalism. In any system of coordinates on a manifold $\mathcal{M}$, a Hamiltonian dynamical system whose solution
is $\Rvec=\Rvec(t;\rvec_0)$ can be written in the following form
\begin{gather}
\label{HamSystSh0} 
 \ds\fracp{\Rvec}{t} = \Pcal(\Rvec)\Nabla_{\rvec} H(\Rvec), ~~ \Rvec(0,\rvec_0) = \rvec_0,
\end{gather}
where  $\Pcal(\rvec)$ is a matrix called the matrix of the Poisson Bracket (or Poisson Matrix in short), and $H(\rvec)$ is called the 
Hamiltonian function. The Poisson Matrix is a skew-symmetric matrix satisfying the Jacobi identity and the Hamiltonian function is a smooth function (see Appendix \ref{AppendixChangeOfCoordRulesForHamAndPM}).
It is obvious to show that dynamical system \eqref{systeqs1}-\eqref{systeqs2} is Hamiltonian and to find its related Poisson Matrix $\grave{\mathcal{P}}_{\eps}\!\left(\mathbf{x},\mathbf{v}\right)$ and Hamiltonian function $\grave{H}_{\eps}\!\left(\mathbf{x},\mathbf{v}\right)$ (see Section \ref{PanoramaSubSection}).
Within this geometrical framework, the goal of the Guiding-Center reduction is to make a succession of changes of coordinates in order to
satisfy the assumptions of the following theorem.
\begin{theorem}
\label{KrTh} 
If, in  a given coordinate system $\rvec =(r_1, r_2, r_3, r_4),$ the Poisson Matrix 
has the following form:
\begin{gather} 
\label{HamSystSh1} 
                   \Pcal(\mathbf{r}) = \left(\begin{array}{c|cc}
                    \text{\huge${\mathtt{M}}$}(\mathbf{r}) & \begin{array}{c}0 \vspace{-3pt} \\0 \end{array} & \begin{array}{c}0 \vspace{-3pt}\\0 \end{array} \\
                    \hline
                     0\ 0&0&\mathcal{P}_{3,4} \\
                     0\ 0&-\mathcal{P}_{3,4} & 0 \\
                                   \end{array}\right),
\end{gather}
where $\mathcal{P}_{3,4}$ is a non-zero constant, and if  the  Hamiltonian function does not depend on the penultimate variable, i.e.
\begin{gather}
\label{HamSystSh2} 
\fracp{H}{r_3} =0,
\end{gather}
then, submatrix $\boldsymbol{\mathtt{M}}$ does not depend on the two last variables, i.e. 
\begin{gather}
\label{HamSystSh3} 
   \ds \fracp{\boldsymbol{\mathtt{M}}}{r_3}=0 \text{ and } \fracp{\boldsymbol{\mathtt{M}}}{r_4}=0.
\end{gather}
Consequently, the time-evolution of the two first components $\Rvec_1, \Rvec_2$ is independent
of the penultimate component $\Rvec_3$;
and, the last component $\Rvec_4$ of the trajectory is not time-evolving, i.e.
\begin{gather}
\label{HamSystSh4} 
  \ds \ds \fracp{\Rvec_4}{t}=0.
\end{gather}
\end{theorem}
Theorem  \ref{KrTh} is the Key Result that brings the understanding of the Guiding-Center reduction: the Guiding-Center reduction consists in writing dynamical system \eqref{systeqs1}-\eqref{systeqs2} within a system of coordinates, called the Guiding-Center Coordinate System, that satisfies the assumptions of Theorem \ref{KrTh} and which is close to the 
Historic Guiding-Center Coordinate System, usually defined by:
  \begin{align} 
  \label{UsualGC1NC} 
     &y_{1}^{hgc}=x_{1}-\eps\frac{v}{B(\mathbf{x})}\cos\left(\theta\right),
     \\
     \label{UsualGC2NC} 
     &y_{2}^{hgc}=x_{2}+\eps\frac{v}{B(\mathbf{x})}\sin\left(\theta\right),
     \\
     \label{UsualGC3NC} 
     &\theta^{hgc}=\theta,
     \\
     \label{UsualGC4NC} 
     &k^{hgc}=\frac{v^{2}}{2B(\mathbf{x})},
\end{align}
where $v =|\vvec|$ and where $\theta$ is 
the angle between the $x_1$-axis and the gyro-radius vector 
$\boldsymbol{\rho}_\eps(\mathbf{x},\mathbf{v})=-\frac{\eps}{B(\mathbf{x})}\mathbf{v}^\perp$
measured in a clockwise sense.


Once this done, if we are just interested in the motion of the particle in the physical space,
 i.e.  just in the evolution of the two first components, solving the dynamical system in the new system of coordinates,  reduces to find a trajectory in $\mathbb{R}^2$, in place of a trajectory in $\mathbb{R}^4$ when it is  solved in the original system of coordinates.

In {\cite{littlejohn:1979}}, Littlejohn proposed a construction of the Guiding-Center Coordinates based on formal series expansion in power of $\eps$. 
This approach cannot be made mathematically rigorous because no argument can insure the validity of the series expansion.
\\
In the present paper we adopt a different strategy.
We will derive for each positive integer $N$ a coordinate system, the so-called Guiding-Center Coordinates of order $N$, 
whose expansion in power of $\eps$, up to  any order $N$,
coincides with the Guiding-Center coordinates given in {\cite{littlejohn:1979}}. Moreover, for each 
integer $N$ we will construct a
Hamiltonian dynamical system satisfying Theorem \ref{KrTh} and approximating uniformly in time,
with accuracy in proportion to $\eps^{N-1}$, 
the Hamiltonian dynamical system \eqref{systeqs1}-\eqref{systeqs2} written within the Guiding-Center Coordinates of order $N$.
\\

The Guiding-Center reduction consists essentially in a succession of three change of coordinates:
a polar in velocity change of coordinates $(\mathbf{x},\mathbf{v})\mapsto(\mathbf{x},\theta,v)$ with $\theta$ and $v$ defined above, a second change of coordinates called the Darboux change of coordinates, and a last change of coordinates 
called the Lie change of coordinates.
The objective of the first change of coordinates is to concentrate the fast oscillations on the $\theta$ variable.
The second one, 
consists in finding a coordinate system in which the Poisson Matrix has the required form to apply Theorem \ref{KrTh}, and 
eventually the last change of coordinates (which is in fact the succession of $N$ changes of coordinates) consists in removing the oscillations from the Hamiltonian function while keeping the same expression of the Poisson Matrix.
\\

All along this paper we will assume that the magnetic field $B$ is analytic, that all its derivatives are bounded, and that $B$ is nowhere close to $0$, \textit{i.e.} that $\inf B>1$. 
\\

The three main results of this paper are the following.
The first one concerns the Darboux change of coordinates.

\begin{theorem}
\label{MainThm1}   
There exists a $\mathcal{C}^\infty\!$-diffeomorphism $\boldsymbol{\Upsilon}\!_\eps : \left(\mathbf{x},\theta,v\right)\mapsto\left(\mathbf{y},\theta,k\right)$ one to one from $\mathbb{R}^{2}\times\mathbb{R}\times\left(0,+\infty\right)$
onto itself,  smooth with respect to $\eps$, such that the Poisson Matrix expressed in the $\left(\mathbf{y},\theta,k\right)$ coordinate system reads:
\begin{align} 
  \label{ExpressionDarbouxMatrix111} 
\bar{\mathcal{P}}_{\eps}\left(\mathbf{y},\theta,k\right)=\left(\begin{array}{cccc}
0 & -\frac{\eps}{B\left(\mathbf{y}\right)} & 0 & 0\\
\frac{\eps}{B\left(\mathbf{y}\right)} & 0 & 0 & 0\\
0 & 0 & 0 & \frac{1}{\eps}\\
0 & 0 & -\frac{1}{\eps} & 0\end{array}\right).
 \end{align}
Moreover, the reciprocal map $\boldsymbol{\kappa}\!_\eps=\boldsymbol{\Upsilon}\!_\eps^{\,-1}$ is smooth with respect to $\eps\in\mathbb{R}_+$, 
and for any positive real numbers $a_{\mathcal{D}}$ and $b_{\mathcal{D}}$ (with $a_{\mathcal{D}}<b_{\mathcal{D}}$) and for any 
$t\in[0,+\infty)$, the trajectory associated with \eqref{systeqs1}-\eqref{systeqs2}, with initial condition 
$(\xvec_0,\vvec_0)\in\rit^2\times \boldsymbol{\mathfrak{C}}(a_{\mathcal{D}},b_{\mathcal{D}})$ (see Notation  \ref{201402230841})
and expressed in the Darboux coordinates, belongs to $\rit^3\times\left[\frac{a_{\mathcal{D}}^2}{2\left\Vert B\right\Vert _{\infty}},\frac{b_{\mathcal{D}}^2}{2}\right]$.
\end{theorem}
Subsections \ref{TheDarCoordSection}, \ref{RegPropertiesOfDarbCCSection}, \ref{ExpOfHamAndPoissMatInDarbCCSection} 
and  \ref{ProofOfFirstPartOfMainThm2} constitute the proof of Theorem \ref{MainThm1}.

 \hspace{0.2cm}

\begin{theorem}
\label{MainThm2}   
For each positive integer $N$, for each compact set $\boldsymbol{K}_{\!\!\mathcal{L}}$, and for each positive real numbers $c_{\!\mathcal{L}}$ and $d_{\!\mathcal{L}}$ (with $c_{\!\mathcal{L}}<d_{\!\mathcal{L}}$), there exists a diffeomorphism $\boldsymbol{\chi}^{N}_\eps:\left(\mathbf{y},\theta,k\right)\mapsto\left(\mathbf{z},\gamma,j\right)$  defined on $\boldsymbol{K}_{\!\!\mathcal{L}}\times\mathbb{R}\times(c_{\!\mathcal{L}},d_{\!\mathcal{L}})$
and a positive real number $\eta_{K_{\!\!\mathcal{L}}}$
 such that, for any $\eps\in[0,\eta_{K_{\!\!\mathcal{L}}}]$, the expansion in power of $\eps$ of  the Hamiltonian function $\hat{H}_{\eps}$
of  system  \eqref{systeqs1}-\eqref{systeqs2} in the $(\mathbf{z},\gamma,j)$ coordinates
does not depend to the oscillation variable $\gamma$ up to order $N$, {\it i.e.}
\begin{align}
   \label{201401230906} 
   &\hat{H}_{\eps}\!\left(\mathbf{z},\gamma,j\right)
   =\hat{H}_{0}\left(\mathbf{z},j\right)+\eps\hat{H}_{1}\left(\mathbf{z},j\right)+\ldots+\eps^{N}\hat{H}_{N}\left(\mathbf{z},j\right)
   +\eps^{N+1}\Rest_{\bar{H}}^{N}\!\left(\eps;\mathbf{z},\gamma,j\right),
\end{align} 
and such that the Poisson Matrix expressed in the $\left(\mathbf{z},\gamma,j\right)$ coordinate system reads:
\begin{align}
     &\eps\hat{\mathcal{P}}_{\eps}\negmedspace\left(\mathbf{z},\gamma,j\right)=\eps\overline{\mathcal{P}}_{\eps}\negmedspace\left(\mathbf{z},\gamma,j\right)+\eps^{N+2}\Rest_{\bar{\mathcal{P}}}^{N}\!\left(\eps;\mathbf{z},\gamma,j\right),
     \label{Formula6Februar2014}  
\end{align}
where $\Rest_{\bar{H}}^{N}$ and $\Rest_{\bar{\mathcal{P}}}^{N}$ are in $\mathcal{C}_{\PerND}^{\infty}\!\left(\left[0,\eta_{K_{\!\mathcal{L}}}\right]\times\boldsymbol{K}_{\!\!\mathcal{L}}\times\rit\times\left[c_{\!\mathcal{L}},d_{\!\mathcal{L}}\right]\right)$.
\end{theorem}
 The proof of Theorem \ref{MainThm2} 
is given in Subsection \ref{SectionProofOfTheoremMainThm2}.

 \vspace{0.3cm}

\begin{remark}
\label{DarbouxRangeInCompact}  
Theorem  \ref{MainThm2} is consistent with  Theorem  \ref{MainThm1}. Indeed,
in Subsection \ref{Subsection4Februar2014} we will show that for any $T\in[0,+\infty)$, for any compact set $\boldsymbol{K}_{\!\mathcal{C}}$, and for any positive real numbers $c_{\mathcal{C}}$ and $d_{\mathcal{C}}$ (with $c_{\mathcal{C}}<d_{\mathcal{C}}$),
 there exists a positive real number $\eta$,  a compact set $\boldsymbol{K}_{\!\mathcal{L}}$, and positive
 real numbers $c_{\mathcal{L}}$ and $d_{\mathcal{L}}$ (with $c_{\mathcal{L}}<d_{\mathcal{L}}$)
  such that  for any $t\in[0,T]$ and for any $\eps\in[0,\eta]$
  the trajectory associated with \eqref{systeqs1}-\eqref{systeqs2}, with initial condition 
$(\xvec_0,\vvec_0)\in\boldsymbol{K}_{\!\mathcal{C}}\times \boldsymbol{\mathfrak{C}}(a_{\mathcal{C}},b_{\mathcal{C}})$ (see Notation  \ref{201402230841})
and expressed in the Darboux coordinates, belongs to $\boldsymbol{K}_{\!\mathcal{L}}\times\rit\times[a_{\mathcal{L}},b_{\mathcal{L}}]$.
\end{remark}  

 \hspace{0.2cm}

\begin{theorem}
\label{MainThm3}   
With the same notations as in Theorem \ref{MainThm2}, 
we consider the function $\hat{H}_{\eps}^{N}$ defined by:
\begin{align}
   &\hat{H}_{\eps}^{N}\left(\mathbf{z},j\right)=\hat{H}_{0}\left(\mathbf{z},j\right)+\eps\hat{H}_{1}\left(\mathbf{z},j\right)+\ldots+\eps^{N}\hat{H}_{N}\left(\mathbf{z},j\right),
\end{align} 
 where  
 $\hat{H}_{0},\ldots,\hat{H}_{N}$ are the $N$ first terms in expansion \eqref{201401230906}  of $\hat{H}_{\eps}$,
 and we denote by $(\mathbf{Z},\boldsymbol{\Gamma},\mathcal{J})=(\mathbf{Z},\boldsymbol{\Gamma},\mathcal{J})(t;\mathbf{z}_{0},\gamma_{0},j_{0})$ the trajectory of Hamiltonian system  \eqref{systeqs1}-\eqref{systeqs2}, expressed in the 
 $(\mathbf{z},\gamma,j)$ coordinate system,  associated with initial condition $\mathbf{z}_{0},\gamma_{0},j_{0}$.
Let
\begin{align}
    &\frac{\partial}{\partial t}\left(\begin{array}{c}
\mathbf{Z}^{T}\\
\boldsymbol{\Gamma}^{T}\\
\mathcal{J}^{T}\end{array}\right)=\overline{\mathcal{P}}_{\eps}\negmedspace\left(\mathbf{Z}^{T}\right)\nabla\hat{H}_{\eps}^{N}\left(\mathbf{Z}^{T},\mathcal{J}^{T}\right),
\end{align}
be the Hamiltonian dynamical system associated with the Hamiltonian function $\hat{H}_{\eps}^{N}$
and with the Poisson Matrix $\overline{\mathcal{P}}_{\eps}$ defined by \eqref{ExpressionDarbouxMatrix111}.
Then, this system satisfies the assumptions of Theorem \ref{KrTh}.
Moreover, for any $T\in[0,+\infty)$, for any compact set $\boldsymbol{K}_{\!\mathcal{C}}$, and for any positive
 real numbers $c_{\mathcal{C}}$ and $d_{\mathcal{C}}$ (with $c_{\mathcal{C}}<d_{\mathcal{C}}$),
 there exists a real number $\eta_{K_{\!\mathcal{C}}}$ and a
 constant $C_{\mathcal{C}}$, independant of $\eps$, such that  for any $t\in[0,T]$ and for any $\eps\in[0,\eta_{K_{\!\mathcal{C}}}]$ 
\begin{align}
    &\sup\left\{ \left|(\mathbf{Z},\mathcal{J})\!\left(t;\mathbf{z}_{0},\gamma_{0},j_{0}\right)-(\mathbf{Z}^{T},\mathcal{J}^{T})\!\left(t;\mathbf{z}_{0},j_{0}\right)\right|, t\in[0,T],\left(\mathbf{z}_{0},\gamma_{0},j_{0}\right)\in\mathcal{U}_{\mathcal{C}}\right\} \leq C\eps^{N-1},
    \label{MainEstimationOfOurPaper}   
\end{align}
where $\mathcal{U}_{\mathcal{C}}$ is the range of $\boldsymbol{K}_{\!\mathcal{C}}\times \boldsymbol{\mathfrak{C}}(c_{\mathcal{C}},d_{\mathcal{C}})$ in the Guiding-Center coordinates of order $N$ {\it i.e.} by diffeomorphism $\boldsymbol{\chi}^{N}_\eps$.
\end{theorem}
The proof of Theorem \ref{MainThm3} 
is led in Subsection \ref{LastSection}.
\\

The paper is organized as follows. In Section \ref{SchemDesGCRedSection} we briefly recall the main steps of the Guiding-Center reduction and we give a proof of Theorem  \ref{KrTh}.
Then, Section \ref{SectDaAl} is devoted to the construction of the Darboux change of coordinates. Especially, we will introduce an intermediary PDE from which the 
Darboux coordinates can be deduced.
We will also perform  a detailed analysis of the regularity of the change of coordinates and its inverse, including the regularity with respect to the small parameter $\eps$,
and we will give the expansions with respect to $\eps$ of the change of coordinates, of its inverse, and of the Hamiltonian function.
In Section \ref{ThePartialLieTransformMethod}, we introduce a partial Lie transform method leading to the Guiding-Center coordinate system of order $N$.
Eventually, in Sections \ref{SectionProofOfTheoremMainThm2} and \ref{LastSection} we will prove Theorems  \ref{MainThm2} and \ref{MainThm3}.


%
 \section{Schematic description of the Guiding-Center reduction}
 \label{SchemDesGCRedSection}   
%

%
\subsection{Panorama}
\label{PanoramaSubSection}   
%

\begin{figure}[htbp]
\begin{center}
\begin{tikzpicture}[node distance=1cm, auto,]
  \node[punkt] (coordusuelles) {{\color{light-red}Usual Coordinates}\vspace{-1mm} \\
   $(\xvec,\vvec)$ \\ ~ \vspace{-2mm} \\
        \scriptsize $~$~~~~~~~~~~~~~~~~~~~\vspace{-3mm}
          \scriptsize \begin{align*} 
             &\fracp{\Xvec}{t}=\Vvec\\
             &\frac{\partial\mathbf{V}}{\partial t}=\frac{1}{\eps}B\left(\mathbf{X}\right)^{\perp}\mathbf{V}
          \end{align*}       
  };
   \node[punkt, inner sep=5pt,below=1.2cm of coordusuelles] (coordcanon) {{\color{light-red}Canonical Coordinates} \\ 
        $(\qvec,\pvec)$ \\ ~ \\
                \scriptsize $\breve H_{\!\eps}\!(\qvec,\pvec),\breve\Pcal_{\!\eps}(\qvec,\pvec)\text{$=$}\Scal$
                ~s.t:~~~~~~~~~~~~~~~~~~~\vspace{-3mm}
         \scriptsize\begin{gather*} 
             \begin{pmatrix} \ds\fracp{\Qvec}{t} \\ ~\vspace{-1.5mm}\\ \ds \fracp{\Pvec}{t} \end{pmatrix} =
             \Scal \Nabla_{\qvec,\pvec} \breve H_{\!\eps}
         \end{gather*}
    }
    edge[pil,<-,] node[auto] {{\color{vert}1}:  Hamiltonian?}(coordusuelles.south);
     \node[punkt,right=1.5cm of coordusuelles] (coordcyl) {{\color{light-red} Usual Coordinates}\\ 
             $(\xvec,\vvec)$ \\ ~ \vspace{-2mm} \\
         \scriptsize $\UsCoo{H}_{\!\eps}\!(\xvec,\vvec),\UsCoo{\Pcal}_{\!\eps}(\xvec,\vvec)$~s.t:~~~~~~~~~~~~~~~~~~~\vspace{-3mm}
         \scriptsize\begin{gather*} 
             \begin{pmatrix} \ds\fracp{\Xvec}{t} \\ ~\vspace{-1.5mm}\\ \ds \fracp{\Vvec}{t} \end{pmatrix} =
             \UsCoo{\Pcal}_{\!\eps}\Nabla_{\xvec,\vvec} \UsCoo{H}_{\!\eps}
         \end{gather*}
     }
     edge[pil,<-,] node[auto] {{\color{vert}2}}(coordcanon);
      \node[punkt,inner sep=5pt,below=0.7cm of coordcyl] (coorddarboux) {{\color{light-red} Polar Coordinates} \\
       $(\xvec,\theta,v)$\\
            \scriptsize$\widetilde H_{\!\eps}\!(v),
            \widetilde\Pcal_{\!\eps}(\xvec,\theta,v)$\vspace{-1mm}     
     }
     edge[pil,<-,] node[auto] {{\color{vert}3}}(coordcyl.south);
     \node[punkt,inner sep=5pt,below=0.7cm of coorddarboux] (coorddarlie) {{\color{light-red} Darboux Almost Canonical Coordinates} \\
      $(\yvec,\theta,k)$\\
            \scriptsize$\overline H_{\!\eps}\!(\yvec,\theta,k),\overline \Pcal_{\!\eps}(\yvec)$ \vspace{-1mm}       
     }
     edge[pil,<-,] node[auto] {{\color{vert}4}: {\color{violet}Darboux Algorithm}}(coorddarboux.south);
     \node[punkt,inner sep=5pt,below=0.7cm of coorddarlie] (coorddablabla) {{\color{light-red} Lie Coordinates} \\
            $(\zvec,\gamma,j)$\\
            \scriptsize$\widehat H_{\!\eps}\!(\zvec,j),\widehat \Pcal_{\!\eps}(\zvec)$\\
            ~~~~~~~~~~~~~~~~~~$=\overline \Pcal_{\!\eps}(\zvec)$    \vspace{-0.7mm}
     }
     edge[pil,<-,] node[auto] {{\color{vert}5}: {\color{violet}Lie Transform}}(coorddarlie.south);
\end{tikzpicture}
\label{figPanorama} 
\caption{A schematic description of the method leading the Gyro-Kinetic Approximation.}
\end{center}
\end{figure}
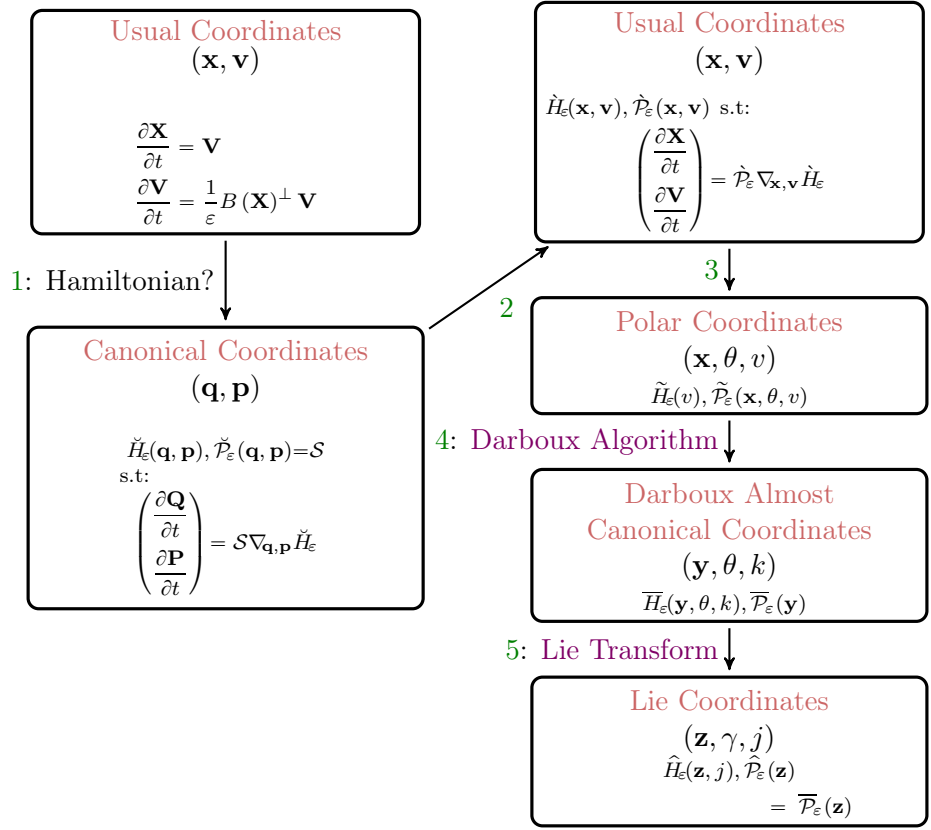
 
A schematic description of the Guiding-Center change of coordinates is summarized in Figure \ref{figPanorama}.
The three main steps of the reduction was already discussed in the introduction. They are symbolized by arrows 3, 4, and 5. 
The first step consists in finding an adequate symplectic structure from which the expressions of the Poisson Matrix and the Hamiltonian function 
are deduced. To achieve this goal we will introduce the canonical coordinates defined by:
\begin{align}
    &\mathbf{q}=\mathbf{x}\text{ and }\mathbf{p}=\frac{\partial L_\eps}{\partial \vvec}\left(\mathbf{x},\mathbf{v}\right)=\mathbf{v}+\frac{1}{\eps}\mathbf{A}\left(\mathbf{x}\right),
\label{DsDim11104} 
\end{align}
where 
\begin{align}
L_{\eps}\left(\mathbf{x},\mathbf{v}\right)=\frac{\left|\mathbf{v}\right|^{2}}{2}+\frac{1}{\eps}\mathbf{v}\cdot\mathbf{A}\left(\mathbf{x}\right),
\label{DsDim11103} 
\end{align}
is the dimensionless electromagnetic Lagrangian and $\mathbf{A}$ is the potential vector.  
Then, the Symplectic Two-Form $\boldsymbol{\Omega}_\eps$ that is considered is the unique Two-Form whose expression in the Canonical Coordinate chart is given by
\begin{align}
\breve{\omega}_{\eps}=d\mathbf{q}\wedge d\mathbf{p}.
\label{DsDim11106} 
\end{align}
Consequently, the Poisson matrix is given by:
\begin{align}
\breve{\mathcal{P}}_{\!\eps}\left(\mathbf{q},\mathbf{p}\right)=\left(\breve{\mathcal{K}}_{\eps}\left(\mathbf{q},\mathbf{p}\right)\right)^{-T}=\mathcal{S}=\left[\begin{array}{cc}
\ 0 & id\\
-id & 0\end{array}\right],
\label{DsDim11108} 
\end{align}
where $\breve{\mathcal{K}}_{\eps}$ is the matrix associated with $\breve{\omega}_{\eps}$.
Eventually it is obvious to show that dynamical system  \eqref{systeqs1}-\eqref{systeqs2} is Hamiltonian with Hamiltonian function 
$\breve{H}_\eps\left(\mathbf{q},\mathbf{p}\right)=\frac{1}{2}\left|\mathbf{p}-\frac{1}{\eps}\mathbf{A}\left(\mathbf{q}\right)\right|^{2}$
and Poisson Matrix $\breve{\mathcal{P}}_{\!\eps}$.
\\

Using the usual change of coordinates rules for the Poisson Matrix and the Hamiltonian function
(see Appendix  \ref{AppendixChangeOfCoordRulesForHamAndPM}) we obtain the following expressions
of the Hamiltonian function and of the Poisson Matrix, in the Cartesian Coordinates:
\begin{align}
\UsCoo{H}_{\eps}\left(\mathbf{x},\mathbf{v}\right)=\frac{\left|\mathbf{v}\right|^{2}}{2},
~~
\UsCoo{\mathcal{P}}_{\eps}\left(\mathbf{x},\mathbf{v}\right)=\left(\begin{array}{cccc}
0 & \ 0 & \ 1 & \ 0\\
0 & \ 0 & \ 0 & \ 1\\
-1 & \ 0 & \ 0 & \frac{B\left(\mathbf{x}\right)}{\eps}\\
0 & -1 & -\frac{B\left(\mathbf{x}\right)}{\eps} & \ 0\end{array}\right),
\label{DsDim111017} 
\end{align}
and, more interesting in the perspective of the next steps, in the Polar in velocity Coordinates: 
 \begin{align} 
       &\tilde{H}_{\eps}\left(\mathbf{x},\theta,v\right)=\frac{v^{2}}{2},
         \label{polarcoordvelham} 
       \\
        &\tilde{\mathcal{P}}_{\eps}\left(\mathbf{x},\theta,v\right)=\left(\begin{array}{cccc}
0 & 0 & -\frac{\cos\left(\theta\right)}{v} & -\sin\left(\theta\right)\\
0 & 0 & \frac{\sin\left(\theta\right)}{v} & -\cos\left(\theta\right)\\
\frac{\cos\left(\theta\right)}{v} & -\frac{\sin\left(\theta\right)}{v} & 0 & \frac{B\left(\mathbf{x}\right)}{\eps v}\\
\sin\left(\theta\right) & \cos\left(\theta\right) & -\frac{B\left(\mathbf{x}\right)}{\eps v} & 0\end{array}\right).
         \label{polarcoordvelPoiss} 
 \end{align}
Before turning to the fourth step we give the proof of Theorem \ref{KrTh}.

%
\subsection{Proof of Theorem \ref{KrTh}}
%

When the Poisson Matrix has the form given by \eqref{HamSystSh1}, the last line of \eqref{HamSystSh0} reads 
\begin{align*}
\frac{\partial \Rvec_{4}}{\partial t}=-\mathcal{P}_{3,4}\frac{\partial H}{\partial r_{3}}\left(\mathbf{R}\right).
\end{align*}
Hence, if the Hamiltonian function does not depend on the penultimate variable, then, the last component $\Rvec_4$ of the trajectory is not time-evolving.
Now, introducing the Poisson Bracket of two functions $f\equiv f\left(\mathbf{r}\right)$ and $g\equiv g\left(\mathbf{r}\right)$ defined by
\begin{align}
\left\{ f,g\right\} _{\!\mathbf{r}}\left(\mathbf{r}\right)=\left[\nabla_{\mathbf{\!r}}f\left(\mathbf{r}\right)\right]^{T}\mathcal{P}\left(\mathbf{r}\right)\nabla_{\!\mathbf{r}}g\left(\mathbf{r}\right),
\label{DsDim111011} 
\end{align}
where $\mathcal{P}\!\left(\mathbf{r}\right)$ is the Poisson Matrix, we have 
\begin{align}
      &\mathcal{P}_{i,j}=\left\{ \mathbf{r}_{i},\mathbf{r}_{j}\right\}_{\mathbf{\!r}} \text{ for }i,j=1,2,3,4,
      \label{PMexpCC}   
\end{align}
where $\mathbf{r}_i$ is the $i$-th coordinate function $\mathbf{r}\mapsto r_i$ and
 a direct computation leads to 
\begin{align}
      &\left\{ \left\{ \mathbf{r}_{1},\mathbf{r}_{2}\right\} _{\mathbf{\! r}},\mathbf{r}_{3}\right\} _{\mathbf{\! r}}\left(\mathbf{r}\right)=-\mathcal{P}_{3,4}\frac{\partial\mathcal{P}_{1,2}}{\partial r_{4}}\left(\mathbf{r}\right)\text{ and }\left\{ \left\{ \mathbf{r}_{1},\mathbf{r}_{2}\right\} _{\!\mathbf{r}},\mathbf{r}_{4}\right\} _{\mathbf{\! r}}\left(\mathbf{r}\right)=\mathcal{P}_{3,4}\frac{\partial\mathcal{P}_{1,2}}{\partial r_{3}}\left(\mathbf{r}\right).
      \label{CPdrP}   
\end{align}
Using the Jacobi identity saying that for any regular function $f,g,h,$
\begin{align}
     &\left\{ \left\{ f,g\right\}_{\mathbf{\!r}} ,h\right\}_{\mathbf{\!r}} +\left\{ \left\{ h,f\right\}_{\mathbf{\!r}} ,g\right\}_{\mathbf{\!r}} +\left\{ \left\{ g,h\right\}_{\mathbf{\!r}} ,f\right\}_{\mathbf{\!r}} =0,
\end{align}
 and the facts that $\mathcal{P}_{3,1}=\mathcal{P}_{2,3}=\mathcal{P}_{4,1}=\mathcal{P}_{2,4}=0$, we obtain
\begin{align}
     &\left\{ \left\{ \mathbf{r}_{1},\mathbf{r}_{2}\right\} _{\!\mathbf{r}},\mathbf{r}_{3}\right\} _{\mathbf{\! r}}=-\left\{ \left\{ \mathbf{r}_{3},\mathbf{r}_{1}\right\} _{\!\mathbf{r}},\mathbf{r}_{2}\right\} _{\mathbf{\! r}}-\left\{ \left\{ \mathbf{r}_{2},\mathbf{r}_{3}\right\} _{\!\mathbf{r}},\mathbf{r}_{1}\right\} _{\mathbf{\! r}}=0,
     \\
     &\left\{ \left\{ \mathbf{r}_{1},\mathbf{r}_{2}\right\} _{\!\mathbf{r}},\mathbf{r}_{4}\right\} _{\mathbf{\! r}}=-\left\{ \left\{ \mathbf{r}_{4},\mathbf{r}_{1}\right\} _{\!\mathbf{r}},\mathbf{r}_{2}\right\} _{\mathbf{\! r}}-\left\{ \left\{ \mathbf{r}_{2},\mathbf{r}_{4}\right\} _{\!\mathbf{r}},\mathbf{r}_{1}\right\} _{\mathbf{\! r}}=0.
\end{align}
and consequently, since  $\mathcal{P}_{1,1}=\mathcal{P}_{2,2}=0$, \eqref{CPdrP} brings \eqref{HamSystSh3}, ending the proof of the theorem.

%
%
%
%
%
%
\section{The Darboux algorithm}
\label{SectDaAl}    
%

\subsection{Objectives}
\label{SubSectDaAlObj}
At this stage, the three first steps of the reduction are already done.
The fourth step (see Figure \ref{figPanorama}) on the way to build the Guiding-Center Approximation is the application of the mathematical algorithm, so called the Darboux Algorithm, to build a global Coordinate System $\left(y_{1},y_{2},\theta,k\right)$ close to the Historic Guiding-Center Coordinate System \eqref{UsualGC1NC}--\eqref{UsualGC4NC}, and in which the Poisson Matrix has the required form \eqref{HamSystSh1} to apply the Key Result (Theorem \ref{KrTh}). In order to manage the small parameter $\eps$, we will build the Coordinate System $\left(y_{1},y_{2},\theta,k\right)$ in order to have $\bar{\mathcal{P}}_{\!\eps}\left(\mathbf{y},\theta,k\right)$ with the following form:
 \begin{gather} 
  \label{DarbouxAlgorithMatrix} 
                   \bar{\mathcal{P}}_{\!\eps}\left(\mathbf{y},\theta,k\right)= \left(\begin{array}{c|cc}
                    \text{\huge${\mathtt M}$}_\eps\left(\mathbf{y}\right) & \begin{array}{c}0 \vspace{-3pt} \\0 \end{array} & \begin{array}{c}0 \vspace{-3pt}\\0 \end{array} \\
                    \hline
                     0\ 0&0&\frac{1}{\eps} \\
                     0\ 0&-\frac{1}{\eps} & 0 \\
                                   \end{array}\right).
\end{gather}
Using the usual change of coordinates rule for the Poisson Matrix, finding this coordinate system 
remains to find a diffeomorphism 
\begin{gather}
\label{DefUpsss}
   \boldsymbol{\Upsilon}\left(\mathbf{x},\theta,v\right)=\left(\boldsymbol{\Upsilon}_{\!1}\left(\mathbf{x},\theta,v\right),\boldsymbol{\Upsilon}_{\!2}\left(\mathbf{x},\theta,v\right),\boldsymbol{\Upsilon}_{\!3}\left(\mathbf{x},\theta,v\right),\boldsymbol{\Upsilon}_{\!4}\left(\mathbf{x},\theta,v\right)\right),
\end{gather}
whose components satisfy the following non-linear hyperbolic system of PDE:
\begin{align}
       &\poibrack{\ytrf_{\! 1}}{\ytrf_{\! 3}}_{\xvec,\theta,v}=0, \hspace{1cm}\poibrack{\ytrf_{\! 1}}{\ytrf_{\! 4}}_{\xvec,\theta,v}=0,
       \label{13ET14}      
        \\
        &\poibrack{\ytrf_{\! 2}}{\ytrf_{\! 3}}_{\xvec,\theta,v}=0,  \hspace{1cm}\poibrack{\ytrf_{\! 2}}{\ytrf_{\! 4}}_{\xvec,\theta,v}=0,
       \label{23ET24}      
        \\
        &\poibrack{\ytrf_{\! 3}}{\ytrf_{\! 4}}_{\xvec,\theta,v}=\frac{1}{\eps}.    
         \label{34}      
\end{align}
The resolution of this set of PDE constitutes the Darboux method.

\vspace{0.5cm}

The first stage of the method consists in setting
\begin{align}
     &\ytrf_{\! 3}=\theta.
\end{align}
Consequently, the non-linear nature of \eqref{13ET14}-\eqref{34} is balanced by the fact that $\theta$ is left unchanged.
With the aim of being
close to the Historical Guiding-Center coordinates (see \eqref{UsualGC1NC}-\eqref{UsualGC4NC}), the boundary conditions are fixed at $v=0$ as follows:
\begin{gather}
\label{BoundCond}  
\begin{aligned}
     &\ytrf_{\!1}\left(\xvec,\theta,0\right)=x_{1},
     \\
     &\ytrf_{\!2}\left(\xvec,\theta,0\right)=x_{2},
     \\
     &\ytrf_{\!4}\left(\xvec,\theta,0\right)=0.
\end{aligned}
\end{gather}
Since the Poisson Matrix $\tilde{\mathcal{P}}_{\eps}$ given by \eqref{polarcoordvelPoiss} has a singularity
at $v=0$, this choice leads to a small difficulty. Nevertheless, it is not a difficult task to fix it. Let $\omega_{\eps}$ be the function defined by:
\begin{gather}
\label{CyCooIVel797} 
 \omega_{\eps}(\xvec,v)=\frac{B(\xvec)}{\eps v},
\end{gather}
and let $\tilde{\mathcal{Q}}_{\eps}$ be the matrix related with the Poisson Matrix by:
\begin{gather}
\label{201305121421} 
\tilde{\mathcal{P}}_{\eps}\left(\mathbf{x},\theta,v\right)=\omega_{\eps}\left(\mathbf{x},v\right)\tilde{\mathcal{Q}}_{\eps}\left(\mathbf{x},\theta,v\right).
\end{gather}
Then, the system of PDE \eqref{13ET14}-\eqref{34} is equivalent, for $v\neq0$, to equations involving $\tilde{\mathcal{Q}}_{\eps}$:
\begin{align}
      &\ds  (\nabla{\ytrf_{1}}) \cdot (\tilde{\mathcal{Q}}_{\eps} (\nabla{\ytrf_{3}} ))=0, \hspace{1cm}\ds (\nabla{\ytrf_{1}}) \cdot (\tilde{\mathcal{Q}}_{\eps} (\nabla{\ytrf_{4}}))=0,
       \label{13ET14WS}      
      \\
      &\ds  (\nabla{\ytrf_{2}}) \cdot (\tilde{\mathcal{Q}}_{\eps}(\nabla{\ytrf_{3}}))=0, \hspace{1cm}\ds  (\nabla{\ytrf_{2}}) \cdot (\tilde{\mathcal{Q}}_{\eps} (\nabla{\ytrf_{4}}))=0,
       \label{23ET24WS}      
      \\
      &\ds  (\nabla{\ytrf_{4}}) \cdot (\tilde{\mathcal{Q}}_{\eps}(\nabla{\ytrf_{3}}))=-\frac{v}{B\left(\mathbf{x}\right)},
      \label{EqWithoutSing} 
\end{align}
that have no singularity in $v=0$. Consequently in place of solving \eqref{13ET14}-\eqref{34}
we will solve 
\eqref{13ET14WS}-\eqref{EqWithoutSing} provided with the set of boundary conditions \eqref{BoundCond}.

\vspace{0.5cm}

In this Section, we will not follow the method given in {\cite{littlejohn:1979}}. We will base the resolution of \eqref{13ET14WS}-\eqref{EqWithoutSing}
on an intermediary PDE from which the solutions of \eqref{13ET14WS}-\eqref{EqWithoutSing} will be deduced. 
Afterwards, we will construct for any fixed $\eps$ map $\boldsymbol{\Upsilon}$.
Then, we will show that  $\boldsymbol{\Upsilon}$ is well a change of coordinates and study its regularity with respect to $\eps$.
Finally, we will prove Theorem \ref{MainThm1} and 
in view of the last Section we will give estimates related to the expression of the characteristics 
expressed in the Darboux Coordinate System.

\subsection{An intermediary equation}
The intermediary equation that we consider in this Section is the following:
\begin{gather}
\label{IntermediaryEq}    
 \left \{
 \begin{aligned}
        &\frac{\partial\varphi}{\partial v}+\eps\boldsymbol{\Lambda}\cdot\varphi  = 0, 
        \\
        &\varphi\left(\mathbf{x},\theta,v=0\right)    =  \varphi_{0}\left(\mathbf{x},\theta\right),
 \end{aligned}
  \right .
  \end{gather}
where 
\begin{align}
     &\varphi_{0}\left(\mathbf{x},\theta\right)=\frac{1}{B\left(\mathbf{x}\right)},
     \label{IntermediaryEqInit}    
\end{align}
and where $\boldsymbol{\Lambda}$ is the vector field defined by:
\begin{align}
    &\boldsymbol{\Lambda}\left(\mathbf{x},\theta\right)=\frac{\cos\left(\theta\right)}{B\left(\mathbf{x}\right)}\frac{\partial}{\partial x_{1}}-\frac{\sin\left(\theta\right)}{B\left(\mathbf{x}\right)}\frac{\partial}{\partial x_{2}}.
    \label{defLambda}    
\end{align}
We denote by $\mathcal{G}_{\lambda}$ its flow and by 
$\boldsymbol{\Lambda}^{n}\cdot$ its iterated application acting on regular functions $f$ as 
\begin{align}
  &\boldsymbol{\Lambda}^{0}\cdot f=f, ~~~~~~~~~~ \boldsymbol{\Lambda}^{1}\cdot f=\frac{\cos\left(\theta\right)}{B\left(\mathbf{x}\right)}\frac{\partial f}{\partial x_{1}}-\frac{\sin\left(\theta\right)}{B\left(\mathbf{x}\right)}\frac{\partial f}{\partial x_{2}},
  \label{defLam1} 
  \\
  & \boldsymbol{\Lambda}^{n}\cdot f=\boldsymbol{\Lambda}\cdot\left(\boldsymbol{\Lambda}^{n-1}\cdot f\right), ~~~\forall n\geq 2.
\label{defLamn} 
\end{align}

In a first place, we give the regularity property of $\mathcal{G}_\lambda$.
\begin{lemma}
\label{regFlowG} 
Flow $\mathcal{G}_\lambda= \mathcal{G}_\lambda(\xvec,\theta)$ of vector field $\boldsymbol{\Lambda}$ is complete, 
in $\mathcal{C}^{\infty}(\rit^3)$,
 $\left(\mathcal{G}_{\lambda}^{1},\mathcal{G}_{\lambda}^{2}\right)$ is in $\mathcal{C}^\infty_{\#,3}(\rit^3)$ (see Notation \ref{201402252149}),
and $\mathcal{G}_{\lambda}^{3}(\xvec,\theta)=\theta.$
\end{lemma}
Then, using this lemma, which proof is straightforward, we obtain the following Theorem.
\begin{theorem}
\label{LemThmOub}    
The unique solution $\varphi$ to \eqref{IntermediaryEq} is given by 
\begin{align}
\varphi\!\left(\mathbf{x},\theta,v\right)=\frac{1}{B\left(\mathcal{G}_{-\eps v}^{1}\left(\mathbf{x},\theta\right),\mathcal{G}_{-\eps v}^{2}\left(\mathbf{x},\theta\right)\right)}.
 \label{SolPDEPhi} 
\end{align}
Moreover, $\varphi$ is in $\Ccal_{\PerND}^\infty(\rit^4)$ and is bounded.
\end{theorem}
\begin{proof}
The proof of Theorem \ref{LemThmOub} is performed with the usual characteristics' method.
Let $\mathcal{F}\left(v,s,\mathbf{x},\theta\right)$ be the characteristic associated with \eqref{IntermediaryEq}, \textit{i.e.}
the solution of 
\begin{gather}
 \left \{
\begin{aligned}
    &\frac{\partial\mathcal{F}_{1}}{\partial v}=\eps\frac{\cos\left(\mathcal{F}_{3}\right)}{B\left(\mathcal{F}_{1},\mathcal{F}_{2}\right)},
    \\
    &\frac{\partial\mathcal{F}_{2}}{\partial v}=-\eps\frac{\sin\left(\mathcal{F}_{3}\right)}{B\left(\mathcal{F}_{1},\mathcal{F}_{2}\right)},
    & ~~~~~~~~~\mathcal{F}\left(s,s,\mathbf{x},\theta\right)=\left(\mathbf{x},\theta\right).
    \\
    &\frac{\partial\mathcal{F}_{3}}{\partial v}=0,
 \end{aligned}
  \right .
  \end{gather}
By definition the flow $\mathcal{G}_{\lambda}$ of $\boldsymbol{\Lambda}$ satisfies:
\begin{gather}
\label{201402262136} 
 \left \{
\begin{aligned}
    &\frac{\partial\mathcal{G}_{\lambda}^{1}}{\partial\lambda}=\frac{\cos\left(\mathcal{G}_{\lambda}^{3}\right)}{B\left(\mathcal{G}_{\lambda}^{1},\mathcal{G}_{\lambda}^{3}\right)},
    \\
    &\frac{\partial\mathcal{G}_{\lambda}^{2}}{\partial\lambda}
    =-\frac{\sin\left(\mathcal{G}_{\lambda}^{3}\right)}{B\left(\mathcal{G}_{\lambda}^{1},\mathcal{G}_{\lambda}^{3}\right)},
    &~~~~~~~~~\mathcal{G}_{0}\left(\mathbf{x},\theta\right)=\left(\mathbf{x},\theta\right).
    \\
    &\frac{\partial\mathcal{G}_{\lambda}^{3}}{\partial\lambda}=0,
 \end{aligned}
  \right .
  \end{gather}
Then, we deduce that $\mathcal{F}\left(v,s,\mathbf{x},\theta\right)=\mathcal{G}_{\eps\left(v-s\right)}\left(\mathbf{x},\theta\right)$.
Eventually Duhamel's formula yields:
\begin{align}
   &\varphi\left(\mathbf{x},\theta,v\right)=\varphi_{0}\left(\mathcal{F}\left(0,v,\mathbf{x},\theta\right)\right)=\frac{1}{B\left(\mathcal{G}_{-\eps v}^{1}\left(\mathbf{x},\theta\right),\mathcal{G}_{-\eps v}^{2}\left(\mathbf{x},\theta\right)\right)}.
\end{align}
This ends the proof of Theorem.
\end{proof}
We will end this Section by giving a Taylor expansion, with respect to $\eps$, of the solution $\varphi$ to \eqref{IntermediaryEq}. 
Such kind of Taylor expansions are usually referred in the literature (see {Olver \cite{Olver}}) as 
\textit{Lie expansions}.

\begin{definition}
\label{LieSums}
If $\boldsymbol{\Lambda}$ is a vector field of $\mathbb{R}^3$ with coefficients which are in $\mathcal{C}^\infty_b\!\!\left(\rit^3\right)$, then we define the Lie Series $S_{L}^{\infty}\left(\boldsymbol{\Lambda}\right)\cdot$ associated with $\boldsymbol{\Lambda}$ by 
\begin{align}
S_{L}^{\infty}\left(\boldsymbol{\Lambda}\right)\cdot=\underset{l\geq0}{\sum}\frac{\left(\boldsymbol{\Lambda}\right)^{l}\cdot}{l!},
\end{align}
where $(\boldsymbol{\Lambda})^{l}$ is defined by \eqref{defLam1} and \eqref{defLamn},
and the partial Lie Sum of order $n$:
\begin{align}
S_{L}^{n}\left(\boldsymbol{\Lambda}\right)\cdot=\underset{l=0}{\overset{n}{\sum}}\frac{\left(\boldsymbol{\Lambda}\right)^{l}\cdot}{l!}.
\end{align}
\end{definition}

It is known that, formally, the flow $\mathcal{G}_{\lambda}$ associated with $\boldsymbol{\Lambda}$ may be expressed in terms of the Lie Series of $\boldsymbol{\Lambda}$:
\begin{align}
\mathcal{G}_{\lambda}=S_{L}^{\infty}\left(\lambda\boldsymbol{\Lambda}\right)\cdot=\underset{l\geq0}{\sum}\frac{\left(\lambda\boldsymbol{\Lambda}\right)^{l}\cdot}{l!}.
\end{align}
More rigorously, as the flow is complete, using its partial Lie Sum we have 
\begin{align}
f\circ\mathcal{G}_{\lambda}=\underset{l=0}{\overset{n}{\sum}}\frac{\lambda^{l}\left(\boldsymbol{\Lambda}\right)^{l}\cdot}{l!}f+\int_{0}^{\lambda}\frac{\left(\lambda-u\right)^{n}}{n!}\left(\boldsymbol{\Lambda}^{n+1}\cdot f\right)\circ\mathcal{G}_{u}du,
\label{LieExp} 
\end{align} 
for any function $f:\ \mathbb{R}^{3}\rightarrow\mathbb{R}$ being $\mathcal{C}^\infty(\mathbb{R}^{3})$.
\\

Taking now $\frac{1}{B}$ as function $f$ and $-\eps v$ as parameter $\lambda$ in (\ref{LieExp}), we obtain
\begin{align}
  &\varphi\left(\mathbf{x},\theta,v\right)=\underset{l=0}{\overset{n}{\sum}}\frac{\left(-\eps v\right)^{l}}{l!}\left(\boldsymbol{\Lambda}^{l}\cdot\frac{1}{B}\right)\left(\mathbf{x},\theta\right)+\int_{0}^{-\eps v}\frac{\left(-\eps v-u\right)^{n}}{n!}\left(\boldsymbol{\Lambda}^{n+1}\cdot\frac{1}{B}\right)\circ\mathcal{G}_{u}\left(\mathbf{x},\theta\right)du.
\end{align}
Hence we have proven the following lemma.

\begin{lemma}
\label{lemMocGen}  
Function $\varphi$, solution to PDE \eqref{IntermediaryEq}, admits for any $n\in\mathbb{N},$ for any
$\eps\in\mathbb{R}$ and for any $\left(\mathbf{x},\theta,v\right)\in\mathbb{R}^4$ the following expansion in power of $\eps$
\begin{gather}
\label{VarphiExpansion}   
\begin{aligned}
    \varphi\left(\mathbf{x},\theta,v\right)
    &=\underset{l=0}{\overset{n}{\sum}}\frac{\left(-\eps v\right)^{l}}{l!}\left(\boldsymbol{\Lambda}^{l}\cdot\frac{1}{B}\right)\left(\mathbf{x},\theta\right)
    \\
    &+\frac{\left(-\eps\right)^{n+1}}{n!}\int_{0}^{v}\frac{\left(v-u\right)^{n}}{n!}\left(\boldsymbol{\Lambda}^{n+1}\cdot\frac{1}{B}\right)\circ\mathcal{G}_{-\eps u}\left(\mathbf{x},\theta\right)du.
\end{aligned}    
\end{gather}
Moreover, for any $l\in\mathbb{N},$ $\left(\boldsymbol{\Lambda}^{l}\cdot\frac{1}{B}\right)$ is in 
$\mathcal{C}_{\PerND,3}^{\infty}(\rit^{3})\cap\mathcal{C}_{b}^{\infty}\!\left(\mathbb{R}^{3}\right)$;
for any $n\in\mathbb{N},$ 
$\left(\eps,\mathbf{x},\theta,v\right)\mapsto\int_{0}^{v}\frac{\left(v-u\right)^{n}}{n!}\left(\boldsymbol{\Lambda}^{n+1}\cdot\frac{1}{B}\right)\circ\mathcal{G}_{-\eps u}\left
(\mathbf{x},\theta\right)du$ is in $\mathcal{C}_{\PerND}^\infty(\rit^5)$; 
and for any $v\in\mathbb{R}$ and any $n\in\mathbb{N}$,
$\left(\eps,\mathbf{x},\theta\right)\mapsto\int_{0}^{v}\frac{\left(v-u\right)^{n}}{n!}\left(\boldsymbol{\Lambda}^{n+1}\cdot\frac{1}{B}\right)\circ\mathcal{G}_{-\eps u}\left(\mathbf{x},\theta\right)du$ 
is bounded by $C_{n}^{\varphi}\!\left(v\right)=\frac{\left|v\right|^{n+1}}{\left(n+1\right)!}\left\Vert \boldsymbol{\Lambda}^{n+1}\cdot\frac{1}{B}\right\Vert _{\infty}.$
\end{lemma}

\subsection{The other equations}
\label{TheOtherEquations}   

In the following Theorem, we will deduce from Theorem \ref{LemThmOub} the solutions $\boldsymbol{\Upsilon}_{\!1}$, $\boldsymbol{\Upsilon}_{\!2}$,
and $\boldsymbol{\Upsilon}_{\!4}$ of the PDEs that are in the left in equalities \eqref{13ET14WS}-\eqref{EqWithoutSing}.
\begin{theorem}
\label{ExpressionOfUpsilonVsG}   
The unique solutions $\boldsymbol{\Upsilon}_{1}$, $\boldsymbol{\Upsilon}_{2}$,
and $\boldsymbol{\Upsilon}_{4}$ of
\begin{align}
      &\ds  (\nabla{\ytrf_{\!1}}) \cdot (\tilde{\mathcal{Q}}_{\eps} (\nabla{\ytrf_{\!3}}))=0, \hspace{2.1cm}\ds \boldsymbol{\Upsilon}_{\negmedspace1}\left(\mathbf{x},\theta,0\right)=x_{1},
       \label{13WSWithInit}      
      \\
      &\ds  (\nabla{\ytrf_{\!2}}) \cdot (\tilde{\mathcal{Q}}_{\eps}(\nabla{\ytrf_{\!3}}))=0, \hspace{2.1cm}\ds  \boldsymbol{\Upsilon}_{\negmedspace2}\left(\mathbf{x},\theta,0\right)=x_{2},
       \label{23WSWithInit}      
      \\
      &\ds  (\nabla{\ytrf_{\!4}}) \cdot (\tilde{\mathcal{Q}}_{\eps}(\nabla{\ytrf_{\!3}}))=-\frac{v}{B\left(\mathbf{x}\right)}, \hspace{1cm}\ds\boldsymbol{\Upsilon}_{\negmedspace4}\left(\mathbf{x},\theta,0\right)=0,
      \label{EqWithoutSingWithInit} 
\end{align}
are given by 
\begin{align}
   &\boldsymbol{\Upsilon}_{\!1}\left(\mathbf{x},\theta,v\right)=x_{1}-\eps\cos\left(\theta\right)\psi\left(\mathbf{x},\theta,v\right),
   \label{DefUps1}  
   \\
   &\boldsymbol{\Upsilon}_{\!2}\left(\mathbf{x},\theta,v\right)=x_{2}+\eps\sin\left(\theta\right)\psi\left(\mathbf{x},\theta,v\right),
   \label{DefUps2}  
   \\
   &\boldsymbol{\Upsilon}_{\!4}\left(\mathbf{x},\theta,v\right)=\int_{0}^{v}\psi\left(\mathbf{x},\theta,s\right)ds,
     \label{UpskvsPsi}   
\end{align}
where $\psi$ is defined by:
\begin{align}
\psi\left(\mathbf{x},\theta,v\right)=\int_{0}^{v}\varphi\left(\mathbf{x},\theta,s\right)ds,
\label{DefPSY}   
\end{align}
with $\varphi$ given by  \eqref{SolPDEPhi}.
\end{theorem}

\begin{proof}
We will only prove Formula  \eqref{DefUps1}. The others (\eqref{DefUps2} and \eqref{UpskvsPsi}) are easily obtained with similar arguments.
Firstly, we notice that \eqref{13WSWithInit} can be rewritten as 
\begin{gather}
\label{IntermEqInProofOfShapeUpsilon}   
 \left \{
\begin{aligned}
    &\frac{\partial \boldsymbol{\Upsilon}_{\!1} }{\partial v}+\eps\frac{\cos\left(\theta\right)}{B\left(\mathbf{x}\right)}\frac{\partial\boldsymbol{\Upsilon}_{\!1}}{\partial x_{1}}-\eps\frac{\sin\left(\theta\right)}{B\left(\mathbf{x}\right)}\frac{\partial\boldsymbol{\Upsilon}_{\!1}}{\partial x_{2}}=0,
    \\
    &\boldsymbol{\Upsilon}_{\negmedspace1}\left(\mathbf{x},\theta,0\right)=x_{1}.
\end{aligned}
  \right .
\end{gather}
Secondly, integrating \eqref{IntermediaryEq} between $0$ and $v$ we obtain
\begin{gather}
 \left \{
\begin{aligned}
    &\frac{\partial\psi}{\partial v}+\eps\frac{\cos\left(\theta\right)}{B\left(\mathbf{x}\right)}\frac{\partial\psi}{\partial x_{1}}-\eps\frac{\sin\left(\theta\right)}{B\left(\mathbf{x}\right)}\frac{\partial\psi}{\partial x_{2}}=\frac{1}{B\left(\mathbf{x}\right)},
    \\
    &\psi\left(\mathbf{x},\theta,0\right)=0.
\end{aligned}
  \right .
\end{gather}
Hence by linearity, $\boldsymbol{\Upsilon}_{\!1}$ given by \eqref{DefUps1} is solution
of \eqref{IntermEqInProofOfShapeUpsilon}. The unicity is obvious.
\end{proof}
To end the resolution of  \eqref{13ET14WS}-\eqref{EqWithoutSing} we only have to check that 
$\boldsymbol{\Upsilon}_{\!1}$ and $\boldsymbol{\Upsilon}_{\!2}$ given by \eqref{DefUps1} and \eqref{DefUps2}
are also solutions to the additional equations that are in the right in \eqref{13ET14WS}-\eqref{EqWithoutSing}.
\begin{theorem}
\label{AdditionnalEqCheckTheorem}   
Functions $\boldsymbol{\Upsilon}_{\!1}$ and $\boldsymbol{\Upsilon}_{\!2}$, defined by \eqref{DefUps1} and \eqref{DefUps2},
and solutions of  \eqref{13WSWithInit} and  \eqref{23WSWithInit}, are also solutions to
\begin{align}
      &\ds  (\nabla{\ytrf_{\!1}}) \cdot (\tilde{\mathcal{Q}}_{\eps} (\nabla{\ytrf_{\!4}}))=0,
       \label{13WSWithInitAd}      
      \\
      &\ds  (\nabla{\ytrf_{\!2}}) \cdot (\tilde{\mathcal{Q}}_{\eps}(\nabla{\ytrf_{\!4}}))=0,
       \label{23WSWithInitAd}      
\end{align}
where $\ytrf_{\!4}$ is defined by \eqref{UpskvsPsi}
and is solution of \eqref{EqWithoutSingWithInit}.
\end{theorem}
\begin{proof}
Firstly, we show that $\left\{ \boldsymbol{\Upsilon}_{\! 1},\boldsymbol{\Upsilon}_{\! 4}\right\}$, which is defined for $v\neq0$ because of the singularity
of $\tilde{\mathcal{P}}_{\eps}$, can be extended smoothly by $0$ in $v=0$. 
Integrating expansion \eqref{VarphiExpansion} (with $n=1$) between $0$ and $v$,
we obtain
 \begin{align}
      &\boldsymbol{\Upsilon}_{\!1}\left(\mathbf{x},\theta,v\right)=x_{1}-\frac{\eps\cos\left(\theta\right)v}{B\left(\mathbf{x}\right)}+\eps^{2}\cos\left(\theta\right)\int_{0}^{v}\left(v-u\right)\left(\boldsymbol{\Lambda}\cdot\frac{1}{B}\right)\left(\mathcal{G}_{-\eps u}\left(\mathbf{x},\theta\right)\right)du.
      \label{IntermExpy10} 
 \end{align}
In the same way, integrating twice \eqref{VarphiExpansion} (with $n=0$) we obtain:
 \begin{align}
      &\boldsymbol{\Upsilon}_{\!4}\left(\mathbf{x},\theta,v\right)=\frac{v^{2}}{2B\left(\mathbf{x}\right)}-\frac{\eps}{2}\int_{0}^{v}\left(v-u\right)\left(\boldsymbol{\Lambda}\cdot\frac{1}{B}\right)\left(\mathcal{G}_{-\eps u}\left(\mathbf{x},\theta\right)\right)du.
      \label{IntermExpk0} 
 \end{align}
Differentiating \eqref{IntermExpy10} with respect to $x_1$ yields 
\begin{multline}
     \frac{\partial\boldsymbol{\Upsilon}_{\!1}}{\partial x_{1}}\left(\mathbf{x},\theta,v\right)=1-\eps\cos\left(\theta\right)v\left(\frac{\partial}{\partial x_{1}}\left(\frac{1}{B}\right)\right)\left(\mathbf{x}\right)
     \\
     +\eps^{2}\cos\left(\theta\right)
     \int_{0}^{v}\left(v-u\right)\left[\left(\frac{\partial}{\partial x_{1}}\left(\boldsymbol{\Lambda}\cdot\frac{1}{B}\right)\right)\left(\mathcal{G}_{-\eps u}\left(\mathbf{x},\theta\right)\right)\frac{\partial\mathcal{G}_{-\eps u}^{1}}{\partial x_{1}}\left(\mathbf{x},\theta\right)\right.
     \\
     + \left.\left(\frac{\partial}{\partial x_{2}}\left(\boldsymbol{\Lambda}\cdot\frac{1}{B}\right)\right)\left(\mathcal{G}_{-\eps u}\left(\mathbf{x},\theta\right)\right)\frac{\partial\mathcal{G}_{-\eps u}^{2}}{\partial x_{1}}\left(\mathbf{x},\theta\right)\right]du.
\end{multline}
As $\frac{1}{B}$ and all its derivatives are bounded and as $\frac{\partial\mathcal{G}_{\lambda}^{1}}{\partial x_{1}}\text{ and }\frac{\partial\mathcal{G}_{\lambda}^{2}}{\partial x_{1}}$ are continuous with respect to $\lambda$ we obtain the following estimate:
\begin{align*}
    &\left|\frac{\partial\boldsymbol{\Upsilon}_{\!1}}{\partial x_{1}}\left(\mathbf{x},\theta,v\right)\right|\leq1+\eps\left|v\right|\left\Vert \frac{\partial}{\partial x_{1}}\frac{1}{B}\right\Vert _{\infty}
    \\
    &\begingroup\scriptsize+\frac{\eps^{2}\left|v\right|^{2}}{2}\times\left[\left\Vert \frac{\partial}{\partial x_{1}}\left(\boldsymbol{\Lambda}\cdot\frac{1}{B}\right)\right\Vert _{\infty}\underset{u\in\left[-\left|v\right|,\left|v\right|\right]}{\sup}\left|\frac{\partial\mathcal{G}_{-\eps u}^{1}}{\partial x_{1}}\right|\left(\mathbf{x},\theta\right)+\left\Vert \frac{\partial}{\partial x_{2}}\left(\boldsymbol{\Lambda}\cdot\frac{1}{B}\right)\right\Vert _{\infty}\underset{u\in\left[-\left|v\right|,\left|v\right|\right]}{\sup}\left|\frac{\partial\mathcal{G}_{-\eps u}^{2}}{\partial x_{1}}\right|\left(\mathbf{x},\theta\right)\right].\endgroup
\end{align*}
 Hence $\frac{\partial\boldsymbol{\Upsilon}_{\!1}}{\partial x_{1}}\left(\mathbf{x},\theta,v\right)=\epsilon_{y_{1}}^{x_{1}}\left(\mathbf{x},\theta,v\right)$ with $\epsilon_{y_{1}}^{x_{1}}\left(\mathbf{x},\theta,v\right)$ such that for any $\left(\mathbf{x},\theta\right),\ v\mapsto\epsilon_{y_{1}}^{x_{1}}\left(\mathbf{x},\theta,v\right)$ is smooth, and is bounded in the neighborhood of $v=0$.
 In the same way, we can show that 
 \begin{gather}
\label{EstimateDiffExp} 
\begin{aligned}
&\ds \frac{\partial\boldsymbol{\Upsilon}_{\!1}}{\partial x_{2}}\left(\mathbf{x},\theta,v\right)=v\epsilon_{y_{1}}^{x_{2}}\left(\mathbf{x},\theta,v\right), ~~~~&
&\ds \frac{\partial\boldsymbol{\Upsilon}_{\!1}}{\partial\theta}\left(\mathbf{x},\theta,v\right)=v\epsilon_{y_{1}}^{\theta}\left(\mathbf{x},\theta,v\right),\\
&\ds \frac{\partial\boldsymbol{\Upsilon}_{\!1}}{\partial v}\left(\mathbf{x},\theta,v\right)=\epsilon_{y_{1}}^{v}\left(\mathbf{x},\theta,v\right), ~~~~&
&\ds \frac{\partial\boldsymbol{\Upsilon}_{\! 4}}{\partial x_{1}}\left(\mathbf{x},\theta,v\right)=v^{2}\epsilon_{k}^{x_{1}}\left(\mathbf{x},\theta,v\right),\\
&\ds \frac{\partial\boldsymbol{\Upsilon}_{\! 4}}{\partial x_{2}}\left(\mathbf{x},\theta,v\right)=v^{2}\epsilon_{k}^{x_{2}}\left(\mathbf{x},\theta,v\right),~~~~&
&\ds \frac{\partial\boldsymbol{\Upsilon}_{\! 4}}{\partial\theta}\left(\mathbf{x},\theta,v\right)=v^{3}\epsilon_{k}^{\theta}\left(\mathbf{x},\theta,v\right),\\
&\ds \frac{\partial\boldsymbol{\Upsilon}_{\! 4}}{\partial v}\left(\mathbf{x},\theta,v\right)=v\epsilon_{k}^{v}\left(\mathbf{x},\theta,v\right).
\end{aligned}
\end{gather}
 with $\epsilon_{y_{1}}^{x_{2}}\left(\mathbf{x},\theta,v\right),\ \epsilon_{y_{1}}^{\theta}\left(\mathbf{x},\theta,v\right),\ \epsilon_{y_{1}}^{v}\left(\mathbf{x},\theta,v\right),\ \epsilon_{k}^{x_{1}}\left(\mathbf{x},\theta,v\right),\  \epsilon_{k}^{x_{2}}\left(\mathbf{x},\theta,v\right), \ \epsilon_{k}^{\theta}\left(\mathbf{x},\theta,v\right),\ \epsilon_{k}^{v}\left(\mathbf{x},\theta,v\right)$ such that for any
  $\left(\mathbf{x},\theta\right)$, the functions $v\mapsto\epsilon_{\bullet}^{\bullet}\left(\mathbf{x},\theta,v\right)$ are smooth, and are bounded in the neighborhood of $v=0$.
 Injecting these expressions in $\left\{ \boldsymbol{\Upsilon}_{\!1},\boldsymbol{\Upsilon}_{\!4}\right\} \left(\mathbf{x},\theta,v\right)=\left(\nabla\boldsymbol{\Upsilon}_{\!1}\right)\cdot\left(\tilde{\mathcal{P}}_{\eps}\nabla\boldsymbol{\Upsilon}_{\!4}\right)$ we obtain $\left\{ \boldsymbol{\Upsilon}_{\!1},\boldsymbol{\Upsilon}_{\!4}\right\} \left(\mathbf{x},\theta,v\right)=v\epsilon_{y_{1},k}\left(\mathbf{x},\theta,v\right)$ with $ \epsilon_{y_{1},k}\left(\mathbf{x},\theta,v\right)$ such that $v\mapsto \epsilon_{y_{1},k}\left(\mathbf{x},\theta,v\right)$ is smooth, and is bounded in the neighborhood of $v=0$ leading that $\left\{ \boldsymbol{\Upsilon}_{\! 1},\boldsymbol{\Upsilon}_{\! 4}\right\} $ can be smoothly extended by $0$ in $v=0$.

As the last step of this proof, because of the Jacobi identity we have
\begin{gather}
\label{LongBar}   
\begin{aligned}
       \forall v\neq0,\ \ \left\{ \left\{ \boldsymbol{\Upsilon}_{\! 1},\boldsymbol{\Upsilon}_{\! 4}\right\} ,\boldsymbol{\Upsilon}_{\!3}\right\} 
       +\left\{ \left\{ \boldsymbol{\Upsilon}_{\!3},\boldsymbol{\Upsilon}_{\negmedspace 1}\right\} ,\boldsymbol{\Upsilon}_{\negmedspace 4}\right\} 
       +\left\{ \left\{ \boldsymbol{\Upsilon}_{\negmedspace 4},\boldsymbol{\Upsilon}_{\!3}\right\} ,\boldsymbol{\Upsilon}_{\negmedspace 1}\right\} =0,
       \end{aligned}
\end{gather}
which reads, because the gradient of a constant is zero, because, according to \eqref{UpskvsPsi}, $\left\{ \boldsymbol{\Upsilon}_{\negmedspace 4},\boldsymbol{\Upsilon}_{\!3}\right\} =\frac{1}{\eps}$ and, as we just saw, because $\boldsymbol{\Upsilon}_{\! 1}$ given by \eqref{DefUps1} satisfies $\left\{ \boldsymbol{\Upsilon}_{\!3},\boldsymbol{\Upsilon}_{\negmedspace 1}\right\}=0$, 
\begin{align}
   &\left\{ \left\{ \boldsymbol{\Upsilon}_{\! 1},\boldsymbol{\Upsilon}_{\! 4}\right\} ,\boldsymbol{\Upsilon}_{\!3}\right\}=0.
   \label{FORM346}  
\end{align}
Dividing  \eqref{FORM346} by $\omega_{\eps}\left(\mathbf{x},\theta\right)$ defined by \eqref{CyCooIVel797}, we obtain that for $v\neq 0$, 
$\left\{ \boldsymbol{\Upsilon}_{\! 1},\boldsymbol{\Upsilon}_{\! 4}\right\}$ is solution to
\begin{align}
      &\left(\nabla\left\{ \boldsymbol{\Upsilon}_{\! 1},\boldsymbol{\Upsilon}_{\! 4}\right\} \right)\cdot\left(\tilde{\mathcal{Q}}_{\eps}\left(\nabla\boldsymbol{\Upsilon}_{\!3}\right)\right)=0.
      \label{EqCPUYUK}   
\end{align}
By continuity of the left hand side of \eqref{EqCPUYUK} on $\mathbb{R}^4,$
we deduce that equality \eqref{EqCPUYUK} is valid on $\mathbb{R}^4.$
As $\left\{ \boldsymbol{\Upsilon}_{\! 1},\boldsymbol{\Upsilon}_{\! 4}\right\}$ may be smoothly extended by $0$ in $v=0$,
and as the unique solution of \eqref{EqCPUYUK} satisfying the boundary condition $\left\{ \boldsymbol{\Upsilon}_{\! 1},\boldsymbol{\Upsilon}_{\! 4}\right\} \left(\mathbf{x},\theta,0\right)=0$ is zero,
we deduce 
that $\boldsymbol{\Upsilon}_{\! 1}$ given by \eqref{DefUps1} satisfies $\left\{ \boldsymbol{\Upsilon}_{\! 1},\boldsymbol{\Upsilon}_{\! 4}\right\}=0$ for all
$\left(\mathbf{x},\theta,v\right)$. Hence \eqref{13WSWithInitAd} follows.

The proof that $\boldsymbol{\Upsilon}_{\!2}$, defined by \eqref{DefUps2}
and solution of \eqref{23WSWithInit}, is solutions of  \eqref{23WSWithInitAd} 
is very similar.
This ends the proof of Theorem  \ref{AdditionnalEqCheckTheorem}.
\end{proof}
%

\subsection{The Darboux coordinate system}
\label{TheDarCoordSection}   

In subsection \ref{TheOtherEquations}  we solved equations \eqref{13ET14WS}-\eqref{EqWithoutSing}, with initial conditions \eqref{BoundCond},
on $\mathbb{R}^4$. Now, we need to check that the restriction of 
$\boldsymbol{\Upsilon}$ to $\mathbb{R}^2\times\rit\times\left(0,+\infty\right)$, also denoted by $\boldsymbol{\Upsilon}$,
 is a diffeomorphism (onto $\mathbb{R}^2\times\rit\times\left(0,+\infty\right)$)  and hence that $\left(\mathbf{y},\theta,k\right)$ makes a true coordinate system on $\mathbb{R}^2\times\rit\times\left(0,+\infty\right)$. 
 
 \vspace{0.3cm}

Firstly, using expressions \eqref{DefUps1} and \eqref{DefUps2} of $\boldsymbol{\Upsilon}_{\!1}$ and $\boldsymbol{\Upsilon}_{\!2}$, formula \eqref{SolPDEPhi}  that gives
the expression of $\varphi=\frac{\partial\psi}{\partial v}$, expression \eqref{defLambda} of $\boldsymbol{\Lambda}$, and by definition of its flow $\mathcal{G}_\lambda$ (see  \eqref{201402262136}), we deduce that 
\begin{align}
       &\frac{\partial\boldsymbol{\Upsilon}_{\!1}}{\partial v}\left(\mathbf{x},\theta,v\right)=\frac{\partial}{\partial v}\mathcal{G}_{-\eps v}^{1}\left(\mathbf{x},\theta\right),
       \label{UpsY1VSFirstCompFlow}   
       \\
       &\frac{\partial\boldsymbol{\Upsilon}_{\!2}}{\partial v}\left(\mathbf{x},\theta,v\right)=\frac{\partial}{\partial v}\mathcal{G}_{-\eps v}^{2}\left(\mathbf{x},\theta\right),
       \label{UpsY2VSSecCompFlow}   
       \\
       &\frac{\partial\boldsymbol{\Upsilon}_{\!3}}{\partial v}\left(\mathbf{x},\theta,v\right)=\frac{\partial}{\partial v}\mathcal{G}_{-\eps v}^{3}\left(\mathbf{x},\theta\right).
       \label{UpsThetaVSThirdCompFlow}   
\end{align}
Hence, since $\boldsymbol{\Upsilon}_{\!1}\left(\mathbf{x},\theta,0\right)=x_{1}$, $\boldsymbol{\Upsilon}_{\!2}\left(\mathbf{x},\theta,0\right)=x_{2}$ and 
$\boldsymbol{\Upsilon}_{\!3}\left(\mathbf{x},\theta,0\right)=\theta$ we obtain that 
\begin{align}
\left(\boldsymbol{\Upsilon}_{\!1}\left(\mathbf{x},\theta,v\right),\boldsymbol{\Upsilon}_{\!2}\left(\mathbf{x},\theta,v\right),\boldsymbol{\Upsilon}_{\!3}\left(\mathbf{x},\theta,v\right)\right)=\mathcal{G}_{-\eps v}\left(\mathbf{x},\theta\right).
\label{LamFlowVSUps}    
\end{align}
From this, it is clear that $\left(\mathbf{y},\theta,v\right)$ makes a coordinate system and that the reciprocal change of coordinates is given by $\left(\mathbf{x},\theta,v\right)=\left(\mathcal{G}_{\eps v}\left(\mathbf{y},\theta\right),v\right)$.\\

In order to show that $\left(\mathbf{y},\theta,k\right)$ makes also a coordinate system we will proceed as follows: we will express $\boldsymbol{\Upsilon}_{\!4}$ in the 
$\left(\mathbf{y},\theta,v\right)$-coordinate system and using this expression, we will express $v$ in terms of $\mathbf{y}$ and $\theta$ and the yielding expression of $\boldsymbol{\Upsilon}_4$ in the 
$\left(\mathbf{y},\theta,v\right)$-coordinate system.
\begin{lemma}
The representative of $\boldsymbol{\Upsilon}_{\!4}$ in the $\left(\mathbf{y},\theta,v\right)$-coordinate system is given by 
\begin{align}
\label{EXPPRRUpsKa}    
\tilde{\boldsymbol{\Upsilon}}_{\!4}\left(\mathbf{y},\theta,v\right)=\int_{0}^{v}\frac{u}{B\left(\mathcal{G}^1_{\eps u}\left(\mathbf{y},\theta\right),\mathcal{G}^2_{\eps u}\left(\mathbf{y},\theta\right)\right)}du.
\end{align}
\end{lemma}
\begin{proof}
Using function $\varphi$ involved in the expression  of $\boldsymbol{\Upsilon}_{\!4}$ (see \eqref{UpskvsPsi}  and \eqref{DefPSY}),
we obtain:
\begin{gather}
\label{NumTrans}    
\begin{aligned}
\boldsymbol{\Upsilon}_{\!4}\left(\mathbf{x},\theta,v\right)
             &=\int_{0}^{v}\left(\int_{0}^{s}\varphi\left(\mathbf{x},\theta,u\right)du\right)ds
             \\
             &=\int_{0}^{v}\left(\int_{u}^{v}\varphi\left(\mathbf{x},\theta,u\right)ds\right)du
             \\
             &=\int_{0}^{v}\left(v-u\right)\varphi\left(\mathbf{x},\theta,u\right)du.
\end{aligned}
\end{gather}
Now, using expressions \eqref{SolPDEPhi}  of $\varphi$ and \eqref{LamFlowVSUps} of $\left(\boldsymbol{\Upsilon}_{\!1},\boldsymbol{\Upsilon}_{\!2},\boldsymbol{\Upsilon}_{\!3}\right)$, 
we obtain
\begin{gather}
\label{NumTrans2}    
\begin{aligned}
\boldsymbol{\Upsilon}_{\!4}\left(\mathbf{x},\theta,v\right)
             &=\int_{0}^{v}\frac{\left(v-u\right)}{B\left(\mathcal{G}^1_{-\eps u}\left(\mathbf{x},\theta\right),\mathcal{G}^2_{-\eps u}\left(\mathbf{x},\theta\right)\right)}du
             \\
             &=\int_{0}^{v}\frac{\left(v-u\right)}{B\left(\mathcal{G}^1_{\eps\left(v-u\right)}\left(\mathcal{G}_{-\eps v}\left(\mathbf{x},\theta\right)\right),
                                                                            \mathcal{G}^2_{\eps\left(v-u\right)}\left(\mathcal{G}_{-\eps v}\left(\mathbf{x},\theta\right)\right)\right)}du
             \\
             &=\int_{0}^{v}(v-u)\;\varphi\!\left(\mathcal{G}_{-\eps v}\left(\mathbf{x},\theta\right),u-v\right)du,
             \\
             &=\int_{0}^{v}u\varphi\left(\boldsymbol{\Upsilon}_{\!1}\left(\mathbf{x},\theta,v\right),\boldsymbol{\Upsilon}_{\!2}\left(\mathbf{x},\theta,v\right),\boldsymbol{\Upsilon}_{\!3}\left(\mathbf{x},\theta,v\right),-u\right)du,
\end{aligned}
\end{gather}
implying, using again \eqref{SolPDEPhi} and that $\boldsymbol{\Upsilon}_{\!1}\left(\mathbf{x},\theta,v\right)$, $\boldsymbol{\Upsilon}_{\!2}\left(\mathbf{x},\theta,v\right)$
and $\boldsymbol{\Upsilon}_{\!3}\left(\mathbf{x},\theta,v\right)$ are the expression of $y_1$, $y_2$ and $\theta$, \eqref{EXPPRRUpsKa} and consequently proving the lemma.
\end{proof}

Having expression \eqref{EXPPRRUpsKa} of $\tilde{\boldsymbol{\Upsilon}}_{\!4}$ on hand, for all $\left(\mathbf{y},\theta\right)\in\mathbb{R}^3$ we can define the parametrized smooth function $\eta=\left[\eta\left(\mathbf{y},\theta\right)\right]$ of $v$  by
\begin{align}
\left[\eta\left(\mathbf{y},\theta\right)\right]\left(v\right)=\tilde{\boldsymbol{\Upsilon}}_{\!4}\left(\mathbf{y},\theta,v\right).
\label{DefEtaFun}   
\end{align}
\begin{lemma}
\label{propeta}
For any $\left(\mathbf{y},\theta\right)\in \rit^2\times\rit$, function $\left[\eta\left(\mathbf{y},\theta\right)\right]$ is a $\mathcal{C}^\infty$\!-diffeomorphism from $(0,+\infty)$ onto 
itself and function $\tilde{\eta}=\tilde{\eta}\left(\mathbf{y},\theta,k\right)$ defined by: 
\begin{align}
      &\tilde{\eta}\left(\mathbf{y},\theta,k\right)=\left[\eta\left(\mathbf{y},\theta\right)\right]^{-1}\left(k\right)
      \label{DefEtaInvTildeV2}     
\end{align}
which gives the expression of $v$, is in $\mathcal{C}_{\PerND}^\infty(\rit^2\times \rit\times(0,+\infty)).$ 
\end{lemma}
\begin{proof}
As 
\begin{align*}
\left[\frac{d[\eta\left(\mathbf{y},\theta\right)]}{dv}\right]\left(v\right)=\frac{v}{B\left(\mathcal{G}^1_{\eps v}\left(\mathbf{y},\theta\right),\mathcal{G}^2_{\eps v}\left(\mathbf{y},\theta\right)\right)}>0,
\end{align*}
$\left[\eta\left(\mathbf{y},\theta\right)\right]$ is a $\mathcal{C}^\infty$\!-diffeomorphism from $\left(0,+\infty\right)$ onto 
\begin{align}
    &\left(\underset{v\rightarrow0}{\lim}\left[\eta\left(\mathbf{y},\theta\right)\right]\left(v\right),\underset{v\rightarrow+\infty}{\lim}\left[\eta\left(\mathbf{y},\theta\right)\right]\left(v\right)\right)
\end{align}
for all $\left(\mathbf{y},\theta\right)$. 
Moreover, according to formula \eqref{EXPPRRUpsKa} we have for any $v>0$ the following estimates:
\begin{align}
      &\frac{v^{2}}{2\left\Vert B\right\Vert _{\infty}}\leq\left[\eta\left(\mathbf{y},\theta\right)\right]\left(v\right)\leq\frac{v^{2}}{2},
\end{align}
and consequently for any $\left(\mathbf{y},\theta\right)\in\mathbb{R}^3$
\begin{align}
     &\left[\eta\left(\mathbf{y},\theta\right)\right]\left(\left(0,+\infty\right)\right)=\left(0,+\infty\right).
\end{align}
Particularly, for any $v\in\left(0,+\infty\right)$ there exists $k\in\left(0,+\infty\right)$ such that
\begin{align}
v=\left[\eta\left(\mathbf{y},\theta\right)\right]^{-1}\left(k\right).
\label{DefEtaInv}    
\end{align}
The regularity of $\tilde\eta$ with respect to $k$ is easily obtained from the fact that  $\left[\eta\left(\mathbf{y},\theta\right)\right]$ 
is a $\mathcal{C}^\infty$\!-diffeomorphism.
The $\mathcal{C}^\infty$\!-nature of $\tilde\eta$ with respect to $\yvec$ and $\theta$ is obtained by computing the successive derivatives of $\left[\eta\left(\mathbf{y},\theta\right)\right]\circ\left[\eta\left(\mathbf{y},\theta\right)\right]^{-1}=id$ and using the regularity of $\eta$ that comes from the regularity of  $\tilde{\boldsymbol{\Upsilon}}_{\!4}$, itself coming from the regularity of $B$ and flow $\mathcal{G}_\lambda$. 
Moreover, the periodicity of $\tilde\eta$ with respect to $\theta$ comes from the fact that 
$\theta\mapsto \left(\mathcal{G}^1_\lambda(\xvec,\theta),\mathcal{G}^2_\lambda(\xvec,\theta)\right)$ is in 
$\mathcal{C}^\infty_\PerOneD(\rit)$ (see Notation \ref{201402262154}) for
any $\xvec\in\rit^2$ as set out in Lemma \ref{regFlowG}.
\end{proof}

 Hence we have proven the following theorem.
\begin{theorem}
\label{ThmInvUpsIsKappa}   
 $\left(\mathbf{y},\theta,k\right)$ makes a coordinate system on $\mathbb{R}^2\times\rit\times\left(0,+\infty\right)$ and  function $\boldsymbol{\kappa}=\boldsymbol{\Upsilon}^{-1}$ is given by 
\begin{align}
\label{1305210347} 
\boldsymbol{\kappa}\left(\mathbf{y},\theta,k\right)=\left(\mathcal{G}_{\eps\tilde{\eta}\left(\mathbf{y},\theta,k\right)}\left(\mathbf{y},\theta\right),\tilde{\eta}\left(\mathbf{y},\theta,k\right)\right),
\end{align}
where $\tilde{\eta}$ is defined by \eqref{DefEtaInvTildeV2}.
\end{theorem}

\subsection{Regularity with respect to $\eps$ of the change of coordinates}
\label{RegPropertiesOfDarbCCSection}   

In this Subection, we will focus on the $\eps$-dependency of $\boldsymbol{\kappa}$. 
According to Formula \eqref{1305210347} and
since $\lambda\mapsto\mathcal{G}_{\lambda}$ is smooth (see Lemma \ref{regFlowG}) we only have to study
the regularity with respect to $\eps$ of
the fourth component $\boldsymbol{\kappa}_v$ of $\boldsymbol{\kappa}$.
To this aim, we will introduce for any $(\mathbf{y},\theta,k)\in\rit^3\times(0,+\infty)$ the parametrized functions 
$\alpha=\left[\alpha\left(\mathbf{y},\theta\right)\right]\!\left(v\right)$, which is defined for $v\in\rit_+$, $\beta=\left[\beta\left(\mathbf{y},\theta,k\right)\right]\!\left(\eps\right)$, which is defined for $\eps\in\left(0,+\infty\right)$, and $\gamma=\left[\gamma\left(\mathbf{y},\theta,k\right)\right]\!\left(\eps\right)$, which is defined for $\eps\in\rit_+$, by 
\begin{align}
          &\left[\alpha\left(\mathbf{y},\theta\right)\right]\left(v\right)=\int_{0}^{v}\frac{s}{B\left(\mathcal{G}_{s}^{1}\left(\mathbf{y},\theta\right),\mathcal{G}_{s}^{2}\left(\mathbf{y},\theta\right)\right)}ds,
          \label{DefAlpha}     
          \\
          &\left[\beta\left(\mathbf{y},\theta,k\right)\right]\left(\eps\right)=\left[\alpha\left(\mathbf{y},\theta\right)\right]^{-1}\left(\eps^{2}k\right),
          \label{DefBeta}     
          \\
          &\left[\gamma\left(\mathbf{y},\theta,k\right)\right]\left(\eps\right)=\sqrt{\frac{\left[\alpha\left(\mathbf{y},\theta\right)\right]\left(\eps\right)}{k}}.
          \label{DefGamma}   
\end{align}
Thus, by construction (see Formula \eqref{EXPPRRUpsKa}) we have
\begin{gather}
\tilde{\boldsymbol{\Upsilon}}_{\!4}\left(\mathbf{y},\theta,v\right)= \frac{1}{\eps^2}\left[\alpha\left(\mathbf{y},\theta\right)\right]\left(\eps v\right)
\text{ or }
\eps v  = \left[\alpha\left(\mathbf{y},\theta\right)\right]^{-1} \!(\eps^2 \tilde{\boldsymbol{\Upsilon}}_{\!4}\!\left(\mathbf{y},\theta,v\right)),
\end{gather}
and in view of \eqref{DefBeta}
\begin{align}
       &\forall\eps>0,\ \boldsymbol{\kappa}_{v}\left(\mathbf{y},\theta,k\right)=\frac{1}{\eps}\left[\beta\left(\mathbf{y},\theta,k\right)\right]\left(\eps\right).
       \label{ExpOfKappavInFuncOfBeta}  
\end{align}
With their help, we can state the following lemma.
\begin{lemma}
\label{TechnicalLemma}   
For any $(\mathbf{y},\theta,k)\in\rit^3\times(0,+\infty)$, function $\beta$ defined by formula \eqref{DefBeta} admits a smooth continuation to $\mathbb{R}_+$ such that 
\begin{align}
    &\left[\beta\left(\mathbf{y},\theta,k\right)\right]\left(0\right)=0,
    \label{DefBeta0}       
\end{align}
Moreover, for any $\eps>0$ we have
\begin{align}
     &\left[\beta\left(\mathbf{y},\theta,k\right)\right]'\!\left(\eps\right)=\frac{1}{\left[\gamma\left(\mathbf{y},\theta,k\right)\right]'\left(\left[\beta\left(\mathbf{y},\theta,k\right)\right]\!\left(\eps\right)\right)},
     \label{DerBeta}    
\end{align}
where $\gamma$ is defined by \eqref{DefGamma}.
\end{lemma}

\begin{proof}
By definition,  function $\eps\mapsto\!\left[\beta(\yvec,\theta,k)\right]\!(\eps)$ is in $\mathcal{C}^{\infty}(\mathbb{R}_{+}^{\star})$ for every $(\yvec,\theta,k)\in\rit^2\times\rit\times(0,+\infty)$.
Moreover,  function $\gamma$ is such that 
\begin{align}
       &\forall\eps>0,\ \left[\gamma\left(\mathbf{y},\theta,k\right)\right]\!\left(\eps\right)=\left[\beta\left(\mathbf{y},\theta,k\right)\right]^{-1}\!\left(\eps\right).
       \label{LinkGamBet}  
\end{align}
Hence, in order to show that $\beta$ admits a smooth continuation on $\mathbb{R}_{+}$ we just have to show that 
$\gamma$ admits a smooth inverse function in the neighborhood of $0$ in $\mathbb{R}_{+}$.
And yet, for all $ \eps \geq 0$, we have 
\begin{align}
     &\left[\gamma\left(\mathbf{y},\theta,k\right)\right]\!\left(\eps\right)=
       \eps\sqrt{\frac{1}{k}\int_{0}^{1}\frac{u}{B\left(\mathcal{G}^1_{\eps u}\left(\mathbf{y},\theta\right),\mathcal{G}^2_{\eps u}\left(\mathbf{y},\theta\right)\right)}du}.
     \label{ExpGam}  
\end{align}
This function is in $\mathcal{C}^{\infty}\!\left(\mathbb{R}_+\right)\!$ and 
\begin{align}
      &\left[\frac{d\gamma\left(\mathbf{y},\theta,k\right)}{d\eps}\right]\left(0\right)=\frac{1}{\sqrt{2kB\left(\mathbf{y}\right)}}\neq0.
\end{align}
Hence, there exists a neighborhood $I$ of $0$ and a smooth function $\delta=\left[\delta\left(\mathbf{y},\theta,k\right)\right]\left(\eps\right)$ defined on $J=\left[\gamma\left(\mathbf{y},\theta,k\right)\right]\left(I\cap\mathbb{R}_{+}\right)$ such that $\left[\gamma\left(\mathbf{y},\theta,k\right)\right]\circ\left[\delta\left(\mathbf{y},\theta,k\right)\right]=id$. Hence we have shown that the smooth function $\beta$ defined on $\mathbb{R}_{+}^{\star}$ admits a smooth continuation to $\mathbb{R}_{+}$.
Then, since \eqref{DerBeta} follows directly \eqref{LinkGamBet}, Lemma \ref{TechnicalLemma} is proven.
\end{proof}

\begin{lemma}
\label{RegBeta}  
Function
\begin{gather}
\left(\mathbf{y},\theta,k,\eps\right) \mapsto \left[\beta\left(\mathbf{y},\theta,k\right)\right]\!\left(\eps\right),
\end{gather}
is in $\mathcal{C}_{\#,3}^{\infty}\!\left(\rit^2\times\rit\times(0,+\infty)\times\rit_{+}\right)$.  
\end{lemma}
The proof of the periodicity with respect to the third variable is similar to the one of Lemma \ref{propeta}.\\

We will now use Formula \eqref{ExpOfKappavInFuncOfBeta}, Lemmas \ref{TechnicalLemma}  and \ref{RegBeta} to deduce an expression of the expansion with respect to 
$\eps$ of the $v$-component of $\boldsymbol{\kappa}=\boldsymbol{\Upsilon}^{-1}$.

\begin{lemma}
\label{SuccessiveDerivativeOfBeta}     
For any $n\in\mathbb{N}^\star$, there exists $P_{n}\in\mathbb{R}_{n-1}\!\left[X_{1},\ldots,X_{n}\right]$ 
(where $\mathbb{R}_{n-1}\!\left[X_{1},\ldots,X_{n}\right]$  stands for the space of the homogeneous polynomial 
of degree $n-1$ in $n$ variables) such that 
\begin{align}
       &\left[\beta\left(\mathbf{y},\theta,k\right)\right]^{\left(n\right)}\left(\eps\right)=\frac{P_{n}\left(\left[\gamma\left(\mathbf{y},\theta,k\right)\right]^{\left(1\right)}\left(\beta\left(\eps\right)\right),\ldots,\left[\gamma\left(\mathbf{y},\theta,k\right)\right]^{\left(n\right)}\left(\beta\left(\eps\right)\right)\right)}{\left(\left[\gamma\left(\mathbf{y},\theta,k\right)\right]^{\left(1\right)}\left(\beta\left(\eps\right)\right)\right)^{2n-1}}.
       \label{ExpSuccessiveDerivativeOfBeta}     
\end{align}
Moreover, thanks to formula \eqref{RecFormForPn}, the $P_n$ can easily be computed by induction.
\\

\end{lemma}

\begin{proof}
Proof of Lemma \ref{SuccessiveDerivativeOfBeta} is easily done by induction. Notice that the inductive formula for $P_n$ is given by:
\begin{gather}
 \label{RecFormForPn}   
\begin{aligned}
     &P_{n+1}\left(X_{1},\ldots,X_{n+1}\right)=-\left(2n-1\right)X_{2}P_{n}\left(X_{1},\ldots,X_{n}\right)+
     \\
     &\hspace{4cm}\overset{n}{\underset{k=1}{\sum}}X_{1}X_{k+1}\frac{\partial P_{n}}{\partial X_{k}}\left(X_{1},\ldots,X_{n}\right).
\end{aligned}
\end{gather}
\end{proof}

Hence finding an expansion of $\boldsymbol{\kappa}_v$ remains to find the successive derivatives of $\big[\!\gamma\left(\mathbf{y},\theta,k\right)\!\big]$
evaluated at $\eps=0$. The following lemma and its proof constitute a constructive way to compute them. 

\begin{lemma}
\label{SuccessiveDerivativesAt0OfBetaSpace}    
For any $l\in\mathbb{N}^\star,$ there exists $a_l\in\mathcal{O}_{T,b}^{\infty}$ (see Notation \ref{201402272144}) such that 
\begin{align}
      &\left[\beta\left(\mathbf{y},\theta,k\right)\right]^{\left(l\right)}\left(0\right)=\sqrt{k}^{l}a_{l}\!\left(\mathbf{y},\theta\right).
      \label{SuccessiveDerivativeAt0OfBetaSpace}    
\end{align}
\end{lemma}


\begin{proof}
On the one hand, for any $\left(\mathbf{y},\theta,k\right)$
and for any $n\in\mathbb{N},$  
$\left[\gamma\left(\mathbf{y},\theta,k\right)\right]$ admits a Taylor-MacLaurin expansion of order $n.$
\begin{gather}
\label{MacLaurenExpOfGamma}   
\begin{aligned}
     \left[\gamma\left(\mathbf{y},\theta,k\right)\right]\left(\eps\right)&=\left[\gamma\left(\mathbf{y},\theta,k\right)\right]\left(0\right)+\eps\left[\gamma\left(\mathbf{y},\theta,k\right)\right]^{\left(1\right)}\left(0\right)+\ldots+\frac{\eps^{n}}{n!}\left[\gamma\left(\mathbf{y},\theta,k\right)\right]^{\left(n\right)}\left(0\right)
     \\
     &+\int_{0}^{\eps}\frac{\left[\gamma\left(\mathbf{y},\theta,k\right)\right]^{\left(n+1\right)}\left(t\right)}{n!}\left(t-\eps\right)^{n}dt.
\end{aligned}
\end{gather}

On the other hand,
applying formula \eqref{LieExp} with $\frac{1}{B},$ multiplying by $\lambda$
and integrating between $0$ and $\eps$ yields:
\begin{gather}
\label{ExpDlOfAlpha}    
\begin{aligned}
        \left[\alpha\left(\mathbf{y},\theta\right)\right]\left(\eps\right)&=\eps^{2}\Big(\overset{n}{\underset{l=0}{\sum}}\frac{\eps^{l}}{\left(l+2\right)l!}\left(\boldsymbol{\Lambda}^{l}\cdot\frac{1}{B}\right)\left(\mathbf{y},\theta\right)
        \\
        &+\frac{\eps^{n+1}}{\left(n+1\right)!}\int_{0}^{1}\left(1-u\right)^{n+1}\left(n+1+u\right)\left(\boldsymbol{\Lambda}^{n+1}\cdot\frac{1}{B}\right)\circ\mathcal{G}_{\eps u}du\Big).
\end{aligned}
\end{gather}
Injecting formula \eqref{ExpDlOfAlpha} in \eqref{DefGamma} yields:
\begin{gather}
\label{SecondExpForExpOfGamma}     
{\small
\begin{aligned}
    &\left[\gamma\left(\mathbf{y},\theta,k\right)\right]\left(\eps\right)=
    \\
    &\frac{\eps}{\sqrt{k}}\sqrt{\left(\overset{n}{\underset{l=0}{\sum}}\frac{\eps^{l}}{\left(l+2\right)l!}\left(\boldsymbol{\Lambda}^{l}\cdot\frac{1}{B}\right)\left(\mathbf{y},\theta\right)+\frac{\eps^{n+1}}{\left(n+1\right)!}\int_{0}^{1}\left(1-u\right)^{n+1}\left(n+1+u\right)\left(\boldsymbol{\Lambda}^{n+1}\cdot\frac{1}{B}\right)\circ\mathcal{G}_{\eps u}du\right)}.
\end{aligned}
}
\end{gather}
Expanding formula  \eqref{SecondExpForExpOfGamma} with respect to $\eps,$ up to order $n,$
by using the usual expansion of $s\mapsto\sqrt{1+s},$
and identifying with formula \eqref{MacLaurenExpOfGamma} yields that for any $l\in\left\{ 0,\ldots,n\right\},$
$\sqrt{k}\left[\gamma\left(\mathbf{y},\theta,k\right)\right]^{\left(l\right)}\left(0\right)\in\mathcal{O}_{T,b}^{\infty}.$

Finally, using formula \eqref{ExpSuccessiveDerivativeOfBeta} we obtain formula \eqref{SuccessiveDerivativeAt0OfBetaSpace}.
This ends the proof of Lemma \ref{SuccessiveDerivativesAt0OfBetaSpace}.
\end{proof}

The two previous Lemmas and Formula \eqref{ExpOfKappavInFuncOfBeta} lead to the following Theorem.

\begin{theorem}
\label{ThExpKv}   
For any $(\mathbf{y},\theta,k)\in\rit\times\rit^2\times(0,+\infty)$, the $v$-component $\boldsymbol{\kappa}_v$ of $\boldsymbol{\kappa}=\boldsymbol{\Upsilon}^{-1}$ admits the following expansion in power of $\eps$:
\begin{gather}
\label{ExpKv}   
\begin{aligned}
       &\boldsymbol{\kappa}_{v}\left(\mathbf{y},\theta,k\right)=\overset{n}{\underset{i=0}{\sum}}\sqrt{k}^{i+1}a_{i+1}\left(\mathbf{y},\theta\right)\frac{\eps^{i}}{\left(i+1\right)!}
       \\
       &\hspace{4cm}+\frac{\eps^{n+1}}{\left(n+1\right)!}\int_{0}^{1}\left(1-u\right)^{n+1}\left[\beta\left(\mathbf{y},\theta,k\right)\right]^{\left(n+2\right)}\!\left(\eps u\right)du,
\end{aligned}
\end{gather}
where the terms $a_i$ of the expansion are defined in Lemma \ref{SuccessiveDerivativesAt0OfBetaSpace}.
Moreover, 
\begin{gather}
\label{Born0}  
\begingroup\scriptsize(\yvec, \theta, k, \eps)\mapsto
\int_{0}^{1}\left(1-u\right)^{n+1}\left[\beta\left(\mathbf{y},\theta,k\right)\right]^{\left(n+2\right)}\!\left(\eps u\right)du
\in \mathcal{C}_{\#,3}^\infty(\rit^2\times\rit\times(0,+\infty)\times\rit_{+}).\endgroup
\end{gather}
\end{theorem}

\begin{remark}
\label{AlgoToComputeTheExpTermOfBeta}   
The terms $a_i$ of expansion \eqref{ExpKv} can be obtained by an inductive process. More precisely,
expanding formula \eqref{SecondExpForExpOfGamma} with respect to $\eps,$ up to order $n,$
by using the usual expansion of $s\mapsto\sqrt{1+s},$ leads to the $n$-th first derivatives of $\left[\gamma\left(\mathbf{y},\theta,k\right)\right]$
evaluated at $\eps=0$. Applying inductively formula \eqref{RecFormForPn} yields easily the expression of $P_n$
involved in Formula \eqref{ExpSuccessiveDerivativeOfBeta}. Thus, evaluating $P_n$
at the $n$ first derivatives of $\left[\gamma\left(\mathbf{y},\theta,k\right)\right]$ evaluated at $\eps=0$, we obtain the expression of 
the coefficient $a_n$ involved in Formula \eqref{ExpKv}.
\end{remark}

Applying Theorem  \ref{ThExpKv}, up to order 2, we obtain 
\begin{gather}
\label{DefExpKvOrder2V12}     
\begin{aligned}
       \boldsymbol{\kappa}_{v}\left(\mathbf{y},\theta,k\right)
        &=\sqrt{2kB\left(\mathbf{y}\right)}+\eps\frac{2kB\left(\mathbf{y}\right)}{3}\hat{a}\left(\theta\right)\cdot\nabla_{\mathbf{x}}B\left(\mathbf{y}\right)
       \\
       &-\eps^{2}k\sqrt{\frac{kB\left(\mathbf{y}\right)}{2}}\left[\frac{7}{18B\left(\mathbf{y}\right)^{3}}\left(\hat{a}\left(\theta\right)\cdot\nabla_{\mathbf{x}}B\left(\mathbf{y}\right)\right)^{2}-\frac{\hat{a}\left(\theta\right)^{T}\mathcal{H}_{B}\left(\mathbf{y}\right)\hat{a}\left(\theta\right)}{2B\left(\mathbf{y}\right)^{2}}\right]
       \\
       &+\frac{\eps^{3}}{3!}\int_{0}^{1}\left(1-u\right)^{3}\left[\beta\left(\mathbf{y},\theta,k\right)\right]^{\left(4\right)}\left(\eps u\right)du
\end{aligned}
\end{gather}
where $\hat{\avec}=\hat{\avec}\left(\theta\right)$ is defined by
\begin{align}
\label{1305200909}  
     &\hat{\avec}\left(\theta\right)=\left(\begin{array}{c}
\cos\left(\theta\right)\\
-\sin\left(\theta\right)\end{array}\right)
\end{align}
and where $\mathcal{H}_{B}$ is the Hessian Matrix of $B$.

\begin{remark}
\label{RemForRev1}   
Formula \eqref{DefExpKvOrder2V12}  can already be found in {Littlejohn \cite{littlejohn:1979}} but without estimation of the rest. 
In the present paper formula 
\eqref{DefExpKvOrder2V12}  gives an expansion in power of $\eps$ of a well defined diffeomorphism even though it is obtained in {Littlejohn \cite{littlejohn:1979}} by truncating 
a formal Hilbert expansion.
\end{remark}

\subsection{Expression of the Hamiltonian function and the Darboux Matrix}
\label{ExpOfHamAndPoissMatInDarbCCSection}   
\begin{theorem}
\label{ExpPoissMat} 
The Poisson Matrix in the Darboux Coordinate System is given by
\begin{align} 
  \label{ExpressionDarbouxMatrix} 
\bar{\mathcal{P}}_{\eps}\left(\mathbf{y},\theta,k\right)=\left(\begin{array}{cccc}
0 & -\frac{\eps}{B\left(\mathbf{y}\right)} & 0 & 0\\
\frac{\eps}{B\left(\mathbf{y}\right)} & 0 & 0 & 0\\
0 & 0 & 0 & \frac{1}{\eps}\\
0 & 0 & -\frac{1}{\eps} & 0\end{array}\right).
 \end{align}
\end{theorem}

\begin{proof}

By construction, from formula \eqref{PMexpCC}, we know all  the Poisson Matrix entries, except its entry number $\left(1,2\right)$: $\left\{ \boldsymbol{\Upsilon}_{\!1},\boldsymbol{\Upsilon}_{\!2}\right\} _{\mathbf{x},\theta,v}\left(\boldsymbol{\kappa}\left(\mathbf{y},\theta,k\right)\right)$. Hence, the proof of Theorem 
\ref{ExpPoissMat} reduces to show that: 
\begin{align}
\left\{ \boldsymbol{\Upsilon}_{\!1},\boldsymbol{\Upsilon}_{\!2}\right\} _{\mathbf{x},\theta,v}\left(\mathbf{x},\theta,v\right)=-\frac{\eps}{B\left(\boldsymbol{\Upsilon}_{\!1}\left(\mathbf{x},\theta,v\right),\boldsymbol{\Upsilon}_{\!2}\left(\mathbf{x},\theta,v\right)\right)}.
\label{PoissEntry} 
\end{align}
For that purpose, we will identify the Poisson Bracket between $\boldsymbol{\Upsilon}_{\!1}$ and $\boldsymbol{\Upsilon}_{\!2}$ as the unique solution of the PDE of unknown $u$
\begin{gather}
 \label{PDEPoissBrack12} 
 \left\{
\begin{aligned}
    &-\eps\boldsymbol{\Lambda}^{1}\cdot u-\frac{\partial u}{\partial v}=0,
    \\
    &u\left(\mathbf{x},\theta,0\right)=\frac{-\eps}{B\left(\mathbf{x}\right)}.
\end{aligned}
\right.
\end{gather}
In a first place, as function $\varphi$ defined by \eqref{SolPDEPhi} is the unique solution of \eqref{IntermediaryEq}, the unique solution of 
\eqref{PDEPoissBrack12} is given by
\begin{align}
   &u\left(\mathbf{x},\theta,v\right)=-\eps\varphi\left(\mathbf{x},\theta,v\right);
   \label{SolPDEPoissBrack12} 
\end{align}
\textit{i.e.} by \eqref{PoissEntry}.

On another hand as for any $v\neq0,\ \left\{ \boldsymbol{\Upsilon}_{\!3},\boldsymbol{\Upsilon}_{\!1}\right\} _{\mathbf{x},\theta,v}=0$ and $\left\{ \boldsymbol{\Upsilon}_{\!2},\boldsymbol{\Upsilon}_{\!3}\right\} _{\mathbf{x},\theta,v}=0$, the Jacobi identity ensures that 
  \begin{align}
 \label{DsDim394} 
\forall v\neq0,\ \left\{ \left\{ \boldsymbol{\Upsilon}_{\!1},\boldsymbol{\Upsilon}_{\!2}\right\} ,\boldsymbol{\Upsilon}_{\!3}\right\} _{\mathbf{x},\theta,v}=0.
 \end{align}   
 Hence, dividing \eqref{DsDim394} by $ \omega_{\eps}(\xvec,v)$, we obtain that for $v\neq0$, $\left\{ \boldsymbol{\Upsilon}_{\!1},\boldsymbol{\Upsilon}_{\!2}\right\}$ is solution of \eqref{PDEPoissBrack12}. 
 Using now the same method as when proving Theorem  \ref{AdditionnalEqCheckTheorem}, we obtain
 \begin{align}
 \left\{ \boldsymbol{\Upsilon}_{\!1},\boldsymbol{\Upsilon}_{\!2}\right\} _{\mathbf{x},\theta,v}\left(\mathbf{x},\theta,v\right)=-\frac{\eps}{B\left(\mathbf{x}\right)}+v\epsilon_{y_{1},y_{2}}\left(\mathbf{x},\theta,v\right),
 \end{align}
 with $\epsilon_{y_{1},y_{2}}\left(\mathbf{x},\theta,v\right)$ such that 
 for any  $\left(\mathbf{x},\theta\right),\ v\mapsto\epsilon_{y_{1},y_{2}}\left(\mathbf{x},\theta,v\right)$ is bounded in the neighborhood of $v=0$ and consequently that 
$\left\{ \boldsymbol{\Upsilon}_{\!1},\boldsymbol{\Upsilon}_{\!2}\right\} _{\mathbf{x},\theta,v}\left(\mathbf{x},\theta,0\right)=\frac{-\eps}{B\left(\mathbf{x}\right)}$.\\

As a conclusion, $\left\{ \boldsymbol{\Upsilon}_{\negmedspace 1},\boldsymbol{\Upsilon}_{\negmedspace 2}\right\} _{\mathbf{x},\theta,v}=u$, and $u$ is given by \eqref{PoissEntry}. Hence the Theorem is proven.
\end{proof}
In the sequel, we will denote by $\left\{ f,g\right\} _{\negmedspace\mathcal{D}}$ the Poisson bracket expressed in the Darboux coordinate system, \textit{i.e.}
for any smooth functions $f=f\!\left(\mathbf{y},\theta,k\right)$ and $g=g\!\left(\mathbf{y},\theta,k\right)$ :
\begin{align}
      &\left\{ f,g\right\} _{\negmedspace\mathcal{D}}=\left(\nabla_{\left(\mathbf{y},\theta,k\right)}f\right)\cdot\left(\bar{\mathcal{P}}_{\eps}\nabla_{\left(\mathbf{y},\theta,k\right)}g\right).
\end{align}

In the Darboux Coordinate System, the Hamiltonian function is given by $\bar{H}_{\eps}\left(\mathbf{y},\theta,k\right)=\tilde{H}_{\eps}\left(\boldsymbol{\kappa}\left(\mathbf{y},\theta,k\right)\right)$. Since $\tilde{H}_{\eps}\left(\mathbf{x},\theta,v\right)=\frac{v^{2}}{2}$, we have 
\begin{align}
\bar{H}_{\eps}\left(\mathbf{y},\theta,k\right)=\frac{\boldsymbol{\kappa}_{v}^{2}\left(\mathbf{y},\theta,k\right)}{2}. 
\label{HamDarb}    
\end{align}
Hence, according to Theorem \ref{ThExpKv}, Hamiltonian function $\bar{H}_{\eps}$ is regular with
respect to $\eps$ on $\mathbb{R}_+$ and it admits an expansion in power of $\eps.$
More precisely, using expansion \eqref{ExpKv}, we obtain the following corollaries.

\begin{corollary}
\label{HamDarbExp}    
The Hamiltonian function in the Darboux Coordinate System admits the following expansion in power of $\eps$:
\begin{align}
 \label{DsDim423} 
\bar{H}_{\eps}\left(\mathbf{y},\theta,k\right)=\bar{H}_{0}\left(\mathbf{y},k\right)+\overset{N}{\underset{n=1}{\sum}}\eps^{n}\bar{H}_{n}\left(\mathbf{y},\theta,k\right)+\eps^{N+1}\iota_{N+1}\left(\eps,\mathbf{y},\theta,k\right),
\end{align}
where function $\iota_{N+1}$ is in $\mathcal{C}_{\PerND}^{\infty}(\rit_+\times\rit^2\times\rit\times(0,+\infty))$.
Moreover, for any $n\in\left\{ 1,\ldots,N\right\} $ there exists a function $b_{n}\in\mathcal{O}_{T,b}^{\infty}$ such that 
\begin{align}
      &\bar{H}_{n}\left(\mathbf{y},\theta,k\right)=\sqrt{k}^{n+2}b_{n}\left(\mathbf{y},\theta\right).
\end{align}
\end{corollary}
\begin{remark}
Notice that expansion \eqref{ExpKv}, where the coefficients $a_i$ are computed by using 
the algorithm given in Remark \ref{AlgoToComputeTheExpTermOfBeta}, 
 and Formula \eqref{HamDarb} give a constructive way to compute the $b_n$ and consequently the $\bar{H}_{n}$.
\end{remark}
For instance up to order $2$ we obtain:

\begin{corollary}
\label{HamDarbOrderTwo}    
The Hamiltonian function in the Darboux Coordinate System admits, up to order 2, the following expansion in power of $\eps$:
\begin{gather}
 \label{DsDim422} 
\begin{aligned}
\bar{H}_{\eps}\left(\mathbf{y},\theta,k\right)
&=B\left(\mathbf{y}\right)k+\eps\frac{\hat{\avec}\left(\theta\right)\cdot\nabla_{\mathbf{x}}B\left(\mathbf{y}\right)}{3B\left(\mathbf{y}\right)^{2}}\left(2B\left(\mathbf{y}\right)k\right)^{\frac{3}{2}}
\\
&+\eps^{2}\frac{\left(2B\left(\mathbf{y}\right)k\right)^{2}}{24B\left(\mathbf{y}\right){}^{2}}\left[-\hat{\avec}\left(\theta\right)\cdot\nabla_{\mathbf{x}}B\left(\mathbf{y}\right)+3B\left(\mathbf{y}\right)\hat{\avec}\left(\theta\right)^{T}\mathcal{H}_{B}\left(\mathbf{y}\right)\hat{\avec}\left(\theta\right)\right]
\\
&+\eps^{3}\iota_{3}\left(\mathbf{y},\theta,k,\eps\right),
\end{aligned}
\end{gather}
where  $\hat{\avec}$ is defined by \eqref{1305200909}, function $\iota_3$ is in $\mathcal{C}_{\#,3}^{\infty}(\rit^2\times\rit\times(0,+\infty)\times\rit_+)$,
and where $\mathcal{H}_{B}$ stands for the Hessian matrix associated with $B$.
\end{corollary}

\begin{remark}
In expression \eqref{DsDim423}, there is an important fact for the setting out of the to come Lie Transform based Method: the first term is independent of $\theta$.
\end{remark}

\begin{remark}
\label{RemForRev2}   
Formula \eqref{DsDim422} can also be found in a formal way in {Littlejohn \cite{littlejohn:1979}}. 
\end{remark}

\subsection{Trajectory localization in the Darboux Coordinate System}
\label{ProofOfFirstPartOfMainThm2}    

Subsequently,  we will denote by 
$(Y_{1}^{\eps}, Y_{2}^{\eps}, \Theta_{\boldsymbol{\mathfrak{Dar}}}^{\eps},  \boldsymbol{\mathcal{K}}_{\boldsymbol{\mathfrak{Dar}}}^{\eps})(t;\mathbf{y},\theta,k)$
the trajectories of the dynamical system expressed in the Darboux Coordinates
and by $(\mathbf{X}_{\boldsymbol{\mathfrak{Pol}}}^{\eps},\Theta^{\eps},\mathcal{V}^{\eps})(t;\mathbf{x},\theta,v)$ their expressions in  
the Polar in velocity Coordinate System. 
 \begin{lemma}
 \label{ThmainresultCharDarbKDarb} 
Let $\left[a,b\right]$ be an interval such that $\left[a,b\right]\subset\left(0,+\infty\right)$.
Then, for any $\left(\mathbf{y},\theta\right)\in\mathbb{R}^{3}$ and for any $v\in\left[a,b\right]$,
$\boldsymbol{\Upsilon}_{\! 4}\!\left(\mathbf{x},\theta,v\right)\in\left[\frac{a^{2}}{2\left\Vert B\right\Vert _{\infty}},\frac{b^{2}}{2}\right]$.
Moreover,
for any initial condition $\left(\mathbf{y},\theta,k\right)\in\boldsymbol{\Upsilon}\!\!\left(\mathbb{R}^{3}\times\left[a,b\right]\right)$,
for any $\eps\in\left(0,+\infty\right),$ and for any 
$t\in\mathbb{R},$ $\boldsymbol{\mathcal{K}}_{\boldsymbol{\mathfrak{Dar}}}^{\eps}\left(t;\mathbf{y},\theta,k\right)\in\left[\frac{a^{2}}{2\left\Vert B\right\Vert _{\infty}},\frac{b^{2}}{2}\right]$. 
\\
\end{lemma}

\begin{lemma}
 \label{ThmainresultCharDarbYDarb} 
Let $\left[a,b\right]$ be an interval such that $\left[a,b\right]\subset\left(0,+\infty\right)$ and $T$ be a positive real number. Then, for any initial condition 
$(\mathbf{x},\theta,v)\in\rit^3\times[a,b]$
and for any $t\in[0,T]$, we have
\begin{gather}
\left\{
\begin{aligned}
     &\boldsymbol{\Upsilon}_{\!1}\big(\mathbf{X}_{\boldsymbol{\mathfrak{Pol}}}^{\eps}\left(t,\mathbf{x},\theta,v\right),\Theta^{\eps}\left(t,\mathbf{x},\theta,v\right),v\big)=\boldsymbol{\Upsilon}_{\!1}\left(\mathbf{x},\theta,v\right)+\boldsymbol{\rho}_{1}\left(t;\eps;\mathbf{x},\theta,v\right),
     \\
     &\boldsymbol{\Upsilon}_{\!2}\big(\mathbf{X}_{\boldsymbol{\mathfrak{Pol}}}^{\eps}\left(t,\mathbf{x},\theta,v\right),\Theta^{\eps}\left(t,\mathbf{x},\theta,v\right),v\big)=\boldsymbol{\Upsilon}_{\!2}\left(\mathbf{x},\theta,v\right)+\boldsymbol{\rho}_{2}\left(t;\eps;\mathbf{x},\theta,v\right),
\end{aligned}
\right.
\end{gather}
where $\boldsymbol{\rho}_{1}$ and $\boldsymbol{\rho}_{2}$ satisfy
\begin{align}
   &\left|\boldsymbol{\rho}_{i}\left(t;\eps;\mathbf{x},\theta,v\right)\right|\leq T\left|\eps\right|b^{2}\underset{(\mathbf{x},\theta)\in\rit^{3}}{\sup}\left|\frac{\hat{c}\left(\theta\right)\cdot\nabla_{\mathbf{x}}B\left(\mathbf{x}\right)}{B\left(\mathbf{x}\right)}\right|+\eps^{2}b^{2}\left\Vert \boldsymbol{\Lambda}\cdot\frac{1}{B}\right\Vert _{\infty}.
\end{align}
 \end{lemma}
 
 \begin{lemma}
 \label{ThmainresultCharDarbYDarb2} 
Let $\left[a,b\right]$ be an interval such that $\left[a,b\right]\subset\left(0,+\infty\right)$. Then, for any $(\mathbf{x},\theta,v)\in\rit^3\times[a,b]$
we have
\begin{gather}
\left\{
\begin{aligned}
     &\boldsymbol{\Upsilon}_{\!1}(\mathbf{x},\theta,v)=x_{1}+\boldsymbol{\rho}_{3}\left(\eps;\mathbf{x},\theta,v\right),
     \\
     &\boldsymbol{\Upsilon}_{\!2}(\mathbf{x},\theta,v)=x_{2}+\boldsymbol{\rho}_{4}\left(\eps;\mathbf{x},\theta,v\right),
\end{aligned}
\right.
\end{gather}
where $\boldsymbol{\rho}_{3}$ and $\boldsymbol{\rho}_{4}$ satisfy
\begin{align}
   &\left|\boldsymbol{\rho}_{i}\left(\eps;\mathbf{x},\theta,v\right)\right|\leq\eps b.
\end{align}
 \end{lemma}
 
We will prove Lemmas \ref{ThmainresultCharDarbKDarb}, \ref{ThmainresultCharDarbYDarb} and \ref{ThmainresultCharDarbYDarb2} in Subsection \ref{SubProofTheoremsThmainresultCharDarbKDarbandY}.


\subsection{Proof of Lemmas \ref{ThmainresultCharDarbKDarb}, \ref{ThmainresultCharDarbYDarb} and \ref{ThmainresultCharDarbYDarb2}}
\label{SubProofTheoremsThmainresultCharDarbKDarbandY}  

By definition $\boldsymbol{\Upsilon}_{\!4}\!\left(\mathbf{x},\theta,v\right)=\int_{0}^{v}\psi\left(\mathbf{x},\theta,s\right)ds$,
where $\psi$ is defined by \eqref{DefPSY} with $\varphi$ given by Theorem \ref{LemThmOub}.  Hence,
\begin{gather}
\label{FormulaInequalityPsiLeft}    
\begin{aligned}
        \psi\left(\mathbf{x},\theta,s\right)
        &=\int_{0}^{s}\varphi\left(\mathbf{x},\theta,u\right)du
             =\int_{0}^{s}\frac{1}{B\left(\mathcal{G}_{-\eps u}^{1}\left(\mathbf{x},\theta\right),\mathcal{G}_{-\eps u}^{2}\left(\mathbf{x},\theta\right)\right)}du
             \geq \frac{s}{\left\Vert B\right\Vert _{\infty}},
\end{aligned}
\end{gather}
and consequently, for any $v\in\left[a,b\right]$ and for any $\left(\mathbf{x},\theta\right)\in\mathbb{R}^2\times\rit$, we obtain
\begin{gather}
\label{FormulaBoldUpsKLeft}    
\begin{aligned}
        \boldsymbol{\Upsilon}_{\!4}\!\left(\mathbf{x},\theta,v\right)
        &=\int_{0}^{v}\psi\left(\mathbf{x},\theta,s\right)ds
         \geq\frac{v^{2}}{2\left\Vert B\right\Vert _{\infty}}
         \geq\frac{a^{2}}{2\left\Vert B\right\Vert _{\infty}}.
\end{aligned}
\end{gather}
On another hand, since $\underset{\mathbf{x}\in\mathbb{R}^{2}}{\inf}B\left(\mathbf{x}\right)\geq1,$
we obtain
$\psi\left(\mathbf{x},\theta,s\right)\leq s,$ and consequently  for any $v\in\left[a,b\right]$
and for any $\left(\mathbf{x},\theta\right),$  we obtain
\begin{align}     
     &\boldsymbol{\Upsilon}_{\! 4}\!\left(\mathbf{x},\theta,v\right)\leq\frac{v^{2}}{2}\leq\frac{b^{2}}{2}.
     \label{EstimationForTheBothLemmas}   
\end{align}

Since 
for any $\left(\mathbf{x},\theta,v\right)\in\mathbb{R}^{2}\times\rit\times\left(0,+\infty\right)$ and for any $t\in\mathbb{R},$
\begin{align}
    &\frac{\partial\mathcal{V}^{\eps}}{\partial t}(t,\mathbf{x},\theta,v)=0,
\end{align}
we obtain $\mathcal{V}^{\eps}\left(t;\mathbf{x},\theta,v\right)=v,$ and consequently for any $\left(\mathbf{y},\theta,k\right)\in\mathbb{R}^{2}\times\rit\times\left(0,+\infty\right)$ and for any $t\in\mathbb{R}$,
we have:
\begin{gather}
\label{CaractDarVsCaracPol}    
\begin{aligned}
     \boldsymbol{\mathcal{K}}_{\boldsymbol{\mathfrak{Dar}}}^{\eps}\left(t;\mathbf{y},\theta,k\right)&=\boldsymbol{\Upsilon}_{\! 4}\left(\mathbf{X}_{\boldsymbol{\mathfrak{Pol}}}^{\eps}\left(t;\boldsymbol{\kappa}\left(\mathbf{y},\theta,k\right)\right),\Theta^{\eps}\left(t;\boldsymbol{\kappa}\left(\mathbf{y},\theta,k\right)\right),\mathcal{V}^{\eps}\left(t;\boldsymbol{\kappa}\left(\mathbf{y},\theta,k\right)\right)\right)
     \\
     &=\boldsymbol{\Upsilon}_{\!4}\left(\mathbf{X}_{\boldsymbol{\mathfrak{Pol}}}^{\eps}\left(t;\boldsymbol{\kappa}\left(\mathbf{y},\theta,k\right)\right),\Theta^{\eps}\left(t;\boldsymbol{\kappa}\left(\mathbf{y},\theta,k\right)\right),\boldsymbol{\kappa}_{v}\left(\mathbf{y},\theta,k\right)\right).
\end{aligned}
\end{gather}
Now, for any $\left(\mathbf{y},\theta,k\right)\in\boldsymbol{\Upsilon}\left(\mathbb{R}^{2}\times\rit\times\left[a,b\right]\right),$
$\boldsymbol{\kappa}_{v}\left(\mathbf{y},\theta,k\right)\in\left[a,b\right]$ and estimates \eqref{FormulaBoldUpsKLeft}   
and  \eqref{EstimationForTheBothLemmas} yield that $\boldsymbol{\mathcal{K}}_{\boldsymbol{\mathfrak{Dar}}}^{\eps}\left(t;\mathbf{y},\theta,k\right)\in\left[\frac{a^{2}}{2\left\Vert B\right\Vert _{\infty}},\frac{b^{2}}{2}\right]$.
This ends the proof of Lemma \ref{ThmainresultCharDarbKDarb}.
{~\hfill $\square$} 
\\

Concerning Lemma \ref{ThmainresultCharDarbYDarb}, for any  $\left(\mathbf{x},\theta\right)\in\mathbb{R}^2\times\rit$ and for any $v\in\left[a,b\right],$ function $\psi$
satisfies $\left|\psi\left(\mathbf{x},\theta,v\right)\right|\leq b.$ 
Applying formula \eqref{IntermExpy10} yields:
\begin{gather}
\begin{aligned}
   \boldsymbol{\Upsilon}_{\! 1}\left(\mathbf{x},\theta,v\right)=\boldsymbol{\Upsilon}_{\! 1}^{s}\left(\mathbf{x},\theta,v\right)+\boldsymbol{\Upsilon}_{\! 1}^{b}\left(\mathbf{x},\theta,v\right),
\end{aligned}
\end{gather}
where 
\begin{gather}
\begin{aligned}
      &\boldsymbol{\Upsilon}_{\! 1}^{s}\left(\mathbf{x},\theta,v\right)=x_{1}-\eps v\frac{\cos\left(\theta\right)}{B\left(\mathbf{x}\right)},
      \\
      &\boldsymbol{\Upsilon}_{\! 1}^{b}\left(\mathbf{x},\theta,v\right)=-\eps^{2}\cos\left(\theta\right)\int_{0}^{v}\left(v-u\right)\left(\boldsymbol{\Lambda}\cdot\frac{1}{B}\right)\left(\mathcal{G}_{-\eps u}\left(\mathbf{x},\theta\right)\right)du.
\end{aligned}
\end{gather}
For any $\left(\mathbf{x},\theta\right)\in\mathbb{R}^{2}\times\rit$, for any $v\in\left[a,b\right]$ and for any $\eps\in\mathbb{R}$ we have:
\begin{align}
   &\left|\boldsymbol{\Upsilon}_{\! 1}^{b}\left(\mathbf{x},\theta,v\right)\right|\leq\frac{\eps^{2}b^{2}}{2}\left\Vert \boldsymbol{\Lambda}\cdot\frac{1}{B}\right\Vert _{\infty},
   \label{BoundedPartOfBoldUpsPetitY1Est}   
\end{align}
and consequently for any $\left(\mathbf{x},\theta\right)\in\mathbb{R}^{2}\times\rit$, for any $v\in\left[a,b\right],$ for any 
$\eps\in\mathbb{R}^\star$ and for any $t\in\mathbb{R}$
\begin{align}
    &\left|\boldsymbol{\Upsilon}_{\! 1}^{b}\left(\mathbf{X}_{\boldsymbol{\mathfrak{Pol}}}^{\eps}\left(t;\mathbf{x},\theta,v\right),\Theta^{\eps}\left(t,\mathbf{x},\theta,v\right),v\right)\right|\leq\frac{\eps^{2}b^{2}}{2}\left\Vert \boldsymbol{\Lambda}\cdot\frac{1}{B}\right\Vert _{\infty}.
    \label{BoundedPartOfBoldUpsY1Est}   
\end{align}

On another hand, evaluating $\boldsymbol{\Upsilon}_{\! 1}^{s}$ in $\left(\mathbf{X}_{\boldsymbol{\mathfrak{Pol}}}^{\eps}\left(t;\mathbf{x},\theta,v\right),\Theta^{\eps}\left(t;\mathbf{x},\theta,v\right),v\right)$ and differentiating with respect to $t$ yields:
\begin{gather}
\begin{aligned}
      \frac{\partial}{\partial t}\left(\boldsymbol{\Upsilon}^s_{\! 1}\left(\mathbf{X}_{\boldsymbol{\mathfrak{Pol}}}^{\eps},\Theta^{\eps},v\right)\right)&=\eps v^{2}\cos\left(\Theta^{\eps}\right)\frac{\hat{\mathbf{c}}\left(\Theta^{\eps}\right)\cdot\nabla_{\mathbf{x}}B\left(\mathbf{X}_{\boldsymbol{\mathfrak{Pol}}}^{\eps}\right)}{B\left(\mathbf{X}_{\boldsymbol{\mathfrak{Pol}}}^{\eps}\right)^{2}},
\end{aligned}
\end{gather}
where 
\begin{gather}	
	\hat{\mathbf{c}}(\theta) = \begin{pmatrix} -\sin(\theta)\\ -\cos(\theta)\end{pmatrix},
\end{gather}
and consequently 
\begin{align}
    &\left|\frac{\partial}{\partial t}\left(\boldsymbol{\Upsilon}_{\! 1}^{s}\left(\mathbf{X}_{\boldsymbol{\mathfrak{Pol}}}^{\eps},\Theta^{\eps},v\right)\right)\right|
\leq\left|\eps\right| b^{2}\underset{\left(\mathbf{x},\theta\right)\in\mathbb{R}^{3}}{\sup}\left|\frac{\hat{\mathbf{c}}\left(\theta\right)\cdot\nabla_{\mathbf{x}}B\left(\mathbf{x}\right)}{B\left(\mathbf{x}\right)^{2}}\right|.
     \label{SlowPartDerOfBoldUpsY1Est}   
\end{align}
This ends the proof of Lemma \ref{ThmainresultCharDarbYDarb}.
{~\hfill $\square$} 

The proof of Lemma \ref{ThmainresultCharDarbYDarb2} is obvious.

\subsection{Proof of Theorem \ref{MainThm1} and Remark \ref{DarbouxRangeInCompact}}
\label{Subsection4Februar2014}  
Theorem  \ref{MainThm1} and remark \ref{DarbouxRangeInCompact} are a synthesis of 
Theorems \ref{ThmInvUpsIsKappa} and  \ref{ExpPoissMat} and of Lemmas  \ref{ThmainresultCharDarbKDarb},
\ref{ThmainresultCharDarbYDarb} and \ref{ThmainresultCharDarbYDarb2}.
{~\hfill $\square$} 
\section{The Partial Lie Transform Method} 
\label{ThePartialLieTransformMethod}    

The last step on the way to build the Guiding-Center Coordinates of order $N$ is to build a coordinate system
$(\mathbf{z},\gamma,j)$ close to the Historical Guiding-Center coordinate system in which the Poisson Matrix and the Hamiltonian function are given by \eqref{201401230906} and \eqref{Formula6Februar2014}. To this aim we will construct a new algorithm, the so-called Partial Lie Transform Method.

\begin{remark}
\label{RemRevision3}  
In {\cite{littlejohn:1979}}, to build the Guiding-Center coordinate system, Littlejohn construct a normal form theory based on formal
Lie series using Hamiltonian vector fields.
The drawback of using such a formal Lie Series method is that its convergence is neither ensured nor controlled.
\end{remark}

%
\subsection{The Partial Lie Change of Coordinates of order $N$} 
\label{ThePartialLieSumsSubSection}   

We start this Section by defining the partial Lie sums. Let $N\in\mathbb{N}^*$. For $i\in\ldbrack1,N\rdbrack$, we define the positive integer $\alpha_{i,N}$ by 
\begin{align}
         &\alpha_{i,N}=\mathbb{E}\left(\frac{N}{i}\right)+1,   
         \label{DefCoefPLTCC}   
\end{align}
where $\mathbb{E}$ stands for the integer part.

\begin{definition}
\label{DefBoldVarijN2}  
For any $\bar{g}=\bar{g}\left(\mathbf{y},\theta,k\right)$ in $\mathcal{C}^\infty_\#(\rit^2\times\rit\times(0,+\infty))$ (see Notation \ref{201402161039}), 
let $\boldsymbol{\vartheta}_{\eps,-\bar{g}}^{\alpha_{i,N},i}$ be the differential 
operator acting on functions $\bar{f}=\bar{f}\left(\mathbf{y},\theta,k\right)$ of $\mathcal{C}^\infty_\#(\rit^2\times\rit\times(0,+\infty))$ in the following way:
\begin{align}
     &\boldsymbol{\vartheta}_{\eps,-\bar{g}}^{\alpha_{i,N},i}\cdot\bar{f}=\underset{k=0}{\overset{\alpha_{i,N}}{\sum}}\frac{\eps^{ik}}{k!}\left(\overline{\mathbf{X}}_{-\eps\bar{g}}^{\eps}\right)^{k}\cdot\bar{f},
     \label{DefVarthetaActFuncDifOp}   
\end{align}
where $\overline{\mathbf{X}}_{-\eps\bar{g}}^{\eps}$ is the Hamiltonian vector field associated with $-\eps\bar{g}$.
From operator $\boldsymbol{\vartheta}_{\eps,-\bar{g}}^{\alpha_{i,N},i}$ we define, with the same notation, function 
$\boldsymbol{\vartheta}_{\eps,-\bar{g}}^{\alpha_{i,N},i}=\boldsymbol{\vartheta}_{\eps,-\bar{g}}^{\alpha_{i,N},i}\!\left(\mathbf{y},\theta,k\right)$ from $\mathbb{R}^2\times\mathbb{R}\times (0,+\infty)$ to $\rit^4$ by
\begin{align}
     &\boldsymbol{\vartheta}_{\eps,-\bar{g}}^{\alpha_{i,N},i}=\left(\left(\boldsymbol{\vartheta}_{\eps,-\bar{g}}^{\alpha_{i,N},i}\cdot\mathbf{y}_{1}\right),\ldots,\left(\boldsymbol{\vartheta}_{\eps,-\bar{g}}^{\alpha_{i,N},i}\cdot\mathbf{k}\right)\right),
     \label{DefVarThetN2}    
\end{align}
where $\mathbf{y}_1$, $\mathbf{y}_2$, $\boldsymbol{\theta}$, $\mathbf{k}$ stand for $\mathbf{y}_{1}:\left(\mathbf{y},\theta,k\right)\mapsto y_{1}$,
$\ldots$, $\mathbf{k}:\left(\mathbf{y},\theta,k\right)\mapsto k$.
\end{definition}

\begin{definition}
\label{1305210335} 
$\boldsymbol{\vartheta}_{\eps,-\bar{g}}^{\alpha_{i,N},i}$ is called the Partial Lie Sum of order 
$\left(i,N\right)$ generated by $\bar{g}$.
\end{definition}


\begin{theorem}
\label{BoldChiNWellDef}   
Let $\bar{g}_{1},\ldots,\bar{g}_N\in\mathcal{Q}_{T,b}^\infty$  (see Notation \ref{201402161042}) and
$c$ and $d$ be positive real numbers (with $c<d$).
Then there exists $\eta>0$ such that for any $\eps\in[-\eta,\eta]$, $\boldsymbol{\chi}^{N}_\eps$, defined by
\begin{align}
         &\boldsymbol{\chi}^{N}_\eps=\boldsymbol{\vartheta}_{\eps,-\bar{g}_{1}}^{\alpha_{1,N},1}\circ\boldsymbol{\vartheta}_{\eps,-\bar{g}_{2}}^{\alpha_{2,N},2}\circ\ldots\circ\boldsymbol{\vartheta}_{\eps,-\bar{g}_{N}}^{\alpha_{N,N},N},
         \label{DefChiN}    
\end{align}
is well defined on $\mathbb{R}^3\times(c,d)$ and is a diffeomorphism.
Moreover, for any intervals  $(c^\star,d^\star)$ and $(c^\bullet,d^\bullet)$
such that $c^\star>0$ and
\begin{align}
    &[c^{\bullet},d^{\bullet}]\subsetneq(c,d)\subsetneq[c,d]\subsetneq(c^{\star},d^{\star})
\end{align}
 there exists a real number $\eta^{\bullet,\star}>0$ such that for any $\eps\in\left[-\eta^{\bullet,\star},\eta^{\bullet,\star}\right]$:
 \begin{align}
    &\mathbb{R}^3\times (c^{\bullet},d^{\bullet})\subset\boldsymbol{\chi}_{\eps}^{N}\big(\mathbb{R}^3\times(c,d)\big)\subset\mathbb{R}^3\times (c^{\star},d^{\star}).
     \label{DomInvWellDef}   
 \end{align}
\end{theorem}
The proof of Theorem \ref{BoldChiNWellDef} is given in subsection \ref{SubsectionProofOfTheoremBoldChiNWellDef}.
\begin{definition}
With assumptions of Theorem \ref{BoldChiNWellDef} on the $\bar{g}_i$, $\boldsymbol{\chi}_{\eps}^{N}$ is called the \textit{partial Lie change of coordinates of order $N$.}
We denote by $\boldsymbol{\lambda}^N_{\eps}$ the inverse function of $\boldsymbol{\chi}_{\eps}^{N}$. 
\end{definition}

\begin{remark}
An immediate Corollary to Theorem \ref{BoldChiNWellDef} is that for $\eps$ small enough $\boldsymbol{\lambda}^N_{\eps}$ is well defined on $\mathbb{R}^3\times (c^{\bullet},d^{\bullet})$.
\end{remark}

\subsection{Main properties of the partial Lie change of coordinates of order $N$} 
\label{MainPropOfThePartialLieCC}    

The main properties of the partial Lie change of coordinates of order $N$ are summarized in the following Theorem.

\begin{theorem}
\label{MainPropOfLieCCOrderNThm}    
With the same notations and under the same assumptions as in Theorem \ref{BoldChiNWellDef},  
assuming moreover that $\bar{g}_{1},\ldots,\bar{g}_N\in\mathcal{A}(\mathbb{R}^3\times(0,+\infty))$ (see Notation \ref{201402161044}),
for any compact set $\boldsymbol{K}\subset\mathbb{R}^2$
and for any interval $\left[c^{\diamondsuit},d^{\diamondsuit}\right]$ such that  $\left[c^{\diamondsuit},d^{\diamondsuit}\right]\subset(c^\bullet,d^\bullet)$, there exists a real number $\eta_K>0$ such that for any $\eps\in[0,\eta_K]$ and for any $(\mathbf{z},\gamma,j)\in\boldsymbol{K}\times\rit\times\left[c^{\diamondsuit},d^{\diamondsuit}\right]$,  the inverse function $\boldsymbol{\lambda}^N_{\eps}$ of $\boldsymbol{\chi}_{\eps}^{N}$ 
has the following expression:
\begin{align}
      &\boldsymbol{\lambda}_{\eps}^{N}\!\!\left(\mathbf{z},\gamma,j\right)=\boldsymbol{\vartheta}_{\eps,\bar{g}_{1}}^{\alpha_{1,N},1}\cdot\boldsymbol{\vartheta}_{\eps,\bar{g}_{2}}^{\alpha_{2,N},2}\cdot\ldots\cdot\boldsymbol{\vartheta}_{\eps,\bar{g}_{N}}^{\alpha_{N,N},N}\!\!\left(\mathbf{z},\gamma,j\right)+\eps^{N+1}\Rest_{\boldsymbol{\lambda}}^{N}\!\left(\eps;\mathbf{z},\gamma,j\right).
       \label{DefLamExpPartLieCCOrderN}  
\end{align}
Moreover, on $\left[0,\eta_K\right]\times\boldsymbol{K}\times\rit\times\left[c^{\diamondsuit},d^{\diamondsuit}\right]$,
 we have the following expressions of the Hamiltonian function $\hat{H}_{\eps}$ and the Poisson Matrix  $\hat{\mathcal{P}}_{\eps}$ in the 
$(\mathbf{z},\gamma,j)$-coordinate system:
\begin{align}
       &\hat{H}_{\eps}\!\left(\mathbf{z},\gamma,j\right)=\boldsymbol{\vartheta}_{\eps,\bar{g}_{1}}^{\alpha_{1,N},1}\cdot\boldsymbol{\vartheta}_{\eps,\bar{g}_{2}}^{\alpha_{2,N},2}\cdot\ldots\cdot\boldsymbol{\vartheta}_{\eps,\bar{g}_{N}}^{\alpha_{N,N},N}\cdot\bar{H}_{\eps}\!\left(\mathbf{z},\gamma,j\right)
       +\eps^{N+1}\Rest_{\bar{H}}^{N}\!\left(\eps;\mathbf{z},\gamma,j\right),
       \label{DefHatHExpPartLieCCOrderN}    
       \\
       &\eps\hat{\mathcal{P}}_{\eps}\left(\mathbf{z},\gamma,j\right)=\eps\bar{\mathcal{P}}_{\eps}\left(\mathbf{z},\gamma,j\right)+\eps^{N+2}\Rest_{\bar{\mathcal{P}}}^{N}\!\left(\eps;\mathbf{z},\gamma,j\right),
        \label{DefHatMatPoissExpPartLieCCOrderN}    
\end{align}
where $\bar{H}_{\eps}$ is given by  \eqref{HamDarb} or \eqref{DsDim423},
$\bar{\mathcal{P}}_{\eps}$ by \eqref{ExpressionDarbouxMatrix}
and where $\Rest^N_{\boldsymbol{\lambda}}$, $\Rest_{\bar{H}}^{N}$, and $\Rest_{\bar{\mathcal{P}}}^{N}$ are in 
$\mathcal{C}_{\PerND}^{\infty}(\left[0,\eta_K\right]\times\boldsymbol{K}\times\rit\times\left[c^{\diamondsuit},d^{\diamondsuit}\right])$.
\end{theorem}
The proof of Theorem \ref{MainPropOfLieCCOrderNThm} is given in subsection \ref{SubsectionProofOfTheoremMainPropOfLieCCOrderNThm}.
%
\subsection{The Partial Lie Change of Coordinates Algorithm} 
\label{ThePartialLieCCAlgoSubSection}    
In this Section, we will deduce from Formula \eqref{DefHatHExpPartLieCCOrderN} the Partial Lie Change of Coordinates Algorithm.
\begin{theorem}
\label{ExpForTheAlgoWellPosed}   
With the same notations and under the same assumptions as in Theorems \ref{BoldChiNWellDef} and \ref{MainPropOfLieCCOrderNThm},
from formula \eqref{DefHatHExpPartLieCCOrderN} we have
\begin{align}
   &\hat{H}_{\eps}\!\left(\mathbf{z},\gamma,j\right)
   =\hat{H}_{0}\left(\mathbf{z},j\right)+\eps\hat{H}_{1}\left(\mathbf{z},\gamma,j\right)+\ldots+\eps^{N}\hat{H}_{N}\left(\mathbf{z},\gamma,j\right)
   +\eps^{N+1}\Rest_{\bar{H}}^{N}\!\left(\eps;\mathbf{z},\gamma,j\right),
\end{align} 
with $\Rest_{\bar{H}}^{N}$ in 
$\mathcal{C}_{\PerND}^{\infty}(\left[0,\eta_K\right]\times\boldsymbol{K}\times\rit\times\left[c^{\diamondsuit},d^{\diamondsuit}\right])$
and
\begin{align}
     \hat{H}_{0}\left(\mathbf{z},j\right)&=\bar{H}_{0}\left(\mathbf{z},j\right),
     \label{LieTransMethodIntroZerothEq2}    
     \\
     \hat{H}_{1}\left(\mathbf{z},\gamma,j\right)&=-B\left(\mathbf{z}\right)\!\frac{\partial\bar{g}_{1}}{\partial\theta}\!\left(\mathbf{z},\gamma,j\right)-\bar{H}_{1}\left(\mathbf{z},\gamma,j\right),
     \label{LieTransMethodIntroFirstEq}    
     \\
     \hat{H}_{2}\left(\mathbf{z},\gamma,j\right)&=-B\left(\mathbf{z}\right)\frac{\partial\bar{g}_{2}}{\partial\theta}\!\left(\mathbf{z},\gamma,j\right)-\mathcal{V}_{2}\!\left(\bar{g}_{1}\right)\left(\mathbf{z},\gamma,j\right),
      \label{LieTransMethodIntroSecondEq2}    
     \\
     &\vdots
     \notag
     \\
     \hat{H}_{N}\left(\mathbf{z},\gamma,j\right)&=-B\left(\mathbf{z}\right)\frac{\partial\bar{g}_{N}}{\partial\theta}\!\left(\mathbf{z},\gamma,j\right)-\mathcal{V}_{N}\!\left(\bar{g}_{1},\ldots,\bar{g}_{N-1}\right)\left(\mathbf{z},\gamma,j\right),
     \label{LieTransMethodIntroLastEq2}    
\end{align}
where for each $i\in\ldbrack1,N\rdbrack$, $\mathcal{V}_{i}\!\left(\bar{g}_{1},\ldots,\bar{g}_{i-1}\right)$
only depends on $\bar{g}_{1},\ldots,\bar{g}_{i-1}$, $\bar{H}_{0},\ldots,\bar{H}_{i}$ and their derivatives.
\end{theorem}
The proof of Theorem \ref{ExpForTheAlgoWellPosed} is given in subsection \ref{SectionProofOfTheoremExpForTheAlgoWellPosed}. \\

From Theorem \ref{ExpForTheAlgoWellPosed} we construct the following inductive Algorithm to determine 
$\bar{g}_{1},\ldots,\bar{g}_{N}$,  $\hat{H}_{1},\ldots,\hat{H}_{N}$ and consequently the partial Lie change of coordinates of order $N$.

\begin{algorithm}
\label{AlgorithmLieRevisedVersion}   
~
Set 
\begin{gather}
\hat{H}_{1}\left(\mathbf{z},\gamma,j\right)=-\frac{1}{2\pi}\int_{0}^{2\pi}  \bar{H}_{1}\left(\mathbf{z},\gamma,j\right)\,d\gamma,
\end{gather}
and get $\bar g_1$ by solving
\begin{gather}
\left \{
\begin{aligned}
       -B\left(\mathbf{z}\right)\!\frac{\partial\bar{g}_{1}}{\partial\gamma}\!\left(\mathbf{z},\gamma,j\right)&=\bar{H}_{1}\left(\mathbf{z},\gamma,j\right)-\frac{1}{2\pi}\int_{0}^{2\pi}\bar{H}_{1}\left(\mathbf{z},\gamma,j\right)\, d\gamma.
       \\
       \bar{g}_{1}\left(\mathbf{z},0,j\right)&=0.
\end{aligned}
\right.
\end{gather}
Then for $i\in\ldbrack1,N\rdbrack$, set
\begin{gather}
\begin{aligned}
     &\hat{H}_{i}\left(\mathbf{z},\gamma,j\right)=-\frac{1}{2\pi}\int_{0}^{2\pi}\left[\mathcal{V}_{i}\left(\bar{g}_{1},\ldots,\bar{g}_{i-1}\right)\right]\!\left(\mathbf{z},\gamma,j\right)d\gamma,
\end{aligned}
\end{gather}
and get $\bar g_i$ and by solving:
\begin{gather}
\left \{
\begin{aligned}
     -B\left(\mathbf{z}\right)\frac{\partial\bar{g}_{i}}{\partial\gamma}\left(\mathbf{z},\gamma,j\right)&=\mathcal{V}_{i}\!\left(\bar{g}_{1},\ldots,\bar{g}_{i-1}\right)\left(\mathbf{z},\gamma,j\right)-\frac{1}{2\pi}\int_{0}^{2\pi}\left[\mathcal{V}_{i}\left(\bar{g}_{1},\ldots,\bar{g}_{i-1}\right)\right]\!\left(\mathbf{z},\gamma,j\right)d\gamma.
     \\
     \bar{g}_{i}\left(\mathbf{z},0,j\right)&=0.
\end{aligned}
\right.
\end{gather}
\end{algorithm}
By construction functions $\bar{g}_{1},\ldots,\bar{g}_{N}$ and  $\hat{H}_{1},\ldots,\hat{H}_{N}$ obtained by applying Algorithm \ref{AlgorithmLieRevisedVersion}
satisfy the following Theorem. 
\begin{theorem}
\label{AlgoLeadsToFuncWellDef}   
Let $\bar{g}_{1},\ldots,\bar{g}_{N}$ and  $\hat{H}_{1},\ldots,\hat{H}_{N}$ be constructed by applying Algorithm \ref{AlgorithmLieRevisedVersion}.
Then, $\bar{g}_{1},\ldots,\bar{g}_{N}\in\mathcal{A}(\mathbb{R}^3\times(0,+\infty))\cap\mathcal{Q}_{T,b}^\infty$,
$\hat{H}_{1},\ldots,\hat{H}_{N}\in\mathcal{C}_{\#}^{\infty}\!\left(\mathbb{R}^{3}\times\left(0,+\infty\right)\right)$,
and for each $i\in\ldbrack1,N\rdbrack$, $\hat{H}_i$ does not depend on $\gamma$.
\end{theorem}
%
\subsection{Proof of Theorem \ref{BoldChiNWellDef}}
\label{SubsectionProofOfTheoremBoldChiNWellDef}   
The first step to prove Theorem \ref{BoldChiNWellDef} consists in proving that the partial Lie sums are diffeomorphisms and to localize
their ranges.
\begin{theorem}
\label{FirstStepProofThmBoldChiNWellDef}  
Let $i\in\ldbrack1,N\rdbrack$, $\bar{g}_i\in\mathcal{Q}_{T,b}^\infty$ and $c$ and $d$ be positive real numbers (with $c<d$).
Then there exists $\eta>0$ such that for any $\eps\in[-\eta,\eta]$,
function $\boldsymbol{\vartheta}_{\eps,-\bar{g}_i}^{\alpha_{i,N},i}$, defined by \eqref{DefVarThetN2}, is a diffeomorphism from $\mathbb{R}^3\times(c,d)$ onto its range.
Moreover, for any interval  $(c^\star,d^\star)$ and $(c^\bullet,d^\bullet)$
such that $c^\star>0$ and
\begin{align}
    &[c^{\bullet},d^{\bullet}]\subsetneq(c,d)\subsetneq[c,d]\subsetneq(c^{\star},d^{\star})
\end{align}
 there exists a real number $\eta^{\bullet,\star}>0$ such that for any $\eps\in\left[-\eta^{\bullet,\star},\eta^{\bullet,\star}\right]$:
 \begin{align}
    &\mathbb{R}^3\times (c^{\bullet},d^{\bullet})\subset\boldsymbol{\vartheta}_{\eps,-\bar{g}_i}^{\alpha_{i,N},i}\big(\mathbb{R}^3\times(c,d)\big)\subset\mathbb{R}^3\times (c^{\star},d^{\star}).
      \label{FirstStepProofThmBoldChiNWellDefInclusionRange}   
 \end{align}
\end{theorem}
\hspace{-0.5cm}Subsequently we will denote by 
\begin{gather}
	\label{201402161137} 
	\boldsymbol{\Xi}_{\eps,-\bar{g}_i}^{\alpha_{i,N},i} \text{ the inverse function of } \boldsymbol{\vartheta}_{\eps,-\bar{g}_i}^{\alpha_{i,N},i}.
\end{gather}
\begin{proof}
In a first place, we will show that $\boldsymbol{\vartheta}_{\eps,-\bar{g}_i}^{\alpha_{i,N},i}$ is a diffeomorphism from $\mathbb{R}^3\times(c,d)$ onto its range.
To this aim, we will check that there exists a real number $\bar{\eta}_1$ such that for any $\eps\in\left[-\bar{\eta}_{1},\bar{\eta}_{1}\right],$
the map $\boldsymbol{\vartheta}_{\eps,-\bar{g}_i}^{\alpha_{i,N},i}$ satisfies the assumptions of the classical
global inversion Theorem. 
\begin{remark}
This theorem claims that if $A$ is a continuous homeomorphism from a Banach space onto a normed vector space and if
$\phi$ is Lipschitz-continuous from the same Banach space onto the same normed vector space with a Lipschitz constant smaller than
$\|A^{-1}\|^{-1}$, then $A+\phi$ is invertible and its inverse map is Lipschitz-continuous.
\end{remark}

Function $\boldsymbol{\nu}_{\eps,-\bar{g}_i}^{\alpha_{i,N},i},$ defined as being such that
\begin{align}
     &\boldsymbol{\vartheta}_{\eps,-\bar{g}_i}^{\alpha_{i,N},i}=id+\eps\boldsymbol{\nu}_{\eps,-\bar{g}_i}^{\alpha_{i,N},i}
     \label{DefNuInTermsOfTheta}    
\end{align}
and whose expression, because of \eqref{DefVarthetaActFuncDifOp}, is given by
\begin{align}
\label{DefNuV2}    
     \boldsymbol{\nu}_{\eps,-\bar{g}_i}^{\alpha_{i,N},i}
     &=\left(\overset{\alpha_{i,N}}{\underset{j=1}{\sum}}\frac{\eps^{ij-1}}{j!}\left(\overline{\Xvec}_{-\eps\bar{g}_i}^{\eps}\right)^{j}\cdot\mathbf{y}_{1},\ldots,\overset{\alpha_{i,N}}{\underset{j=1}{\sum}}\frac{\eps^{ij-1}}{j!}\left(\overline{\Xvec}_{-\eps\bar{g}_i}^{\eps}\right)^{j}\cdot\mathbf{k}\right),
\end{align}
is differentiable and its differential is bounded on $\mathbb{R}^3\times[c,d]$. Moreover, $\eps\mapsto\boldsymbol{\nu}_{\eps,-\bar{g}_i}^{\alpha_{i,N},i}\left(\mathbf{y},\theta,k\right)$ is clearly
in $\mathcal{C}^\infty\!\left(\rit\right)$ for any $\left(\mathbf{y},\theta,k\right)\in\rit^3\times\left(0,+\infty\right)$. Hence, we can define 
\begin{align}
        &\left\Vert \boldsymbol{\nu}_{\eps,-\bar{g}_i}^{\alpha_{i,N},i}\right\Vert _{1,\infty}=\underset{\left(\mathbf{y},\theta,k\right)\in\mathbb{R}^{3}\times[c,d]}{\sup}\left|\left(d\boldsymbol{\nu}_{\eps,-\bar{g}_i}^{\alpha_{i,N},i}\right)_{\!\left(\mathbf{y},\theta,k\right)}\right|\negmedspace,
\end{align}
where function $\eps\mapsto\left\Vert \boldsymbol{\nu}_{\eps,-\bar{g}_i}^{\alpha_{i,N},i}\right\Vert _{1,\infty}$ is clearly in $\mathcal{C}^\infty\!\left(\rit\right)$. 
Now, since $\eps\left\Vert \boldsymbol{\nu}_{\eps,-\bar{g}_i}^{\alpha_{i,N},i}\right\Vert _{1,\infty}\!\!\rightarrow 0$  when ${\eps\rightarrow0}$,  there exists a real number $\eta'>0$ such that 
\begin{align}
     &\forall\eps\in\left[-\eta',\eta'\right],\ \left|\eps\left\Vert \boldsymbol{\nu}_{\eps,-\bar{g}_i}^{\alpha_{i,N},i}\right\Vert _{1,\infty}\right|<1.
\end{align}
Hence, we deduce that for $\eps$ small enough $\eps\boldsymbol{\nu}_{\eps,-\bar{g}_i}^{\alpha_{i,N},i}$ is Lipschitz continuous on $\mathbb{R}^3\times[c,d]$
and that its Lipschitz constant is smaller than $\left\Vert id^{-1}\right\Vert ^{-1}_\infty=1.$ 
Consequently \eqref{DefNuInTermsOfTheta}  and the global inversion Theorem imply that $\boldsymbol{\vartheta}_{\eps,-\bar{g}_i}^{\alpha_{i,N},i}$
is invertible and Lipschitz continuous.

The second step consists in checking that for any $\left(\mathbf{y},\theta,k\right)\in\rit^3\times\left[c,d\right]$ the differential
 \begin{align*}
      &\left(d\boldsymbol{\vartheta}_{\eps,-\bar{g}_i}^{\alpha_{i,N},i}\right)_{\left(\mathbf{y},\theta,k\right)}
\end{align*}
is an isomorphism. As
\begin{align}
       &\left(d\boldsymbol{\vartheta}_{\eps,-\bar{g}_i}^{\alpha_{i,N},i}\right)_{\negmedspace\left(\mathbf{y},\theta,k\right)}=id+\eps\left(d\boldsymbol{\nu}_{\eps,-\bar{g}_i}^{\alpha_{i,N},i}\right)_{\negmedspace\left(\mathbf{y},\theta,k\right)},
\end{align}
the Jacobian Matrix of $\boldsymbol{\vartheta}_{\eps,-\bar{g}_i}^{\alpha_{i,N},i}$ in $\left(\mathbf{y},\theta,k\right)\in\rit^3\times[c,d]$ can be rewritten as 
\begin{align}
    &\text{Jac}({\boldsymbol{\vartheta}_{\eps,-\bar{g}_i}^{\alpha_{i,N},i}})\left(\mathbf{y},\theta,k\right)=1+\eps\chi\!\left(\eps,\mathbf{y},\theta,k\right),
\end{align} 
where $\chi$ is bounded with respect to
$\left(\mathbf{y},\theta,k\right)\in\rit^3\times\left[c,d\right]$ and $\eps\mapsto\chi\!\left(\eps,\mathbf{y},\theta,k\right)$ is in $\mathcal{C}^\infty(\rit)$ for any $\left(\mathbf{y},\theta,k\right)\in\rit^3\times\left[c,d\right]$.
Hence, denoting $\Vert \chi\left(\eps,\cdot\right)\Vert _{\infty,\rit^{3}\times[c,d]} = \sup_{(\mathbf{y},\theta,k)\in\rit^{3}\times[c,d]} |\chi\left(\eps,\cdot\right)|$,
 there exists a real number $\eta''>0$ such that for any 
 $\eps\in[-\eta'',\eta''],\ \left|\eps\left\Vert \chi\left(\eps,\cdot\right)\right\Vert _{\infty,\rit^{3}\times[c,d]}\right|<1.$
Consequently, 
$(\mathbf{y},\theta,k)\mapsto(d\boldsymbol{\vartheta}_{\eps,-\bar{g}_i}^{\alpha_{i,N},i})_{\negmedspace\left(\mathbf{y},\theta,k\right)}$
is invertible and Lipschitz continuous.

Hence for $\left|\eps\right|<\bar{\eta}_1,$ where $\bar{\eta}_1=\min\left(\eta',\eta''\right),$ we can conclude that 
$\boldsymbol{\vartheta}_{\eps,-\bar{g}_i}^{\alpha_{i,N},i}$ is a diffeomorphism on $\rit^3\times[c,d]$. 
\\


The second part of the proof concerns inclusions \eqref{FirstStepProofThmBoldChiNWellDefInclusionRange}.
Using Formula \eqref{DefNuInTermsOfTheta} we obtain easily the second inclusion. Hence, we will focus on the first one.
Its proof is based on the Brouwer Theorem (see {Brouwer \cite{brouwer1912}} or {Istratescu \cite{Istratescu2001}}).

We fix two positive real numbers $R'_0$ and $R^\bullet_0$ such that $R^\bullet_0<R'_0$. 
Then we will fix $\boldsymbol{m}_0\in\mathbb{R}^2$ and we will show that there exists a positive real number $\eta$, that does not depend on 
$\boldsymbol{m}_0$, such that for any $\eps\in[-\eta,\eta]$
\begin{align}
     &\ball^{2}\left(\boldsymbol{m}_{0},R^{\bullet}_{0}\right)\times\rit\times(c^{\bullet},d^{\bullet})
     \subset\boldsymbol{\vartheta}_{\eps,-\bar{g}_{i}}^{\alpha_{i,N},i}\big(\ball^{2}\left(\boldsymbol{m}_{0},R'{}_{\negmedspace0}\right)\times\rit\times(c,d)\big),
     \label{17012014}  
\end{align}
or according to Notation  \ref{201402160959},
\begin{align}
     &\mathcal{CO}(\boldsymbol{m}_{0},R^{\bullet}_{0};c^{\bullet},d^{\bullet})
     \subset\boldsymbol{\vartheta}_{\eps,-\bar{g}_{i}}^{\alpha_{i,N},i}\big(\mathcal{CO}(\boldsymbol{m}_{0},R'{}_{\negmedspace0};c,d)\big).
     \label{201402161003}  
\end{align}
Consequently, since $\eta$ does not depend on $\boldsymbol{m}_0$ we will obtain  \eqref{FirstStepProofThmBoldChiNWellDefInclusionRange}.

Let $R_{0}^{\left(2\right)},\ R_{0}^{\left(3\right)},\ \alpha_{0}^{\left(2\right)}$, $\alpha_{0}^{\left(3\right)}$,
$c^{(2)}$, $c^{(3)}$, $d^{(2)}$ and $d^{(3)}$
be real numbers satisfying 
\begin{align}
     &R_{0}^\bullet<R_{0}^{\left(2\right)}<R_{0}^{\left(3\right)}<R'_{0},\hspace{0.5cm}0<\alpha_{0}^{\left(2\right)}<\alpha_{0}^{\left(3\right)},\hspace{0.5cm}\text{and }
     \notag
     \\
     &\left[c^{\bullet},d^{\bullet}\right]\subset\left(c^{(2)},d^{(2)}\right)\subset\left[c^{(2)},d^{(2)}\right]\subset\left(c^{(3)},d^{(3)}\right)\subset\left[c^{(3)},d^{(3)}\right]\subset\left(c,d\right),
     \notag
\end{align}
$l$ be an integer, and let
$\boldsymbol{\mathfrak{K}}^l_2$
and $\boldsymbol{\mathfrak{K}}^l_3$ be the compact 
and  convex subsets of $\mathbb{R}^2\times\rit\times(0,+\infty)$ defined by 
\begin{align}
       &\boldsymbol{\mathfrak{K}}_{2}^{l}=\overline{\ball^{2}\!(\boldsymbol{m}_{0},R_{0}^{(2)})}\times\left[\left(l-1\right)\pi-\alpha_{0}^{\left(2\right)},\left(l+1\right)\pi+\alpha_{0}^{\left(2\right)}\right]\times\left[c^{(2)},d^{(2)}\right]
\end{align}
and 
\begin{align}
     &\boldsymbol{\mathfrak{K}}_{3}^{l}=\overline{\ball^{2}\!(\boldsymbol{m}_{0},R_{0}^{(3)})}\times\left[\left(l-1\right)\pi-\alpha_{0}^{\left(3\right)},\left(l+1\right)\pi+\alpha_{0}^{\left(3\right)}\right]\times\left[c^{(3)},d^{(3)}\right].
\end{align}
Since $\eps\left\Vert \boldsymbol{\nu}_{\eps,-\bar{g}_i}^{\alpha_{i,N},i}\right\Vert _{\infty,\mathbb{R}^{3}\times(c,d)}\!\rightarrow0$ when ${\eps\rightarrow0}$, 
we can define $\eta>0$ (that depends neither on $l$ nor $\boldsymbol{m}_0$) such that for any $\eps\in[-\eta,\eta],$ for any $l\in\mathbb{Z},$ and for any $\left(\mathbf{y}',\theta',k'\right)\in\boldsymbol{\mathfrak{K}}_{2}^{l}$,
\begin{gather}
\begin{aligned}
      &\left|\mathbf{y}'-\boldsymbol{m}_{0}\right|+\left|\eps\right|\left\Vert \boldsymbol{\nu}_{\eps,-\bar{g}_i}^{\alpha_{i,N},i}\right\Vert _{\infty,\rit^{3}\times(c,d)}\leq R_{0}^{\left(3\right)},
      \\
      &\left|\theta'-l\pi\right|+\left|\eps\right|\left\Vert \boldsymbol{\nu}_{\eps,-\bar{g}_i}^{\alpha_{i,N},i}\right\Vert _{\infty,\rit^{3}\times(c,d)}\leq\alpha_{0}^{\left(3\right)},
      \\
      &k'\pm\left|\eps\right|\left\Vert \boldsymbol{\nu}_{\eps,-\bar{g}_i}^{\alpha_{i,N},i}\right\Vert _{\infty,\rit^{3}\times(c,d)}\in\big[c^{(3)},d^{(3)}\big].
\end{aligned}
\end{gather}

Now, for all $\left(\mathbf{y}',\theta',k'\right)\in\boldsymbol{\mathfrak{K}}_{2}^{l}$, we define the function 
$F_{\left(\mathbf{y}',\theta',k'\right)}^{\eps}$   by 
 \begin{align}
         &F_{\left(\mathbf{y}',\theta',k'\right)}^{\eps}:\ \boldsymbol{\mathfrak{K}}_{3}^{l}\rightarrow\mathbb{R}^{4};\ \left(\mathbf{y},\theta,k\right)\mapsto\left(\mathbf{y}',\theta',k'\right)-\eps\boldsymbol{\nu}_{\eps,-\bar{g}_i}^{\alpha_{i,N},i}\left(\mathbf{y},\theta,k\right).
 \end{align}
 By construction and because of the properties of $\boldsymbol{\nu}_{\eps,-\bar{g}_i}^{\alpha_{i,N},i}$, $F_{\left(\mathbf{y}',\theta',k'\right)}^{\eps}$ is continuous on 
 $\boldsymbol{\mathfrak{K}}^l_3$ and for any $\eps\in[-\eta,\eta]$ and any $\left(\mathbf{y},\theta,k\right)\in\boldsymbol{\mathfrak{K}}_{3}^{l}$,
 \begin{gather}
 \begin{aligned}
        &\left|\left(F_{\left(\mathbf{y}',\theta',k'\right)}^{\eps}\negmedspace\left(\mathbf{y},\theta,k\right)\right)_{1,2}-\boldsymbol{m}_{0}\right|\leq\left|\mathbf{y}'-\boldsymbol{m}_{0}\right|+\left|\eps\right|\left|\left(\boldsymbol{\nu}_{\eps,-\bar{g}_i}^{\alpha_{i,N},i}\left(\mathbf{y},\theta,k\right)\right)_{1,2}\right|\leq R_{0}^{\left(3\right)},
        \\
        &\left|\left(F_{\left(\mathbf{y}',\theta',k'\right)}^{\eps}\negmedspace\left(\mathbf{y},\theta,k\right)\right)_{3}-l\pi\right|\leq\left|\theta'-l\pi\right|+\left|\eps\right|\left|\left(\boldsymbol{\nu}_{\eps,-\bar{g}_i}^{\alpha_{i,N},i}\left(\mathbf{y},\theta,k\right)\right)_{3}\right|\leq\alpha_{0}^{\left(3\right)},
 \end{aligned}
 \end{gather}
 and 
 \begin{align}
      &k'-\left(\eps\boldsymbol{\nu}_{\eps,-\bar{g}_i}^{\alpha_{i,N},i}\left(\mathbf{y},\theta,k\right)\right)_{4}\in\left[c^{(3)},d^{(3)}\right],
 \end{align}
meaning $F_{\left(\mathbf{y}',\theta',k'\right)}^{\eps}\negmedspace\left(\boldsymbol{\mathfrak{K}}_{3}^{l}\right)\subset\boldsymbol{\mathfrak{K}}_{3}^{l}$. Hence, invoking the Brouwer Theorem and more precisely its convex compact version, function $F_{\left(\mathbf{y}',\theta',k'\right)}^{\eps}$ has a fixed point in $\boldsymbol{\mathfrak{K}}_{3}^{l}$.
So we have proven that 
\begin{gather}
\exists\eta>0,\ \forall\eps,\left|\eps\right|<\eta,\ \forall l\in\mathbb{Z},\ \forall\left(\mathbf{y}',\theta',k'\right)\in\boldsymbol{\mathfrak{K}}_{2}^{l},\ \exists\left(\mathbf{y},\theta,k\right)\in\boldsymbol{\mathfrak{K}}_{3}^{l},\ \boldsymbol{\vartheta}_{\eps,-\bar{g}_i}^{\alpha_{i,N},i}\left(\mathbf{y},\theta,k\right)=\left(\mathbf{y}',\theta',k'\right),
\notag
\end{gather}
meaning that $\boldsymbol{\vartheta}_{\eps,-\bar{g}_i}^{\alpha_{i,N},i}\!\left(\boldsymbol{\mathfrak{K}}_{3}^{l}\right)\supset\boldsymbol{\mathfrak{K}}_{2}^{l}$ and consequently, since $\eta$ does not depend on $l,$ that
\begin{align}
    &\boldsymbol{\vartheta}_{\eps,-\bar{g}_i}^{\alpha_{i,N},i}\!\left(\underset{l\in\mathbb{Z}}{\cup}\boldsymbol{\mathfrak{K}}_{3}^{l}\right)\supset\underset{l\in\mathbb{Z}}{\cup}\boldsymbol{\mathfrak{K}}_{2}^{l}.
    \label{TwoBallDecPerSubsetCup}    
\end{align}
Since (see Notation  \ref{201402160959}) 
$\overline{\mathcal{CO}(\boldsymbol{m}_{0},R_{0}^{(3)};c^{(3)},d^{(3)})}=\underset{l\in\mathbb{Z}}{\cup}\boldsymbol{\mathfrak{K}}_{3}^{l}$
and $\overline{\mathcal{CO}(\boldsymbol{m}_{0},R_{0}^{(2)};c^{(2)},d^{(2)})}=\underset{l\in\mathbb{Z}}{\cup}\boldsymbol{\mathfrak{K}}_{2}^{l},$
\eqref{TwoBallDecPerSubsetCup} can be rewritten as $\boldsymbol{\vartheta}_{\eps,-\bar{g}_{i}}^{\alpha_{i,N},i}\!\left(\overline{\mathcal{CO}(\boldsymbol{m}_{0},R_{0}^{(3)};c^{(3)},d^{(3)})}\right)\supset\overline{\mathcal{CO}(\boldsymbol{m}_{0},R_{0}^{(2)};c^{(2)},d^{(2)})}.$
Finally, since 
\begin{align}
    &\mathcal{CO}(\boldsymbol{m}_{0},R_{0}^{\bullet};c^{\bullet},d^{\bullet})\subset\overline{\mathcal{CO}(\boldsymbol{m}_{0},R_{0}^{(2)};c^{(2)},d^{(2)})}\subset\overline{\mathcal{CO}(\boldsymbol{m}_{0},R_{0}^{(3)};c^{(3)},d^{(3)})}\subset\mathcal{CO}(\boldsymbol{m}_{0},R_{0}';c,d)
    \notag
\end{align}
we obtain the inclusion in \eqref{17012014}. 

Eventually, since $\eta$ does not depend on $\boldsymbol{m}_{0}$ we obtain the first inclusion of \eqref{FirstStepProofThmBoldChiNWellDefInclusionRange},
ending the proof of the theorem
\end{proof}

\begin{proof}[Proof of Theorem \ref{BoldChiNWellDef}]
Eventually, a straightforward induction using essentially Theorem 
\ref{FirstStepProofThmBoldChiNWellDef}, leads to Theorem \ref{BoldChiNWellDef}.
\end{proof}

\begin{remark}
Notice that $\boldsymbol{\lambda}_{\eps}^{N}$ is given by
\begin{align}
   &\boldsymbol{\lambda}_{\eps}^{N}=\boldsymbol{\Xi}_{\eps,-\bar{g}_{N}}^{\alpha_{N,N},N}\circ\ldots\circ\boldsymbol{\Xi}_{\eps,-\bar{g}_{1}}^{\alpha_{1,N},1}.
\end{align}
\end{remark}

\subsection{Proof of Theorem \ref{MainPropOfLieCCOrderNThm}}
\label{SubsectionProofOfTheoremMainPropOfLieCCOrderNThm}   

We begin by giving and proving preliminary results that are needed for the proof of Theorem 
\ref{MainPropOfLieCCOrderNThm}. Its proof is then led in the last part of this subsection.

\begin{property}
\label{PropPartLieFuncProd}   
Let $i\in\ldbrack1,N\rdbrack$ and $f$, $\bar{g}_i$ and $h$ be three functions in $\mathcal{C}_{\#}^{\infty}\left(\rit^{3}\times\left(0,+\infty\right)\right)$.
Then, the following equalities hold true on $\mathbb{R}^3\times(0,+\infty)$:
\begin{align}
           &\left(\boldsymbol{\vartheta}_{\eps,-\bar{g}_{i}}^{\alpha_{i,N},i}\cdot\left(fh\right)\right)=\left(\boldsymbol{\vartheta}_{\eps,-\bar{g}}^{\alpha_{i,N},i}\cdot f\right)\left(\boldsymbol{\vartheta}_{\eps,-\bar{g}}^{\alpha_{i,N},i}\cdot h\right)+\eps^{N+1}\rest_{FP}^{N,i}\!\left(\eps;\ \cdot\ \right),
           \label{PartLieFuncProd}   
           \\
           &\left(\boldsymbol{\vartheta}_{\eps,-\bar{g}_{i}}^{\alpha_{i,N},i}\cdot\left\{ f,h\right\} _{\mathcal{D}}\right)=\left\{ \boldsymbol{\vartheta}_{\eps,-\bar{g}_{i}}^{\alpha_{i,N},i}\cdot f,\boldsymbol{\vartheta}_{\eps,-\bar{g}_{i}}^{\alpha_{i,N},i}\cdot h\right\} _{\mathcal{D}}+\eps^{N+1}\rest_{PC}^{N,i}\!\left(\eps,\ \cdot\ \right),
           \label{PartLiePoissCom}   
\end{align}
where $\rest^{N,i}_{FP}$ and $\rest^{N,i}_{PC}$ are in $\mathcal{C}_{\#}^{\infty}\left(\rit\times\rit^{3}\times\left(0,+\infty\right)\right)$.
\end{property}

\begin{proof}
The proofs of Formulas \eqref{PartLieFuncProd} and \eqref{PartLiePoissCom} are very similar.
Consequently, we will only give the proof of Formula \eqref{PartLiePoissCom}.

In a first place, starting from is the following equality
\begin{align}
       &\overline{\mathbf{X}}_{-\eps\bar{g}_{i}}^{\eps}\!\cdot\!\left\{ f,h\right\} _{\mathcal{D}}=\left\{ \overline{\mathbf{X}}_{-\eps\bar{g}_{i}}^{\eps}\cdot f,h\right\} _{\mathcal{D}}+\left\{ f,\overline{\mathbf{X}}_{-\eps\bar{g}_{i}}^{\eps}\cdot h\right\} _{\mathcal{D}},
\end{align}
which is a direct consequence of the Jacobi identity, it is obvious to show by induction that
\begin{align}
       &\left(\overline{\mathbf{X}}_{-\eps\bar{g}_{i}}^{\eps}\right)^{n}\!\!\cdot\!\left\{ f,h\right\} _{\mathcal{D}}=\sum_{k=0}^{n}C_{n}^{k}\left\{ \left(\overline{\mathbf{X}}_{-\eps\bar{g}_{i}}^{\eps}\right)^{k}\cdot f,\left(\overline{\mathbf{X}}_{-\eps\bar{g}_{i}}^{\eps}\right)^{n-k}\cdot h\right\} _{\mathcal{D}}.
  \label{SpecProp1}   
\end{align}

Secondly, we will define on $\mathbb{R}^3\times(0,+\infty)$ the function $\left\{ f,h\right\} ^{\!\!\bar{\mathcal{T}}_{\eps}}=\left\{ f,h\right\} ^{\!\!\bar{\mathcal{T}}_{\eps}}\!\left(\mathbf{y},\theta,k\right)$ by 
\begin{align}
       &\left\{ f,h\right\} ^{\!\!\bar{\mathcal{T}}_{\eps}}\!\left(\mathbf{y},\theta,k\right)=\left(\bar{\mathcal{T}}_{\eps}\!\left(\mathbf{y},\theta,k\right)\nabla h\!\left(\mathbf{y},\theta,k\right)\right)\cdot\left(\nabla f\!\left(\mathbf{y},\theta,k\right)\right),
\end{align}
where 
\begin{align}
    &\bar{\mathcal{T}}_{\eps}=\eps\bar{\mathcal{P}}_{\eps},
\end{align}
and notice that $\eps\mapsto \left\{ f,h\right\} ^{\!\!\bar{\mathcal{T}}_\eps}\!\left(\mathbf{y},\theta,k\right)$ is in 
$\mathcal{C}^\infty\!\left(\rit\right)$ for any $\left(\mathbf{y},\theta,k\right)\in\mathbb{R}^3\times(0,+\infty)$.
\\
Hence, expanding $\boldsymbol{\vartheta}_{\eps,-\bar{g}_i}^{\alpha_{i,N},i}\cdot\left\{ f,h\right\} $ using Formula \eqref{SpecProp1}, expanding
$\left\{ \boldsymbol{\vartheta}_{\eps,-\bar{g}_i}^{\alpha_{i,N},i}\cdot f,\boldsymbol{\vartheta}_{\eps,-\bar{g}_i}^{\alpha_{i,N},i}\cdot h\right\} ,$ and making the difference between these two expansions yields \eqref{PartLiePoissCom}
with 
\begin{gather}
 \label{FormulaBigRestInPrpertyPoissBracketPointCirc}   
\begin{aligned}
      &\rest_{PC}^{N,i}\left(\eps,\ \cdot\ \right)=
      \\
      &-\negmedspace\overset{2\alpha_{i,N}}{\underset{k=\alpha_{i,N}+1}{\sum}}\eps^{ik-\left(N+2\right)}\negmedspace\negmedspace\negmedspace\negmedspace\negmedspace\negmedspace\!\underset{\left(m,p\right)\in\ldbrack1,N\rdbrack^{2}\ s.t.\ m+p=k}{\sum}\frac{1}{m!p!}\left\{ \left(\overline{\mathbf{X}}_{-\eps\bar{g}_{i}}^{\eps}\right)^{m}\cdot f,\left(\overline{\mathbf{X}}_{-\eps\bar{g}_{i}}^{\eps}\right)^{p}\cdot h\right\} ^{\!\!\bar{\mathcal{T}}_{\eps}}.
\end{aligned}
\end{gather}
As $i\alpha_{i,N}\geq N+1$, all $k\geq \alpha_{i,N}+1$ satisfy $ik\geq N+2$. Consequently, $\eps\mapsto \rest_{PC}^{N,i}\!\left(\eps;\mathbf{y},\theta,k\right)$ is in 
$\mathcal{C}^\infty\!\left(\rit\right)$ for any $\left(\mathbf{y},\theta,k\right)\in\rit^3\times(0,+\infty)$.
 In addition,  
$\left(\mathbf{y},\theta,k\right)\mapsto \rest_{PC}^{N,i,j}\!\left(\eps,\mathbf{y},\theta,k\right)$ is clearly in $\mathcal{C}_{\#}^{\infty}\left(\rit^{3}\times\left(0,+\infty\right)\right)$ for any $\eps\in\rit$. 
\end{proof}
\begin{remark}
\label{ExpOfTheRestInPartLieFuncProd}  
The expression of the rest in Formula \eqref{PartLiePoissCom} is given by
\begin{gather}
 \label{FormulaBigRestInPrpertyPartLieFuncProd}  
\begin{aligned}
      &\rest_{FP}^{N,i}\left(\eps,\ \cdot\ \right)=
      \\
      &-\negmedspace\eps\overset{2\alpha_{i,N}}{\underset{k=\alpha_{i,N}+1}{\sum}}\eps^{ik-\left(N+2\right)}\negmedspace\negmedspace\negmedspace\negmedspace\negmedspace\negmedspace\!\underset{\left(m,p\right)\in\ldbrack1,N\rdbrack^{2}\ s.t.\ m+p=k}{\sum}\frac{1}{m!p!}\left(\left(\overline{\mathbf{X}}_{-\eps\bar{g}_{i}}^{\eps}\right)^{m}\cdot f\right)\left(\left(\overline{\mathbf{X}}_{-\eps\bar{g}_{i}}^{\eps}\right)^{p}\cdot h\right).
\end{aligned}
\end{gather}
\end{remark}
\begin{theorem}
\label{ThemQuasiComRondPoint}          
With the same notations and under the same assumptions as in Theorem \ref{MainPropOfLieCCOrderNThm},
let $i\in\ldbrack1,N\rdbrack$ and $h_\eps$ be in $\mathcal{Q}_{T,b}^{\infty}\cap\mathcal{A}\!\left(\mathbb{R}^{2}\times\rit\times\left(0,+\infty\right)\right)$
for every $\eps$ in some interval $I$ containing $0$ and such that $\eps\mapsto h_\eps( \bar{\boldsymbol{r}})$ is in $C^\infty(I)$ for any 
$ \bar{\boldsymbol{r}}\in\mathbb{R}^3\times(0,+\infty)$. Then,
there exists a real number $\eta_K>0$
such that for any $\eps\in[-\eta_K,\eta_K]\cap I$ and for any $(\mathbf{y},\theta,k)\in\boldsymbol{K}\times\rit\times\left[c^{\diamondsuit},d^{\diamondsuit}\right]$, we have
\begin{align}
       &h_{\eps}\!\left(\boldsymbol{\vartheta}_{\eps,-\bar{g}_{i}}^{\alpha_{i,N},i}\left(\mathbf{y},\theta,k\right)\right)=\boldsymbol{\vartheta}_{\eps,-\bar{g}_{i}}^{\alpha_{i,N},i}\cdot h_{\eps}\!\left(\mathbf{y},\theta,k\right)+\eps^{N+1}\Rest_{h}^{N}\!\left(\eps;\mathbf{y},\theta,k\right),
       \label{QuasiComRondPoint}    
\end{align}
where $\Rest_{h}^{N}$ is in $\mathcal{C}_{\PerND}^{\infty}\!\left((I\cap\left[-\eta_K,\eta_K\right])\times\boldsymbol{K}\times\rit\times\left[c^{\diamondsuit},d^{\diamondsuit}\right]\right)$.
\end{theorem}

\begin{proof}
Since $h_\eps\in\mathcal{Q}_{T,b}^{\infty},$ and by linearity, the proof of the theorem reduces to prove formula \eqref{QuasiComRondPoint}  
with function $h_\eps$ of the form 
\begin{align}
\label{201402161112} 
     &h_{\eps}\!\left(\mathbf{y},\theta,k\right)=\cos^{l}\!\left(\theta\right)\sin^{m}\!\left(\theta\right)d^\eps\!\left(\mathbf{y}\right)\sqrt{k}^{n},
\end{align}
where $d^{\eps}=d^{\eps}\!\left(\mathbf{y}\right)\in\mathcal{A}\!\left(\mathbb{R}^{2}\right)\cap\mathcal{C}_{b}^{\infty}\left(\mathbb{R}^{2}\right)$.

Let $\boldsymbol{\bar{r}}_{0}=\left(\mathbf{y}_{0},\theta_{0},k_{0}\right)\in\boldsymbol{K}\times\rit\times\left[c^{\diamondsuit},d^{\diamondsuit}\right]$.
As $d^{\eps}\in\mathcal{A}\!\left(\mathbb{R}^{2}\right),$ and as $\left(k\mapsto\sqrt{k}^{n}\right)\in
\mathcal{A}\!\left(\left(0,+\infty\right)\right)$, there exists a real number $R_{\boldsymbol{\bar{r}}_{0}}>0$ and
a formal power series $T_{\boldsymbol{\bar{r}}_{0}}$ of three variables whose set of convergence contains 
the closure of $\ball^{3}\!\big(0,R_{\boldsymbol{\bar{r}}_{0}}\big)$, which are such that
 $\ball^{3}\!\left(\left(\mathbf{y}_{0},k_{0}\right),R_{\boldsymbol{\bar{r}}_{0}}\right)\subset\mathbb{R}^{2}\times\left(c^{\bullet},d^{\bullet}\right)$ 
and such that for any $\left(\mathbf{y},k\right) \in\ball^{3}\!\left(\left(\mathbf{y}_{0},k_{0}\right),R_{\boldsymbol{\bar{r}}_{0}}\right)$,
\begin{gather}
    d^\eps\left(\mathbf{y}\right)\sqrt{k}^{n}=T_{\boldsymbol{\bar{r}}_{0}}\!\left(\left(\mathbf{y},\theta\right)-\left(\mathbf{y}_{0},\theta_{0}\right)\right)
   =\underset{\mathbf{l}\in\mathbb{N}^{3}}{\sum}a_{\eps}^{\mathbf{l},\boldsymbol{\bar{r}}_{0}}\left(\left(\mathbf{y},k\right)-\left(\mathbf{y}_{0},\theta_{0}\right)\right)^{\mathbf{l}}.
\end{gather}
In addition, since $\left(\theta\mapsto\cos^{l}\!\left(\theta\right)\sin^{m}\!\left(\theta\right)\right)$
is a power series of convergence radius $+\infty$ with respect to $\theta,$ there exists a formal power series
$S_{\boldsymbol{\bar{r}}_{0}}$ such that $\overline{\ball^{\#}\!\left(0,R_{\boldsymbol{\bar{r}}_{0}}\right)}\subset{\boldsymbol{\Sigma}}_{S_{\boldsymbol{\bar{r}}_{0}}}$
  and such that 
\begin{align}
   &\forall\boldsymbol{\bar{r}}=\left(\mathbf{y},\theta,k\right)\in\ball^{\#}\!\left(\boldsymbol{\bar{r}}_{0},R_{\boldsymbol{\bar{r}}_{0}}\right),\ h_{\eps}\!\left(\boldsymbol{\bar{r}}\right)=S_{\boldsymbol{\bar{r}}_{0}}\!\left(\boldsymbol{\bar{r}}-\boldsymbol{\bar{r}}_{0}\right)=\sum_{\mathbf{l}\in\mathbb{N}^{4}}h_{\eps}^{\mathbf{l},\boldsymbol{\bar{r}}_{0}}\left(\boldsymbol{\bar{r}}-\boldsymbol{\bar{r}}_{0}\right)^{\mathbf{l}}.
\end{align}
Let $R'_{\boldsymbol{\bar{r}}_{0}}\in\left(0,R_{\boldsymbol{\bar{r}}_{0}}\right)$. Then, using similar arguments as in the proof of Theorem \ref{BoldChiNWellDef}
we easily obtain that
there exists a real number 
$\eta_{R_{\boldsymbol{\bar{r}}_{0}},R'_{\boldsymbol{\bar{r}}_{0}}}>0$ such that for any
$\begingroup\small\eps\in\big[-\eta_{R_{\boldsymbol{\bar{r}}_{0}},R'_{\boldsymbol{\bar{r}}_{0}}},\eta_{R_{\boldsymbol{\bar{r}}_{0}},R'_{\boldsymbol{\bar{r}}_{0}}}\big],\endgroup$ 
$\boldsymbol{\vartheta}_{\eps,-\bar{g}_{i}}^{\alpha_{i,N},i}\!\left(\ball^{\#}\!\left(\boldsymbol{\bar{r}}_{0},R'_{\boldsymbol{\bar{r}}_{0}}\right)\right)\subset\ball^{\#}\!\left(\boldsymbol{\bar{r}}_{0},R_{\boldsymbol{\bar{r}}_{0}}\right).$
Hence, for any $\boldsymbol{\bar{r}}=\left(\mathbf{y},\theta,k\right)\in\ball^{\#}\!\left(\boldsymbol{\bar{r}}_{0},R'_{\boldsymbol{\bar{r}}_{0}}\right)$, we have 
\begin{align}
    &h_{\eps}\!\left(\boldsymbol{\vartheta}_{\eps,-\bar{g}_{i}}^{\alpha_{i,N},i}\!\left(\boldsymbol{\bar{r}}\right)\right)=
    \sum_{\mathbf{l}\in\mathbb{N}^{4}}\; h_{\eps}^{\mathbf{l},\boldsymbol{\bar{r}}_{0}}\!\left(\boldsymbol{\vartheta}_{\eps,-\bar{g}_{i}}^{\alpha_{i,N},i}\!\left(\boldsymbol{\bar{r}}\right)-\boldsymbol{\bar{r}}_{0}\right)^{\mathbf{l}}.
    \label{FirstIntFormInProofOfThm424}   
\end{align}

On another hand, let $\Theta_{\eps}=\Theta_{\eps}\!\left(\boldsymbol{\bar{r}}\right)=\left(\Theta_{\eps,\mathbf{m}}\!\left(\boldsymbol{\bar{r}}\right)\right)_{\mathbf{m}\in\mathbb{N}^{4}\, s.t.\ \left|\mathbf{m}\right|\leq i}$ be the smooth function that are such that, for all smooth functions $f_\eps,$
\begin{align}
       &\left(\boldsymbol{\vartheta}_{\eps,-\bar{g}_{i}}^{\alpha_{i,N},i}\!\cdot f_{\eps}\right)\!\left(\boldsymbol{\bar{r}}\right)=\underset{\left|\mathbf{m}\right|\leq i}{\sum}\Theta_{\eps,\mathbf{m}}\!\left(\boldsymbol{\bar{r}}\right)\frac{\partial f_{\eps}}{\partial\boldsymbol{\bar{r}}^{\mathbf{m}}}\!\left(\boldsymbol{\bar{r}}\right).
\end{align}
We have, for any $\boldsymbol{\bar{r}}\in\ball^{\#}\!\left(\boldsymbol{\bar{r}}_{0},R_{\boldsymbol{\bar{r}}_{0}}\right)$,
\begin{gather}
      \left(\boldsymbol{\vartheta}_{\eps,-\bar{g}_{i}}^{\alpha_{i,N},i}\!\cdot h_{\eps}\right)\!\left(\boldsymbol{\bar{r}}\right)=\underset{\left|\mathbf{m}\right|\leq i}{\sum}\Theta_{\eps,\mathbf{m}}\left(\boldsymbol{\bar{r}}\right)\!\frac{\partial h_{\eps}}{\partial\boldsymbol{\bar{r}}^{\mathbf{m}}}\!\left(\boldsymbol{\bar{r}}\right)=\left(\underset{\left|\mathbf{m}\right|\leq i}{\sum}\Theta_{\eps,\mathbf{m}}\!\left(\boldsymbol{\bar{r}}\right)\frac{\partial}{\partial\boldsymbol{\bar{r}}^{\mathbf{m}}}\right)\left(\underset{\mathbf{l}\in\mathbb{N}^{4}}{\sum}h_{\eps}^{l,\boldsymbol{\bar{r}}_{0}}\,\FctCoordPuiss{l}{\boldsymbol{\bar{r}}_{0}}\right)\!\left(\boldsymbol{\bar{r}}\right),
\end{gather}
where
$\FctCoordPuiss{l}{\boldsymbol{\bar{r}}_{0}}$ stand for the function $\boldsymbol{\bar{r}}\mapsto\left(\bar{r}_{1}-(\bar{\boldsymbol{r}}_{0})_{1}\right)^{l_{1}}\left(\bar{r}_{2}-(\bar{\boldsymbol{r}}_{0})_{2}\right)^{l_{2}}\left(\bar{r}_{3}-(\bar{\boldsymbol{r}}_{0})_{3}\right)^{l_{3}}\left(\bar{r}_{4}-(\bar{\boldsymbol{r}}_{0})_{4}\right)^{l_{4}}$.
\\
Since $\ball^{\#}\!\left(0,R_{\boldsymbol{\bar{r}}_{0}}\right)\subset{\boldsymbol{\Sigma}}_{S_{\boldsymbol{\bar{r}}_{0}}},$ we can 
permute summation and derivations and we obtain:
\begin{gather}
    \left(\boldsymbol{\vartheta}_{\eps,-\bar{g}_{i}}^{\alpha_{i,N},i}\!\cdot h_{\eps}\right)\!\left(\boldsymbol{\bar{r}}\right)=\underset{\left|\mathbf{m}\right|\leq i}{\sum}\Theta_{\eps,\mathbf{m}}\!\left(\boldsymbol{\bar{r}}\right)\underset{\mathbf{l}\in\mathbb{N}^{4}}{\sum}h_{\eps}^{l,\boldsymbol{\bar{r}}_{0}}\frac{\partial\FctCoordPuiss{l}{\boldsymbol{\bar{r}}_{0}}}{\partial\boldsymbol{\bar{r}}^{m}}\left(\boldsymbol{\mathfrak{r}}\right)=\underset{\mathbf{l}\in\mathbb{N}^{4}}{\sum}h_{\eps}^{\mathbf{l},\boldsymbol{\bar{r}}_{0}}\left(\boldsymbol{\vartheta}_{\eps,-\bar{g}_{i}}^{\alpha_{i,N},i}\!\cdot\FctCoordPuiss{l}{\boldsymbol{\bar{r}}_{0}}\right)\!\left(\boldsymbol{\bar{r}}\right).
\end{gather} 
Besides, using Property  \ref{PropPartLieFuncProd} and the link \eqref{DefVarThetN2} between function 
$\boldsymbol{\vartheta}_{\eps,-\bar{g}_i}^{\alpha_{i,N},i}$ and operator $\boldsymbol{\vartheta}_{\eps,-\bar{g}_i}^{\alpha_{i,N},i}$, we obtain that, for any $\mathbf{l}\in\mathbb{N}^{4}$,
\begin{gather}
 \label{SecondIntFormInProofOfThm424}   
\begin{aligned}
     \left(\boldsymbol{\vartheta}_{\eps,-\bar{g}_{i}}^{\alpha_{i,N},i}\!\cdot\FctCoordPuiss{l}{\boldsymbol{\bar{r}}_{0}}\right)\left(\boldsymbol{\bar{r}}\right)
     &=\left(\boldsymbol{\vartheta}_{\eps,-\bar{g}_{i}}^{\alpha_{i,N},i}\!\cdot(\boldsymbol{\bar{r}}_{1}-({\boldsymbol{\bar{r}}_{0}})_{1}),\dots,\boldsymbol{\vartheta}_{\eps,-\bar{g}_{i}}^{\alpha_{i,N},i}\!\cdot(\boldsymbol{\bar{r}}_{4}-({\boldsymbol{\bar{r}}_{0}})_{4})\right)^{\!\mathbf{l}}\left(\boldsymbol{\bar{r}}\right)+\eps^{N+1}\rest_{\mathbf{l},\boldsymbol{\bar{r}}_{0}}^{N,i,j}\!\left(\eps,\boldsymbol{\bar{r}}\right)
     \\
     &=\left(\left(\boldsymbol{\vartheta}_{\eps,-\bar{g}_{i}}^{\alpha_{i,N},i}\!\cdot\boldsymbol{\bar{r}}_{1},\dots,\boldsymbol{\vartheta}_{\eps,-\bar{g}_{i}}^{\alpha_{i,N},i}\!\cdot\boldsymbol{\bar{r}}_{4}\right)\!\left(\boldsymbol{\bar{r}}\right)-{\boldsymbol{\bar{r}}_{0}}\right)^{\!\mathbf{l}}+\eps^{N+1}\rest_{\mathbf{l},\boldsymbol{\bar{r}}_{0}}^{N,i,j}\!\left(\eps,\boldsymbol{\bar{r}}\right)
     \\
     &=\left(\left(\boldsymbol{\vartheta}_{\eps,-\bar{g}_{i}}^{\alpha_{i,N},i}\right)\!\left(\boldsymbol{\bar{r}}\right)-{\boldsymbol{\bar{r}}_{0}}\right)^{\!\mathbf{l}}+\eps^{N+1}\rest_{\mathbf{l},\boldsymbol{\bar{r}}_{0}}^{N,i,j}\!\left(\eps,\boldsymbol{\bar{r}}\right),
\end{aligned}
\end{gather}
with $\boldsymbol{\bar{r}}\mapsto\rest_{\mathbf{l},\boldsymbol{\bar{r}}_{0}}^{N,i,j}\!\left(\ \cdot \ , \boldsymbol{\bar{r}} \right)$
in  $Q_{T,b}^\infty$ and $\eps\mapsto\rest_{\mathbf{l},\boldsymbol{\bar{r}}_{0}}^{N,i,j}\!\left(\eps,\ \cdot \ \right)$ in $\mathcal{C}^\infty(\mathbb{R})$.

As both 
$\underset{\mathbf{l}\in\mathbb{N}^{4}}{\sum}h_{\eps}^{\mathbf{l},\boldsymbol{\bar{r}}_{0}}\left(\boldsymbol{\vartheta}_{\eps,-\bar{g}_{i}}^{\alpha_{i,N},i}\!\left(\boldsymbol{\bar{r}}\right)-\boldsymbol{\bar{r}}_{0}\right)^{\mathbf{l}}$
and 
$\underset{\mathbf{l}\in\mathbb{N}^{4}}{\sum}h_{\eps}^{\mathbf{l},\boldsymbol{\bar{r}}_{0}}\left(\boldsymbol{\vartheta}_{\eps,-\bar{g}_{i}}^{\alpha_{i,N},i}\!\cdot\FctCoordPuiss{l}{\boldsymbol{\bar{r}}_{0}}\right)\!\left(\boldsymbol{\bar{r}}\right)$ 
converge normally on

\hspace{-0.6cm}$\ball^{\#}\!\left(\boldsymbol{\bar{r}}_{0},R'_{\boldsymbol{\bar{r}}_{0}}\right)$,
their difference 
\begin{gather}
\eps^{N+1}\left(-\underset{\mathbf{l}\in\mathbb{N}^{4}}{\sum}h_{\eps}^{\mathbf{l},\boldsymbol{\bar{r}}_{0}}\rest_{\mathbf{l},\boldsymbol{\bar{r}}_{0}}^{N,i,j}\left(\eps,\boldsymbol{\bar{r}}\right)\right)
\end{gather}
also converges normally on this subset and we can deduce that, for any 
$\boldsymbol{\bar{r}}\in\ball^{\#}\!\left(\boldsymbol{\bar{r}}_{0},R'_{\boldsymbol{\bar{r}}_{0}}\right)$,
\begin{align}
       \left(h_{\eps}\circ\boldsymbol{\vartheta}_{\eps,-\bar{g}_{i}}^{\alpha_{i,N},i}\!\right)\left(\boldsymbol{\bar{r}}\right)
       &=\left(\boldsymbol{\vartheta}_{\eps,-\bar{g}_{i}}^{\alpha_{i,N},i}\!\cdot h_{\eps}\right)\left(\boldsymbol{\bar{r}}\right)+\eps^{N+1}\left(-\underset{\mathbf{l}\in\mathbb{N}^{4}}{\sum}h_{\eps}^{\mathbf{l},\boldsymbol{\bar{r}}_{0}}\rest_{\mathbf{l},\boldsymbol{\bar{r}}_{0}}^{N,i,j}\left(\eps,\boldsymbol{\bar{r}}\right)\right).
\end{align}

Finally, as 
\begin{gather}
\boldsymbol{K}\times\left[c^{\diamondsuit},d^{\diamondsuit}\right]\subset\underset{\left(\mathbf{y}_{0},k_{0}\right)\in\boldsymbol{K}\times\left[c^{\diamondsuit},d^{\diamondsuit}\right]}{\cup}\ball^{3}\!\left(\left(\boldsymbol{y}_{0},k_{0}\right),R'_{\boldsymbol{\bar{r}}_{0}}\right)
\end{gather}
and as $\boldsymbol{K}\times\left[c^{\diamondsuit},d^{\diamondsuit}\right]$ is compact, there exists $\left(\boldsymbol{y}_{0}^{1},k_{0}^{1}\right),\ldots,\left(\boldsymbol{y}_{0}^{p},k_{0}^{p}\right)$ such that 
\begin{gather}
\boldsymbol{K}\times[c^{\diamondsuit},d^{\diamondsuit}]\subset\underset{i=1}{\overset{p}{\cup}}\ball^{3}\!\left(\left(\boldsymbol{y}_{0}^{i},k_{0}^{i}\right),R'_{\boldsymbol{\bar{r}}_{0}^{i}}\right).
\end{gather}
Setting $\ds\bar{\eta}_{5}=\min_{i=1,\dots,p}\ \eta_{R_{\boldsymbol{\bar{r}}_{0}^{i}},R'_{\boldsymbol{\bar{r}}_{0}^{i}}}$, we obtain equality \eqref{QuasiComRondPoint} 
for all  $\left(\mathbf{y},\theta,k\right)\in\boldsymbol{K}\times\rit\times\left[c^{\diamondsuit},d^{\diamondsuit}\right]$ and for all $\eps\in I\cap\left[-\bar{\eta}_{5},\bar{\eta}_{5}\right]$, for any function of the form \eqref{201402161112}. 

Consequently, as seen in the beginning of the proof, equality \eqref{QuasiComRondPoint}  is true for  any $h_\eps$ in $\mathcal{Q}_{T,b}^{\infty}\cap\mathcal{A}\!\left(\mathbb{R}^{2}\times\rit\times\left(0,+\infty\right)\right)$, proving the Theorem.
\end{proof}

\begin{remark}
\label{FormulaPropAndRondPointExtension}  
Formulas \eqref{PartLieFuncProd}, \eqref{PartLiePoissCom} and \eqref{QuasiComRondPoint} are also valid if we replace 
\begingroup\small
$\boldsymbol{\vartheta}_{\eps,-\bar{g}_i}^{\alpha_{i,N},i}$ 
\endgroup
by 
\begingroup\small
$\boldsymbol{\vartheta}_{\eps,\bar{g}_{1}}^{\alpha_{1,N},1}\cdot\boldsymbol{\vartheta}_{\eps,\bar{g}_{2}}^{\alpha_{2,N},2}\cdot\ldots\cdot\boldsymbol{\vartheta}_{\eps,\bar{g}_{N}}^{\alpha_{N,N},N}\!\!$ 
\endgroup
or by 
\begingroup\small
$\boldsymbol{\vartheta}_{\eps,-\bar{g}_{N}}^{\alpha_{N,N},N}\cdot\boldsymbol{\vartheta}_{\eps,-\bar{g}_{N-1}}^{\alpha_{N-1,N},N-1}\cdot\ldots\cdot\boldsymbol{\vartheta}_{\eps,-\bar{g}_{1}}^{\alpha_{1,N},1}$.
\endgroup
The extensions of \eqref{PartLieFuncProd} and \eqref{PartLiePoissCom} are easily obtained by induction.
The extensions of Formula \eqref{QuasiComRondPoint} are obtained by replacing 
\begingroup\small
$\boldsymbol{\vartheta}_{\eps,-\bar{g}_i}^{\alpha_{i,N},i}$
\endgroup
by 
\begingroup\small
$\boldsymbol{\vartheta}_{\eps,\bar{g}_{1}}^{\alpha_{1,N},1}\cdot\boldsymbol{\vartheta}_{\eps,\bar{g}_{2}}^{\alpha_{2,N},2}\cdot\ldots\cdot\boldsymbol{\vartheta}_{\eps,\bar{g}_{N}}^{\alpha_{N,N},N}\!\!$
\endgroup
or by 
\begingroup\small
$\boldsymbol{\vartheta}_{\eps,-\bar{g}_{N}}^{\alpha_{N,N},N}\cdot\boldsymbol{\vartheta}_{\eps,-\bar{g}_{N-1}}^{\alpha_{N-1,N},N-1}\cdot\ldots\cdot\boldsymbol{\vartheta}_{\eps,-\bar{g}_{1}}^{\alpha_{1,N},1}\!\!$
\endgroup
in the proof of Theorem  \ref{ThemQuasiComRondPoint}.
\end{remark}

Property \ref{PropPartLieFuncProd}, Theorem \ref{ThemQuasiComRondPoint}, and Remark \ref{FormulaPropAndRondPointExtension} are the main tools we need to prove Theorem \ref{MainPropOfLieCCOrderNThm}.
 \begin{proof}[Proof of Theorem \ref{MainPropOfLieCCOrderNThm}]
 Since function $\boldsymbol{\nu}_{\eps,-\bar{g}_i}^{\alpha_{i,N},i}$ of  equality \eqref{DefNuInTermsOfTheta} 
 satisfies the assumptions of Theorem \ref{ThemQuasiComRondPoint}, formula \eqref{QuasiComRondPoint} 
is valid with $\boldsymbol{\vartheta}_{\eps,-\bar{g}_i}^{\alpha_{i,N},i}$ in the role of  $h_\eps$. Hence, 
for any  $i\in\ldbrack1,N\rdbrack$, we deduce
\begin{align}
   &\boldsymbol{\vartheta}_{\eps,-\bar{g}_{i}}^{\alpha_{i,N},i}\circ\boldsymbol{\vartheta}_{\eps,\bar{g}_{i}}^{\alpha_{i,N},i}\left(\mathbf{y},\theta,k\right)=\boldsymbol{\vartheta}_{\eps,\bar{g}_{i}}^{\alpha_{i,N},i}\cdot\boldsymbol{\vartheta}_{\eps,-\bar{g}_{i}}^{\alpha_{i,N},i}\left(\mathbf{y},\theta,k\right)+\eps^{N+1}\boldsymbol{\rho}_{FC}^{N,i}\left(\eps;\mathbf{y},\theta,k\right),
   \label{Form22January}   
\end{align}
with $\boldsymbol{\rho}_{FC}^{N,i}$ in $\mathcal{C}_{\PerND}^{\infty}\!\left((I\cap\left[-\eta_K,\eta_K\right])\times\boldsymbol{K}\times\rit\times\left[c^{\diamondsuit},d^{\diamondsuit}\right]\right)$.
And, an easy computation leads to
\begin{align}
    \boldsymbol{\vartheta}_{\eps,\bar{g}_{i}}^{\alpha_{i,N},i}\cdot\boldsymbol{\vartheta}_{\eps,-\bar{g}_{i}}^{\alpha_{i,N},i}=\left(\overset{\alpha_{i,N}}{\underset{l=0}{\sum}}\frac{\left(\eps^{i}\right)^{l}}{l!}\left(\bar{\mathbf{X}}_{\eps\bar{g}_{i}}^{\eps}\right)^{l}\cdot\right)\left(\overset{\alpha_{i,N}}{\underset{k=0}{\sum}}\frac{\left(-\eps^{i}\right)^{k}}{k!}\left(\bar{\mathbf{X}}_{\eps\bar{g}_{i}}^{\eps}\right)^{k}\cdot\right)
    &=id+\eps^{N+1}\boldsymbol{\rho}_{c}^{N,i}\left(\eps,\cdot\right),
    \label{23Januar1200}  
\end{align}
with $(\mathbf{y},\theta,k)\mapsto\boldsymbol{\rho}_{c}^{N,i}\!\left(\ \cdot \ ,\mathbf{y},\theta,k\right)$
in  $Q_{T,b}^\infty$ and $\eps\mapsto\boldsymbol{\rho}_{c}^{N,i}\!\left(\eps,\ \cdot \ \right)$ in $\mathcal{C}^\infty(\mathbb{R})$.
Injecting \eqref{23Januar1200} in \eqref{Form22January},
applying $\boldsymbol{\Xi}_{\eps,-\bar{g}_{i}}^{\alpha_{i,N},i}$ (see \eqref{201402161137}) to both sides, and using a Taylor expansion
to expand 
 $\boldsymbol{\Xi}_{\eps,-\bar{g}_{i}}^{\alpha_{i,N},i}\!\big( id+\eps^{N+1}\boldsymbol{\rho}_{c}^{N,i} 
 +\eps^{N+1}\boldsymbol{\rho}_{FC}^{N,i} \big)$ (which is possible since Theorem  \ref{BoldChiNWellDef} 
 implies that $\boldsymbol{\Xi}_{\eps,-\bar{g}_{i}}^{\alpha_{i,N},i}$ is well defined on $\mathbb{R}^3\times(c,d)$), we obtain that 
\begin{align}
     &\boldsymbol{\Xi}_{\eps,-\bar{g}_{i}}^{\alpha_{i,N},i}\left(\mathbf{y},\theta,k\right)
     =\boldsymbol{\vartheta}_{\eps,\bar{g}_{i}}^{\alpha_{i,N},i}\left(\mathbf{y},\theta,k\right)
     +\eps^{N+1}\boldsymbol{\rho}_{\boldsymbol{\Xi}}^{N,i}\left(\eps;\mathbf{y},\theta,k\right),
\end{align} 
with $\boldsymbol{\rho}_{\boldsymbol{\Xi}}^{N,i}$ in $\mathcal{C}_{\PerND}^{\infty}\!\left((I\cap\left[-\eta_K,\eta_K\right])\times\boldsymbol{K}\times\rit\times\left[c^{\diamondsuit},d^{\diamondsuit}\right]\right)$.
\\
Then, a straightforward induction, using essentially the extension of Formula \eqref{QuasiComRondPoint} given in Remark \ref{FormulaPropAndRondPointExtension},
we obtain Formula \eqref{DefLamExpPartLieCCOrderN}.
\\

By definition, the expression of the Hamiltonian function in the Partial Lie Coordinate System of order $N$ is given by  
\begin{align}
    &\hat{H}_\eps(\mathbf{z},\gamma,j)=\bar{H}_\eps\left(\boldsymbol{\lambda}_{\eps}^{N}(\mathbf{z},\gamma,j)\right).
\end{align}
Hence, using Formula \eqref{DefLamExpPartLieCCOrderN}, making a Taylor expansion, and using the extension of 
Theorem \ref{ThemQuasiComRondPoint} and  formula \eqref{QuasiComRondPoint} given in 
Remark \ref{FormulaPropAndRondPointExtension}, we obtain Formula \eqref{DefHatHExpPartLieCCOrderN}. 
According to the regularity property of $\bar{H}_\eps$ with respect to $\eps$ we can take interval $I$ of 
Theorem \ref{ThemQuasiComRondPoint} as being $[0,+\infty)$.
\\

By definition,  entry $(l,m)$ of the Poisson Matrix, expressed in the Partial Lie Coordinate System of order $N$ induced by $\boldsymbol{\chi}_{\eps}^{N}$, 
is given by 
\begin{align}
   &\left(\hat{\mathcal{P}}_{\eps}\right)_{l,m}\negmedspace\left(\mathbf{z},\gamma,j\right)=\left\{ \left(\boldsymbol{\chi}_{\eps}^{N}\right)_{l},\left(\boldsymbol{\chi}_{\eps}^{N}\right)_{m}\right\} _{\mathcal{D}}\left(\boldsymbol{\lambda}_{\eps}^{N}\negthickspace\left(\mathbf{z},\gamma,j\right)\right).
   \label{AprioriExpOfLiePoissMat}  
\end{align}
Using the extension of Formula \eqref{QuasiComRondPoint} given in Remark \ref{FormulaPropAndRondPointExtension}, it is an easy task to show by induction that the $l$-th component of $\boldsymbol{\chi}_{\eps}^{N}$ is given by 
\begin{align}
   &\left(\boldsymbol{\chi}_{\eps}^{N}\right)_{l}=\boldsymbol{\vartheta}_{\eps,-\bar{g}_{N}}^{\alpha_{N,N},N}\cdot\ldots\cdot\boldsymbol{\vartheta}_{\eps,-\bar{g}_{1}}^{\alpha_{1,N},1}\cdot\bar{\boldsymbol{r}}_{l}+\eps^{N+1}\boldsymbol{\rho}_{1}^{N},
   \label{Formula220114}  
\end{align}
where $\bar{\boldsymbol{r}}_{1}=\mathbf{y}_1$, $\bar{\boldsymbol{r}}_{2}=\mathbf{y}_2$, $\bar{\boldsymbol{r}}_{3}=\boldsymbol{\theta}$, and
$\bar{\boldsymbol{r}}_{4}=\mathbf{k}$ (see Definition \ref{DefBoldVarijN2}), and 
with $\boldsymbol{\rho}_{1}^{N}$ in $\mathcal{C}_{\PerND}^{\infty}\!\left((I\cap\left[-\eta_K,\eta_K\right])\times\boldsymbol{K}\times\rit\times\left[c^{\diamondsuit},d^{\diamondsuit}\right]\right)$.
\\
Injecting Formula \eqref{Formula220114}  in the right hand side of \eqref{AprioriExpOfLiePoissMat}, using the bi-linearity of the Poisson Bracket, and the extension 
of Formula  \eqref{PartLiePoissCom} given in Remark \ref{FormulaPropAndRondPointExtension} we obtain 
\begin{align}
     &\left\{ \left(\boldsymbol{\chi}_{\eps}^{N}\right)_{l},\left(\boldsymbol{\chi}_{\eps}^{N}\right)_{m}\right\}_{\!\mathcal{D}}\negthickspace\left(\mathbf{y},\theta,k\right)=\boldsymbol{\vartheta}_{\eps,-\bar{g}_{N}}^{\alpha_{N,N},N}\cdot\ldots\cdot\boldsymbol{\vartheta}_{\eps,-\bar{g}_{1}}^{\alpha_{1,N},1}\cdot\left\{ \bar{\boldsymbol{r}}_{l},\bar{\boldsymbol{r}}_{m}\right\}_{\!\mathcal{D}}\negthickspace\left(\mathbf{y},\theta,k\right)+\eps^{N+1}\boldsymbol{\rho}_{2}^{N}\left(\eps;\mathbf{y},\theta,k\right),
\end{align}
with $\boldsymbol{\rho}_{2}^{N}$ in $\mathcal{C}_{\PerND}^{\infty}\!\left((I\cap\left[-\eta_K,\eta_K\right])\times\boldsymbol{K}\times\rit\times\left[c^{\diamondsuit},d^{\diamondsuit}\right]\right)$.
\\
On another hand, using formula \eqref{Formula220114} and
the extension 
of Formula  \eqref{QuasiComRondPoint} given in Remark \ref{FormulaPropAndRondPointExtension} we obtain 
\begin{align}
    &\left\{ \bar{\boldsymbol{r}}_{l},\bar{\boldsymbol{r}}_{m}\right\} _{\mathcal{D}}\left(\boldsymbol{\chi}_{\eps}^{N}\left(\mathbf{y},\theta,k\right)\right)=\boldsymbol{\vartheta}_{\eps,-\bar{g}_{N}}^{\alpha_{NN},N}\cdot\ldots\cdot\boldsymbol{\vartheta}_{\eps,-\bar{g}_{1}}^{\alpha_{1,N},1}\cdot\left\{ \bar{\boldsymbol{r}}_{l},\bar{\boldsymbol{r}}_{m}\right\} _{\mathcal{D}}\left(\mathbf{y},\theta,k\right)+\eps^{N+1}\boldsymbol{\rho}_{3}^{N}\left(\eps;\mathbf{y},\theta,k\right),
\end{align}
with $\boldsymbol{\rho}_{3}^{N}$ in $\mathcal{C}_{\PerND}^{\infty}\!\left((I\cap\left[-\eta_K,\eta_K\right])\times\boldsymbol{K}\times\rit\times\left[c^{\diamondsuit},d^{\diamondsuit}\right]\right)$.
\\
Eventually, combining the two previous Formulas yields Formula \eqref{DefHatMatPoissExpPartLieCCOrderN}.
This ends the proof of Theorem \ref{MainPropOfLieCCOrderNThm}
\end{proof}

\subsection{Proof of Theorem \ref{ExpForTheAlgoWellPosed} }
\label{SectionProofOfTheoremExpForTheAlgoWellPosed}   

Having expansion \eqref{DsDim423} in mind, the proof of Theorem \ref{ExpForTheAlgoWellPosed} consists essentially in ordering the terms in
Formula \eqref{DefHatHExpPartLieCCOrderN} with respect to their power of $\eps$. More precisely, we will focus on 
expanding 
\begin{align}
     &\boldsymbol{\vartheta}_{\eps,\bar{g}_{1}}^{\alpha_{1,N},1}\cdot\boldsymbol{\vartheta}_{\eps,\bar{g}_{2}}^{\alpha_{2,N},2}\cdot\ldots\cdot\boldsymbol{\vartheta}_{\eps,\bar{g}_{N}}^{\alpha_{N,N},N}\cdot\bar{H}^N_{\eps}\!\!\left(\mathbf{z},\gamma,j\right),
     \label{23January20141626}  
\end{align}
where 
\begin{align}
    &\bar{H}_{\eps}^{N}\left(\mathbf{y},\theta,k\right)=\bar{H}_{0}\left(\mathbf{y},k\right)+\eps\bar{H}_{1}\left(\mathbf{y},\theta,k\right)+\ldots+\eps^{N}\bar{H}_{N}\left(\mathbf{y},\theta,k\right).
\end{align}
We easily obtain  that Formula \eqref{23January20141626} can be rewritten as 
\begin{align}
   &\overset{N}{\underset{n=0}{\sum}}\eps^{n}\left(\overset{n}{\underset{k=0}{\sum}}\mathbf{\overline{V}}_{n-k}^{\eps}\cdot\bar{H}_{k}\left(\mathbf{z},\gamma,j\right)\right)+\eps^{N+1}\iota_{\bar{H}}^{N,\bullet}\!\left(\eps;\mathbf{z},\gamma,j\right),
\end{align}
with $\iota_{\bar{H}}^{N,\bullet}$ in $\mathcal{C}_{\PerND}^{\infty}\!\left((I\cap\left[-\eta_K,\eta_K\right])\times\boldsymbol{K}\times\rit\times\left[c^{\diamondsuit},d^{\diamondsuit}\right]\right)$,
where 
\begin{align}
     &\overline{\mathbf{V}}_{l}^{\eps}=\underset{\left(m_{1},\ldots,m_{l}\right)\in\mathcal{U}_{l}}{\sum}\frac{\left(\overline{\mathbf{X}}_{\eps\bar{g}_{1}}^{\eps}\right)^{m_{1}}\cdot\ldots\cdot\left(\overline{\mathbf{X}}_{\eps\bar{g}_{l}}^{\eps}\right)^{m_{l}}\cdot}{m_{1}!\ldots m_{l}!},
     \label{DefVl}    
\end{align}
with 
\begin{align}
     &\mathcal{U}_{p}=\big\{ \left(m_{1},\ldots,m_{p}\right)\in\mathbb{N}^{p} \text{ s.t. }\underset{k=1}{\overset{p}{\sum}}km_{k}=p\big\}.
\end{align} 
The only possible values that $m_l$ can have in formula \eqref{DefVl} 
are $0$ and $1$. If $m_l=1,$ then $m_1=m_2=\ldots=m_{l-1}=0$. Hence, the only
term in the sum of the right hand side of \eqref{DefVl} that involves function $\bar{g}_l$ is $\overline{\mathbf{X}}_{\eps\bar{g}_{l}}^{\eps}.$ 
Consequently the only term in
\begin{align}
     &\overset{n}{\underset{k=0}{\sum}}\overline{\mathbf{V}}_{n-k}^{\eps}\cdot\bar{H}_{k},
\end{align}
that involves function $\bar{g}_n$ is 
\begin{gather}
\begin{aligned}
    \overline{\mathbf{X}}_{\eps\bar{g}_{n}}^{\eps}\cdot\bar{H}_{0}&=\frac{\partial\bar{g}_{n}}{\partial k}\frac{\partial\bar{H}_{0}}{\partial\theta}-\frac{\partial\bar{g}_{n}}{\partial\theta}\frac{\partial\bar{H}_{0}}{\partial k}-\frac{\eps^{2}}{B\left(\mathbf{y}\right)}\left(\frac{\partial\bar{g}_{n}}{\partial y_{2}}\frac{\partial\bar{H}_{0}}{\partial y_{1}}-\frac{\partial\bar{g}_{n}}{\partial y_{1}}\frac{\partial\bar{H}_{0}}{\partial y_{2}}\right)
    \\
    &=-\frac{\partial\bar{g}_{n}}{\partial\theta}\frac{\partial\bar{H}_{0}}{\partial k}-\frac{\eps^{2}}{B\left(\mathbf{y}\right)}\left(\frac{\partial\bar{g}_{n}}{\partial y_{2}}\frac{\partial\bar{H}_{0}}{\partial y_{1}}-\frac{\partial\bar{g}_{n}}{\partial y_{1}}\frac{\partial\bar{H}_{0}}{\partial y_{2}}\right).
\end{aligned}
\end{gather}
Consequently, gathering terms having the same power of $\eps$ we obtain
\begin{align}
     &\hat{H}_{\eps}^{N}\!\left(\mathbf{z},\gamma,j\right)=\bar{H}_{0}\!\left(\mathbf{z},\gamma,j\right)+\eps\hat{H}_{1}\!\left(\mathbf{z},\gamma,j\right)+\ldots+\eps^{N}\hat{H}_{N}\!\left(\mathbf{z},\gamma,j\right),
\end{align}
where 
\begin{gather}
\label{201305290934}  
\hat{H}_{1}=\bar{H}_{1}-\frac{\partial\bar{g}_{1}}{\partial\theta}\frac{\partial\bar{H}_{0}}{\partial k},
\end{gather}
and, for any $i\in\left\{ 2,\ldots,N\right\}$
\begin{align}
      &\hat{H}_{i}=-\frac{\partial\bar{g}_{i}}{\partial\theta}\frac{\partial\bar{H}_{0}}{\partial k}-\mathcal{V}\left(\bar{g}_{1},\ldots,\bar{g}_{i-1}\right),
      \label{SystOfPDELie}   
\end{align}
with $\mathcal{V}\left(\bar{g}_{1},\ldots,\bar{g}_{i-1}\right)$ depending only on $\bar{g}_{1},\ldots,\bar{g}_{i-1}$ and their derivatives.
Since 
\begin{align}
    &\frac{\partial\bar{H}_{0}}{\partial k}\left(\mathbf{z},\gamma,j\right)=B\left(\mathbf{z}\right),
\end{align}
we obtain Formulas \eqref{LieTransMethodIntroZerothEq2}-\eqref{LieTransMethodIntroLastEq2}.


\subsection{Proof of Theorem \ref{MainThm2}}
\label{SectionProofOfTheoremMainThm2}   
%
Theorem \ref{MainThm2} is a direct consequence of Theorems \ref{BoldChiNWellDef}, \ref{MainPropOfLieCCOrderNThm},  
\ref{ExpForTheAlgoWellPosed}, Algorithm \ref{AlgorithmLieRevisedVersion}, and Theorem \ref{AlgoLeadsToFuncWellDef}.
{~\hfill $\square$} 
%
\subsection{Proof of Theorem \ref{MainThm3}}
\label{LastSection}     

Applying Lemmas \ref{ThmainresultCharDarbKDarb}, \ref{ThmainresultCharDarbYDarb} and \ref{ThmainresultCharDarbYDarb2}
and because the Lie change of coordinates is close to the identity (see formula \eqref{DefNuInTermsOfTheta}),
it is clear that there exists a compact set $\boldsymbol{K}_{\!\mathcal{L}}$, positive real numbers
$c_{\mathcal{L}}$ and $d_{\mathcal{L}}$, and a positive real number $\eta_{K_{\!\mathcal{L}}}$ such that for any $\eps\in[0,\eta_{K_{\mathcal{L}}}]$, for any 
$t\in[0,T]$, and for any $(\mathbf{x}_0,\mathbf{v}_0)\in\boldsymbol{K}_{\!\!\mathcal{C}}\times \boldsymbol{\mathfrak{C}}(c_{\mathcal{C}},d_{\mathcal{C}})$,
 the characteristic $(\mathbf{Z},\boldsymbol{\Gamma},\mathcal{J})$ associated with the Hamiltonian system \eqref{systeqs1}-\eqref{systeqs2} and expressed in the 
 $(\mathbf{z},\gamma,j)$ coordinate system stays in $\boldsymbol{K}_{\!\mathcal{L}}\times\rit\times(c_{\mathcal{L}},d_{\mathcal{L}})$.
 Consequently, we can apply Theorem \ref{MainPropOfLieCCOrderNThm}.
\\

To end this proof we will prove estimate \eqref{MainEstimationOfOurPaper}.
Setting 
\begin{gather}
\Rest_{\mathcal{P}}^{N}(\eps;\mathbf{z},\gamma,j)=\left(\begin{array}{cc}
\left(\Rest_{\mathcal{P}}^{N}(\eps;\mathbf{z},\gamma,j)\right)^{\text{\tiny\! TL}} & \left(\Rest_{\mathcal{P}}^{N}(\eps;\mathbf{z},\gamma,j)\right)^{\text{\tiny\! TR}}\\
\left(\Rest_{\mathcal{P}}^{N}(\eps;\mathbf{z},\gamma,j)\right)^{\text{\tiny\! BL}} & \left(\Rest_{\mathcal{P}}^{N}(\eps;\mathbf{z},\gamma,j)\right)^{\text{\tiny\! BR}}\end{array}\right)
=\left(\left(\Rest_{\mathcal{P}}^{N}(\eps;\mathbf{z},\gamma,j)\right)^{i,j}\right)_{i,j=1,\dots4},
\end{gather}
and using the skew-symmetry of $\hat{\Pcal}_{\eps}$ in \eqref{DefHatMatPoissExpPartLieCCOrderN} yields:
\begin{multline}
\label{1212311003} 
\hat{\Pcal}_{\eps}(\mathbf{z},\gamma,j)=\\
\left(\begin{array}{cccc}
0 & -\frac{\eps}{B(\mathbf{z})}+\eps^{N+1}\!\left(\Rest_{\mathcal{P}}^{N}\right)^{\!1,2} & \eps^{N+1}\!\left(\Rest_{\mathcal{P}}^{N}\right)^{\!1,3} & \eps^{N+1}\!\left(\Rest_{\mathcal{P}}^{N}\right)^{\!1,4}\\
-\frac{\eps}{B(\mathbf{z})}-\eps^{N+1}\!\left(\Rest_{\mathcal{P}}^{N}\right)^{\!1,2} & 0 & \eps^{N+1}\!\left(\Rest_{\mathcal{P}}^{N}\right)^{\!2,3} & \eps^{N+1}\!\left(\Rest_{\mathcal{P}}^{N}\right)^{\!2,4}\\
-\eps^{N+1}\!\left(\Rest_{\mathcal{P}}^{N}\right)^{\!1,3} & -\eps^{N+1}\!\left(\Rest_{\mathcal{P}}^{N}\right)^{\!2,3} & 0 & \frac{1}{\eps}+\eps^{N+1}\!\left(\Rest_{\mathcal{P}}^{N}\right)^{\!3,4}\\
-\eps^{N+1}\!\left(\Rest_{\mathcal{P}}^{N}\right)^{\!1,4} & -\eps^{N+1}\!\left(\Rest_{\mathcal{P}}^{N}\right)^{\!2,4} & -\frac{1}{\eps}-\eps^{N+1}\!\left(\Rest_{\mathcal{P}}^{N}\right)^{\!3,4} & 0\end{array}\right).
\end{multline}

Now, we will check that $(\mathbf{Z},\mathcal{J})$ is in
$ \mathcal{C}^{N-1}\!\left(\left[0,\eta_{K_{\mathcal{L}}}\right]\right)$.
In order to check this,
we define
for any $\eps\in\left(0,\eta_{K_{\mathcal{L}}}\right],$ for any $t\in[0,T],$ and for any $(\mathbf{z},\gamma,j)\in\mathcal{U}_{\mathcal{C}}$, 
 $(\widetilde{\mathbf{Z}}, \widetilde{\Gamma}, \widetilde{\mathcal{J}})$ by
\begin{gather}
\begin{aligned}
      &\left(\begin{array}{c}
\widetilde{\mathbf{Z}}(t;\mathbf{z},\gamma,j)\\
\widetilde{\Gamma}(t;\mathbf{z},\gamma,j)\\
\widetilde{\mathcal{J}}(t;\mathbf{z},\gamma,j)\end{array}\right)=\left(\begin{array}{c}
\mathbf{Z}(\eps t;\mathbf{z},\gamma,j)\\
\Gamma(\eps t;\mathbf{z},\gamma,j)\\
\mathcal{J}(\eps t;\mathbf{z},\gamma,j)\end{array}\right).
      \label{DefOfTildeRInProof}    
\end{aligned}
\end{gather}
It satisfies
\begin{align}
      &\frac{\partial}{\partial t\!}\left(\begin{array}{c}
\widetilde{\mathbf{Z}}\\
\widetilde{\Gamma}\\
\widetilde{\mathcal{J}}\end{array}\right)\left(t\right)=\eps\hat{\Pcal}_{\eps}\left(\widetilde{\mathbf{Z}}\left(t\right),\widetilde{\Gamma}(t),\widetilde{\mathcal{J}}(t)\right)\nabla\hat{H}_{\eps}\!\left(\widetilde{\mathbf{Z}}\left(t\right),\widetilde{\Gamma}(t),\widetilde{\mathcal{J}}(t)\right),
        \label{FirstEqProofTroncReg}    
\end{align}
Since $\eps\mapsto\eps\hat{\Pcal}_{\eps}$ is in $\mathcal{C}^{\infty}\left(\left[0,\eta_{K_{\mathcal{L}}}\right]\right)$, the solution of \eqref{FirstEqProofTroncReg}  
depends smoothly on the parameter $\eps.$ In particular function $(\widetilde{\mathbf{Z}}, \widetilde{\Gamma}, \widetilde{\mathcal{J}})$, defined by \eqref{DefOfTildeRInProof}, is smoothly 
extensible at $\eps=0$.
 On another hand, for any $\eps\in\left(0,\eta_{K_{\mathcal{L}}}\right],$
and for any $t\in\left[0,T\right],$
$(\mathbf{Z}, \mathcal{J})$ is solution to
\begin{gather}
\label{EqSatByThreeComp}    
\begin{aligned}
         &\fracp{\mathbf{Z}}{t}=\mathtt{M}_{\eps}\negmedspace\left(\mathbf{Z}\right)\begin{pmatrix}\frac{\partial\hat{H}_{\eps}^{N}}{\partial z_{1}}\\
\frac{\partial\hat{H}_{\eps}^{N}}{\partial z_{2}}\end{pmatrix}\!\left(\mathbf{Z}\left(t\right),\mathcal{J}\left(t\right)\right)+
         \\
         &\eps^{N+1}\left[\mathtt{M}_{\eps}\!\begin{pmatrix}\frac{\partial\rest_{H}^{N}}{\partial z_{1}}\\
\frac{\partial\rest_{H}^{N}}{\partial z_{2}}\end{pmatrix}\!+\left(\Rest_{\mathcal{P}}^{N}(\eps,.)\right)^{\text{\tiny\!\ TL}}\!\!\begin{pmatrix}\frac{\partial\hat{H}_{\eps}}{\partial z_{1}}\\
\frac{\partial\hat{H}_{\eps}}{\partial z_{2}}\end{pmatrix}\!\!+\left(\Rest_{\mathcal{P}}^{N}(\eps,.)\right)^{\text{\tiny\!\ TR}}\!\!\begin{pmatrix}\frac{\partial\hat{H}_{\eps}}{\partial\gamma}\\
\frac{\partial\hat{H}_{\eps}}{\partial j}\end{pmatrix}\!\!\!\right]\!\left(\mathbf{Z},\widetilde{\Gamma}\!\left(\frac{t}{\eps}\right),\mathcal{J}\right),
         \\
         &\frac{\partial\mathcal{J}}{\partial t}=-\eps^{N}\left[\eps\left(\Rest_{\mathcal{P}}^{N}(\eps,.)\right)^{\!1,4}\frac{\partial\hat{H}_{\eps}}{\partial z_{1}}+\eps\left(\Rest_{\mathcal{P}}^{N}(\eps,.)\right)^{\!2,4}\frac{\partial\hat{H}_{\eps}}{\partial z_{2}}+\eps\left(\Rest_{\mathcal{P}}^{N}(\eps,.)\right)^{\!3,4}\frac{\partial\hat{H}_{\eps}}{\partial\gamma}+\frac{\partial\rest_{H}^{N}}{\partial\gamma}(\eps,.)\right]
         \\ 
         &~\hspace{10.3cm}
         \left(\mathbf{Z},\widetilde{\Gamma}\!\left(\frac{t}{\eps}\right),\mathcal{J}\right),
\end{aligned}
\end{gather}
where 
\begin{align}
      \label{DefMathttMEps}    
     &\mathtt{M}_{\eps}\!\left(\mathbf{z}\right)=\left(\begin{array}{cc}
0 & \ds -\frac{\eps}{B\left(\mathbf{z}\right)}\\
\ds \frac{\eps}{B\left(\mathbf{z}\right)} & 0\end{array}\right).
\end{align}
Notice that, in this system, $\widetilde{\Gamma}$ is known and then considered as given. Besides, 
\begin{align}
      &\mathtt{M}_{\eps}\!\begin{pmatrix}\frac{\partial\rest_{H}^{N}}{\partial z_{1}}\\
\frac{\partial\rest_{H}^{N}}{\partial z_{2}}\end{pmatrix}\!+\left(\Rest_{\mathcal{P}}^{N}(\eps,.)\right)^{\text{\tiny\!\ TL}}\!\!\begin{pmatrix}\frac{\partial\hat{H}_{\eps}}{\partial z_{1}}\\
\frac{\partial\hat{H}_{\eps}}{\partial z_{2}}\end{pmatrix}\!\!+\left(\Rest_{\mathcal{P}}^{N}(\eps,.)\right)^{\text{\tiny\!\ TR}}\!\!\begin{pmatrix}\frac{\partial\hat{H}_{\eps}}{\partial z_{1}}\\
\frac{\partial\hat{H}_{\eps}}{\partial z_{2}}\end{pmatrix}\!\!\!
\end{align}
and 
\begin{gather}
 \eps\left(\Rest_{\mathcal{P}}^{N}(\eps,.)\right)^{\!1,4}\frac{\partial\hat{H}_{\eps}}{\partial z_{1}}+\eps\left(\Rest_{\mathcal{P}}^{N}(\eps,.)\right)^{\!2,4}\frac{\partial\hat{H}_{\eps}}{\partial z_{2}}+\eps\left(\Rest_{\mathcal{P}}^{N}(\eps,.)\right)^{\!3,4}\frac{\partial\hat{H}_{\eps}}{\partial\gamma}+\frac{\partial\rest_{H}^{N}}{\partial\gamma}(\eps,.)
 \label{PrincipalObstruction} 
\end{gather}
are $2\pi$-periodic and smooth, and consequently $\mathcal{C}_{b}^{\infty}\!\left(\mathbb{R}\right)$ with respect to the third variable $\gamma$. Hence, computing the successive derivatives of \eqref{EqSatByThreeComp}  with respect to $\eps$, we obtain that  $\eps\mapsto\left(\mathbf{Z}(t),\mathcal{J}(t)\right)$ is $\mathcal{C}^{N-1}$
in the neighborhood of $\eps=0.$
\begin{remark}
The only obstruction to show that $\eps\mapsto\left(\mathbf{Z}(t),\mathcal{J}(t)\right)$ is $\mathcal{C}^{N}$ is the last term of Formula \eqref{PrincipalObstruction}.
\end{remark}

Moreover, as 
$\left(\mathbf{Z}^{T},\mathcal{J}^{T}\right)$ is solution to
\begin{gather}
\label{FormForTruncDynSys12Juin}    
\begin{aligned}
          &\frac{\partial\mathbf{Z}^{T}}{\partial t}=\frac{\eps}{B\left(\mathbf{Z}^{T}\right)}\left(\begin{array}{c}
-\frac{\partial\hat{H}_{\eps}^{N}}{\partial z_{2}}\\
\frac{\partial\hat{H}_{\eps}^{N}}{\partial z_{1}}\end{array}\right)\negmedspace\left(\mathbf{Z}^{T},\mathcal{J}^{T}\right)
         \\
         &\frac{\partial\mathcal{J}^{T}}{\partial t}=0,
\end{aligned}
\end{gather}
$\left(\mathbf{Z}^{T},\mathcal{J}^{T}\right)$ is smooth with respect to $\eps$,
for any $t\in[0,T]$.
\\

Now, we will show that $\mathbf{L}^{\eps}$ defined for $\eps\in(0,\eta_{K_{\mathcal{L}}}]$ by 
\begin{align}
    &\mathbf{L}^{\eps}=\left(\begin{array}{c}
\mathbf{L}_{1}^{\eps}\\
\mathbf{L}_{2}^{\eps}\\
\mathbf{L}_{3}^{\eps}\end{array}\right)=\frac{1}{\eps^{N-1}}\left(\left(\begin{array}{c}
\mathbf{Z}\\
\mathcal{J}\end{array}\right)-\left(\begin{array}{c}
\mathbf{Z}^{T}\\
\mathcal{J}^{T}\end{array}\right)\right)
    \label{DefOfLeps}   
\end{align}
is extensible 
to $[0,\eta_{K_{\mathcal{L}}}]$ and that the yielding extension is continuous with respect to $\eps.$ By definition for any $\eps\in(0,\eta_{K_{\mathcal{L}}}]$, for any $t\in[0,T],$
 $\eps\mapsto\mathbf{L}^{\eps}$ is $\mathcal{C}^{N-1}\!((0,\eta_{K_{\mathcal{L}}}])$.
So, we just have to show that $\eps\mapsto\mathbf{L}^{\eps}$ is extensible as a continuous function 
on $[0,\eta_{K_{\mathcal{L}}}]$, {\it i.e.} that $\eps=0$ is not a singularity.
\\
In a first place, for any $\eps\in(0,\eta_{K_{\mathcal{L}}}],$
we will explicit the dynamical system $\mathbf{L}^{\eps}$ satisfies.
Injecting  
\begin{align}
        &\left(\begin{array}{c}
\mathbf{Z}\\
\mathcal{J}\end{array}\right)=\left(\begin{array}{c}
\mathbf{Z}^{T}\\
\mathcal{J}^{T}\end{array}\right)+\eps^{N-1}\mathbf{L}^{\eps},
        \label{EstDynSysVSTroncDynSysF12}      
\end{align}
in \eqref{EqSatByThreeComp} gives
\begin{gather}
\label{EqSatByThreeComp2}    
\begin{aligned}
         &\fracp{\left(\begin{array}{c}
\mathbf{Z}_{1}^{T}+\eps^{N-1}\mathbf{L}_{1}^{\eps}\\
\mathbf{Z}_{2}^{T}+\eps^{N-1}\mathbf{L}_{2}^{\eps}\end{array}\right)}{t}
\\
         &~=\mathtt{M}_{\eps}\!\left(\mathbf{Z}_{1}^{T}+\eps^{N-1}\mathbf{L}_{1}^{\eps},\mathbf{Z}_{2}^{T}+\eps^{N-1}\mathbf{L}_{2}^{\eps}\right)\begin{pmatrix}\frac{\partial\hat{H}_{\eps}^{N}}{\partial z_{1}}\\
\frac{\partial\hat{H}_{\eps}^{N}}{\partial z_{2}}\end{pmatrix}\!\!\left(\left(\mathbf{Z}^{T},\mathcal{J}^{T}\right)+\eps^{N-1}\mathbf{L}^{\eps}\right)
         \\
         &~~~~+\eps^{N+1}\left[\mathtt{M}_{\eps}\!\begin{pmatrix}\frac{\partial\rest_{H}^{N}}{\partial z_{1}}\\
\frac{\partial\rest_{H}^{N}}{\partial z_{2}}\end{pmatrix}\!\!+\left(\Rest_{\mathcal{P}}^{N}(\eps,.)\right)^{\text{\tiny\! TL}}\!\!\begin{pmatrix}\frac{\partial\hat{H}_{\eps}}{\partial z_{1}}\\
\frac{\partial\hat{H}_{\eps}}{\partial z_{2}}\end{pmatrix}\!\!+\left(\Rest_{\mathcal{P}}^{N}(\eps,.)\right)^{\text{\tiny\! TR}}\!\!\begin{pmatrix}\frac{\partial\hat{H}_{\eps}}{\partial\gamma}\\
\frac{\partial\hat{H}_{\eps}}{\partial j}\end{pmatrix}\!\!\right]
          \\ 
         &~\hspace{3cm}
        \left(\mathbf{Z}_{1}^{T}+\eps^{N-1}\mathbf{L}_{1}^{\eps},\mathbf{Z}_{2}^{T}+\eps^{N-1}\mathbf{L}_{2}^{\eps},\widetilde{\Gamma}\!\left(\frac{t}{\eps}\right),\mathcal{J}+\eps^{N-1}\mathbf{L}_{3}^{\eps}\right),
         \\
         &\frac{\partial\mathcal{J}^{T}}{\partial t}+\eps^{N-1}\frac{\partial\mathbf{L}_{3}^{\eps}}{\partial t}
         \\
         &~~=-\eps^{N}\left[\eps\left(\Rest_{\mathcal{P}}^{N}(\eps,.)\right)^{\!1,4}\frac{\partial H_{\eps}}{\partial z_{1}}+\eps\left(\Rest_{\mathcal{P}}^{N}(\eps,.)\right)^{\!2,4}\frac{\partial H_{\eps}}{\partial z_{2}}+\eps\left(\Rest_{\mathcal{P}}^{N}(\eps,.)\right)^{\!3,4}\frac{\partial H_{\eps}}{\partial\gamma}+\frac{\partial\rest_{H}}{\partial\gamma}(\eps,.)\right]
         \\ 
         &~\hspace{3cm}
         \left(\mathbf{Z}_{1}^{T}+\eps^{N-1}\mathbf{L}_{1}^{\eps},\mathbf{Z}_{2}^{T}+\eps^{N-1}\mathbf{L}_{2}^{\eps},\widetilde{\Gamma}\!\left(\frac{t}{\eps}\right),\mathcal{J}+\eps^{N-1}\mathbf{L}_{3}^{\eps}\right).
\end{aligned}
\end{gather}
Making a Taylor expansion in
\begin{align*}
   &\mathtt{M}_{\eps}\!\left(\mathbf{Z}_{1}^{T}+\eps^{N-1}\mathbf{L}_{1}^{\eps},\mathbf{Z}_{2}^{T}+\eps^{N-1}\mathbf{L}_{2}^{\eps}\right)\begin{pmatrix}\frac{\partial\hat{H}_{\eps}^{N}}{\partial z_{1}}\\
\frac{\partial\hat{H}_{\eps}^{N}}{\partial z_{2}}\end{pmatrix}\!\!\left(\left(\mathbf{Z}^{T},\mathcal{J}^{T}\right)+\eps^{N-1}\mathbf{L}^{\eps}\right)
\end{align*}
we obtain
\begin{gather}
\label{BigFormula12Juin}    
\begin{aligned}
   &\mathtt{M}_{\eps}\!\left(\mathbf{Z}_{1}^{T}+\eps^{N-1}\mathbf{L}_{1}^{\eps},\mathbf{Z}_{2}^{T}+\eps^{N-1}\mathbf{L}_{2}^{\eps}\right)\begin{pmatrix}\frac{\partial\hat{H}_{\eps}^{N}}{\partial z_{1}}\\
\frac{\partial\hat{H}_{\eps}^{N}}{\partial z_{2}}\end{pmatrix}\!\!\left(\left(\mathbf{Z}^{T},\mathcal{J}^{T}\right)+\eps^{N-1}\mathbf{L}^{\eps}\right)
   \\
&~~~~~~=\mathtt{M}_{\eps}\!\left(\mathbf{Z}^{T}\right)\left(\begin{array}{c}
\ds\frac{\partial H_{\eps}^{N}}{\partial z_{1}}\\
\ds\frac{\partial H_{\eps}^{N}}{\partial z_{2}}\end{array}\right)\negmedspace\left(\mathbf{Z}^{T},\mathcal{J}^{T}\right)+\eps^{N-1}\beta_{1}\left(\eps,\mathbf{Z}^{T},\mathcal{J}^{T},\mathbf{L}^{\eps}\right),
\end{aligned}
\end{gather}
where $\beta_{1}$ is smooth and periodic of period 2$\pi$ with respect to $\gamma$.
Injecting \eqref{BigFormula12Juin} in \eqref{EqSatByThreeComp2} and using \eqref{FormForTruncDynSys12Juin} yields 
\begin{gather}
\label{DefCaractLSolDynSysV12}    
       \frac{\partial\left(\begin{array}{c}
\mathbf{L}_{1}^{\eps}\\
\mathbf{L}_{2}^{\eps}\end{array}\right)}{\partial t}=\beta_{1}\!\left(\eps,\mathbf{L}^{\eps},\mathbf{Z}^{T},\mathcal{J}^{T}\right)+\eps\beta_{2}\!\left(\eps,\mathbf{L}^{\eps},\mathbf{Z}^{T},\mathcal{J}^{T},\widetilde{\Gamma}\!\left(\frac{t}{\eps}\right)\right),
\end{gather}
and 
\begin{gather}
         \frac{\partial\mathbf{L}_{3}^{\eps}}{\partial t}=\eps\beta_{3}\!\left(\eps,\mathbf{L}^{\eps},\mathbf{Z}^{T},\mathcal{J}^{T},\widetilde{\Gamma}\!\left(\frac{t}{\eps}\right)\right),
\end{gather}
where $\beta_2$ and $\beta_3$ are smooth and $2\pi$-periodic with respect to $\gamma.$ 
Besides, the solutions of this dynamical system are continuous with respect to $\eps.$ 
Clearly the initial data for $\mathbf{L}^{\eps}$ is $\mathbf{L}^{\eps}(0)= 0.$
Hence, $\mathbf{L}^{\eps}$ is continuous with respect to $\eps.$
Since $\left(\mathbf{Z},\mathcal{J}\right)-\left(\mathbf{Z}^{T},\mathcal{J}^{T}\right)=\eps^{N-1}\mathbf{L}^{\eps}$, estimate \eqref{MainEstimationOfOurPaper} follows.
This ends the proof of Theorem \ref{MainThm3}.
{~\hfill $\square$}

%
\begin{appendix}
\section{Appendix : Change of coordinates rules for the Poisson Matrix and the Hamiltonian Function}
 \label{AppendixChangeOfCoordRulesForHamAndPM} 
%

$\underline{A}$ Poisson Matrix $\mathcal{P}$ on an open subset of $\mathbb{R}^4$ is a skew-symmetric matrix satisfying:
\begin{align}
       &\forall i,j,k\in\left\{ 1,\ldots,4\right\} ,\ \left\{ \left\{ \mathbf{r}_{i},\mathbf{r}_{j}\right\} ,\mathbf{r}_{k}\right\} +\left\{ \left\{ \mathbf{r}_{k},\mathbf{r}_{i}\right\} ,\mathbf{r}_{j}\right\} +\left\{ \left\{ \mathbf{r}_{j},\mathbf{r}_{k}\right\} ,\mathbf{r}_{i}\right\} =0,
       \label{JacobiJacoby}     
\end{align}  
where $\mathbf{r}_i$ is the i-th coordinate function $\mathbf{r}\mapsto r_i$
and the Poisson Bracket $\left\{ f,g\right\} $ between smooth functions $f$ and $g$ is defined by \eqref{DsDim111011}.
\\
In the case of
a symplectic manifold, $\underline{the}$ Poisson Matrix in a given coordinate system is defined as follow:
it is the inverse of the transpose of the matrix of the expression of the Symplectic Two-Form in this coordinate system.
Notice that the Jacoby identities \eqref{JacobiJacoby} are direct consequences of the closure of the Symplectic Two-Form.
\\

We now turn to the change-of-coordinates rule for the Poisson Matrix.
Firstly, if in a given coordinate chart $\boldsymbol{m}$, the matrix associated with the Symplectic Two-Form reads $\mathcal{K}$, then, according to the previous definition, the Poisson Matrix is given by 
\begin{align}
\mathcal{P}\left(\boldsymbol{m}\right)=\left(\mathcal{K}\left(\boldsymbol{m}\right)\right)^{\!-T}.
\end{align}
If we make the change of coordinates $\sigma:\ \boldsymbol{m}\mapsto\boldsymbol{r}$, then the usual change-of-coordinates rule for the expression of the Symplectic Two-Form leads to the following change of coordinates rule for the Poisson Matrix
\begin{align}
\mathcal{P}'\left(\boldsymbol{r}\right)=\nabla_{\!\boldsymbol{m}}\sigma\left(\sigma^{-1}\left(\boldsymbol{r}\right)\right)\mathcal{P}\left(\sigma^{-1}\left(\boldsymbol{r}\right)\right)\left[\nabla_{\!\boldsymbol{m}}\sigma\left(\sigma^{-1}\left(\boldsymbol{r}\right)\right)\right]^{T}.
\end{align}
Using the Poisson Bracket defined in formula \eqref{DsDim111011},
the change-of-coordinates rule for the Poisson Matrix reads
\begin{align}
       &\mathcal{P}'{}_{\! i,j}\!\left(\boldsymbol{r}\right)=\left\{ \sigma_{i},\sigma_{j}\right\} _{\!\boldsymbol{m}}\left(\sigma^{-1}\left(\boldsymbol{r}\right)\right).
\end{align}

A Hamiltonian function on  a symplectic manifold $(\mathcal{M},\boldsymbol{\Omega})$ is a smooth function on $\mathcal{M}$ and
the Hamiltonian vector field associated with Hamiltonian function $\mathcal{G}$ is the unique vector field $\boldsymbol{\mathcal{X}}_{\mathcal{G}}$ satisfying
\begin{align}
      &i_{\boldsymbol{\mathcal{X}}_{\mathcal{G}}}d\boldsymbol{\Omega}=d\mathcal{G},
\end{align}
where $i_{\boldsymbol{\mathcal{X}}_{\mathcal{G}}}d\boldsymbol{\Omega}$ is the interior product of 
differential two-form $d\boldsymbol{\Omega}$ by vector field ${\boldsymbol{\mathcal{X}}_{\mathcal{G}}}$.
\\
The expression of the Hamiltonian vector field associated with the Hamiltonian function $\mathcal{G},$ 
in the coordinate system $\boldsymbol{m},$ is the vector field which reads:
\begin{align}
\mathbf{X}_{G}\left(\boldsymbol{m}\right)=\mathcal{P}\left(\boldsymbol{m}\right)\nabla_{\!\boldsymbol{m}}G\left(\boldsymbol{m}\right),
\end{align}
where $G$ is the representative of $\mathcal{G}$ in this coordinate system.
In fact, we can consider Hamiltonian vector fields on $\mathcal{M}$, which requires that the Hamiltonian functions are smooth functions on $\mathcal{M}$, or just Hamiltonian vector fields on an open subset of $\mathcal{M}$, which requires that the Hamiltonian functions are defined on this open subset. 

The Hamiltonian dynamical system  associated with Hamiltonian function $\mathcal{G}$ on $\mathcal{M}$ is the dynamical system which reads 
\begin{align}
          &\frac{\partial\boldsymbol{\mathcal{R}}}{\partial t}\left(t\right)=\boldsymbol{\mathcal{X}}_{\mathcal{G}}\left(\boldsymbol{\mathcal{R}}\left(t\right)\right),
\end{align}
or equivalently as said in the introduction, the dynamical system whose expression in every coordinate system $\boldsymbol{r}$ is given by 
\begin{align}
\frac{\partial\mathbf{R}}{\partial t}=\mathcal{P}'\!\left(\mathbf{R}\right)\nabla_{\!\boldsymbol{r}}G'\!\left(\mathbf{R}\right).
\label{DsDim111015} 
\end{align}
where $G'$ is the representative of $\mathcal{G}$ in this coordinate system, and $\mathcal{P}'$ the expression of the Poisson Matrix in this coordinate system.
In particular, if we check that on a global coordinate chart, a dynamical system is Hamiltonian, then the dynamical system is Hamiltonian on $\mathcal{M}$ and 
its expression in every coordinate chart $\boldsymbol{r}$ is given by \eqref{DsDim111015}.

\end{appendix}

\bibliographystyle{plain}
\bibliography{biblio}

\end{document}